\documentclass[12pt]{article}
\usepackage{amsmath}
\usepackage{amsthm}
\usepackage{amssymb}
\usepackage{amsfonts}
\usepackage{epic}
\usepackage{eepic}
 \usepackage{times}
\usepackage[matrix,arrow]{xy}
\usepackage{macros}
\usepackage{epsfig}
\usepackage{bm}

\theoremstyle{plain}

\newtheorem{thm}{Theorem}
\newtheorem{propn}{Proposition}
\newtheorem{lem}{Lemma}

\theoremstyle{remark}

\newcommand{\sst}{\scriptscriptstyle}

\renewcommand{\1}{\one}
\renewcommand{\2}{\two}
\newcommand{\3}{\three}
\newcommand{\4}{{\mathfrak 4}}

\newcommand{\cone}{\check{\1}}
\newcommand{\ctwo}{\check{\2}}

\newcommand{\id}{{\rm id}}

\newcommand{\pa}{\partial}
\newcommand{\ot}{\otimes}
\newcommand{\ra}{\to}

\newcommand{\fr}[2]{{\textstyle \frac{#1}{#2} }}

\newcommand{\fsl}{{\mathfrak s}{\mathfrak l}}

\renewcommand{\=}[1]{\stackrel{(\ref{#1})}{=}}
\newcommand{\df}{\equiv}

\newcommand{\al}{\alpha}

\newcommand{\be}{\beta}
\newcommand{\ga}{\gamma}
\newcommand{\Ga}{\Gamma}
\newcommand{\de}{\delta}
\newcommand{\De}{\Delta}
\newcommand{\ep}{\epsilon}

\newcommand{\om}{\omega}

\newcommand{\si}{\sigma}

\newcommand{\vf}{\varphi}

\newcommand{\CA}{{\mathcal A}}

\newcommand{\CC}{{\mathcal C}}
\newcommand{\CD}{{\mathcal D}}

\newcommand{\CH}{{\mathcal H}}

\newcommand{\CI}{{\mathcal I}}

\newcommand{\CK}{{\mathcal K}}

\newcommand{\CM}{{\mathcal M}}

\newcommand{\CP}{{\mathcal P}}

\newcommand{\CS}{{\mathcal S}}
\newcommand{\CT}{{\mathcal T}}

\newcommand{\CU}{{\mathcal U}}

\newcommand{\SA}{{\mathsf A}}
\newcommand{\SB}{{\mathsf B}}
\newcommand{\SC}{{\mathsf C}}

\newcommand{\SE}{{\mathsf E}}
\newcommand{\SF}{{\mathsf F}}

\newcommand{\SK}{{\mathsf K}}

\newcommand{\SL}{{\mathsf L}}

\newcommand{\SM}{{\mathsf M}}

\newcommand{\SO}{{\mathsf O}}
\newcommand{\SP}{{\mathsf P}}
\newcommand{\SQ}{{\mathsf Q}}
\newcommand{\SR}{{\mathsf R}}
\renewcommand{\SS}{{\mathsf S}}
\newcommand{\ST}{{\mathsf T}}
\newcommand{\SU}{{\mathsf U}}
\newcommand{\SV}{{\mathsf V}}
\newcommand{\SW}{{\mathsf W}}
\newcommand{\SX}{{\mathsf X}}

\newcommand{\SZ}{{\mathsf Z}}

\newcommand{\fv}{{\mathfrak v}}

\newcommand{\sa}{{\mathsf a}}

\newcommand{\sd}{{\mathsf d}}
\newcommand{\se}{{\mathsf e}}

\renewcommand{\sf}{{\mathsf f}}
\newcommand{\sh}{{\mathsf h}}
\newcommand{\sll}{{\mathsf l}}

\newcommand{\sq}{{\mathsf q}}
\newcommand{\spp}{{\mathsf p}}
\newcommand{\sr}{{\mathsf r}}
\newcommand{\mss}{{\mathsf s}}
\newcommand{\mst}{{\mathsf t}}
\newcommand{\su}{{\mathsf u}}

\newcommand{\sw}{{\mathsf w}}

\newcommand{\sy}{{\mathsf y}}
\newcommand{\sz}{{\mathsf z}}

\newcommand{\zero}{{\mathfrak 0}}
\newcommand{\0}{{\mathfrak 0}}
\newcommand{\one}{{\mathfrak 1}}
\newcommand{\two}{{\mathfrak 2}}
\newcommand{\three}{{\mathfrak 3}}

\newcommand{\BR}{{\mathbb R}}

\newcommand{\BC}{{\mathbb C}}

\newcommand{\BS}{{\mathbb S}}

\newcommand{\rf}[1]{(\ref{#1})}
\newcommand{\fus}[6]{\big\{\,{}^{#1}_{#3}\;{}^{#2}_{#4}\;{}^{#5}_{#6}\,
\big\}_b} 
\newcommand{\Fus}[6]{F_{#5#6}^{}\big[\,{}^{#3}_{#4}\;{}^{#2}_{#1}\,\big]}

\newcommand{\nc}{\newcommand}

\nc{\rnc}{\renewcommand} \nc{\beq}{\begin{equation}}
\nc{\eeq}{\end{equation}} \nc{\beqa}{\begin{eqnarray}}
\nc{\eeqa}{\end{eqnarray}}

\begin{document}
\title{On  
the relation between the modular double of $\CU_{q}(\fsl(2,\BR))$\\ 
and the quantum Teichm\"uller theory}
\author{Iurii Nidaiev\footnote{Current address:
Department of Physics and Astronomy,
Rutgers University, Piscataway, NJ 08855, USA,}, J\"org Teschner}
\address{
DESY Theory, Notkestr. 85, 22603 Hamburg, Germany\\
{\tt teschner@mail.desy.de}}
\maketitle

\begin{quote}
\centerline{\bf Abstract}
{\small We exhibit direct relations 
between the modular double
of $\CU_{q}(\fsl(2,\BR))$ and the
quantum Teichm\"uller theory. 
Explicit representations for the fusion- and 
braiding operations of the quantum Teichm\"uller theory are 
immediate consequences.
Our results include a simplified derivation of
the Clebsch-Gordan decomposition for 
the principal series of representation
of the modular double
of $\CU_{q}(\fsl(2,\BR))$. 
}\end{quote}


\tableofcontents

\section{Introduction}
\setcounter{equation}{0}
The modular double 
of $\CU_{q}(\fsl(2,\BR))$ is a non-compact
quantum group closely related to the quantum deformation 
of the universal enveloping algebra of $\fsl(2,\BR)$. It has
some interesting features that are responsible for its relevance
to conformal field theory \cite{PT1,Teschner:2001rv}, 
integrable models \cite{BT}, and quantum Teichm\"uller theory.
The generators of the modular double are represented
by {\it positive} self-adjoint operators, which was shown 
in \cite{BT1} to be responsible for the remarkable
self-duality of the modular double: It is simultaneously
the modular double of $\CU_{\tilde{q}}(\fsl(2,\BR))$, with 
deformation parameter $\tilde{q}$ given as 
$\tilde{q}=e^{\pi i b^{-2}}$ if $e^{\pi i b^2}$. This self-duality 
was pointed out independently in \cite{FMD} and in \cite{PT1},
and it has profound consequences in the applications
of this mathematical structure. One may, for example, use
it to explain the quantum-field theoretical self-dualities 
of the Liouville theory \cite{Teschner:2001rv} 
and of the Sinh-Gordon model \cite{BT}.

There are various hints that there must be close
connections between the quantization of the Teichm\"uller spaces
constructed in \cite{F97,CF,Ka1} on the one hand,
and the modular double of $\CU_{q}(\fsl(2,\BR))$ on the other hand. 
First hints came from the observations
made in \cite{T03} that the fusion move in the quantum 
Teichm\"uller theory gets represented in terms of the 
6j-symbols of the modular double \cite{PT1,PT2}. 
One may also observe \cite{FK} that the quantum Teichm\"uller theory 
is essentially build from the basic data of the 
modular double of the quantum $(ax+b)$-group, the so-called multiplicative
unitary.
As the $(ax+b)$-group is nothing but the Borel half of $SL(2,\BR)$,
one may expect relations between the 
quantum Teichm\"uller theory and the modular double 
of  $\CU_{\tilde{q}}(\fsl(2,\BR))$ to follow by combining the
quantum double construction of 
${\rm Fun}_q(SL(2,\BR))$ from the quantum $(ax+b)$-group \cite{Ip} 
with the duality between the modular doubles of 
${\rm Fun}_q(SL(2,\BR))$ and $\CU_{q}(\fsl(2,\BR))$ 
described in \cite{PT1}, and proven in \cite{Ip}.

However, all these hints are somewhat indirect. We'll here exhibit
a direct link by establishing a relation between
the Casimir operator of $\CU_{q}(\fsl(2,\BR))$ and the 
geodesic length operators of the quantum Teichm\"uller theory.
The key observation is that the co-product of the modular
double of  $\CU_{\tilde{q}}(\fsl(2,\BR))$ gets represented by 
an operator in the quantum Teichm\"uller theory that
has a simple geometrical interpretation in terms of 
changes of triangulation of the underlying Riemann surfaces.
The combinatorial structure of the quantum
Teichm\"uller theory can be used to find an explicit 
expression for the Clebsch-Gordan operator that
describes the decomposition of the tensor product 
of two irreducible representations of the modular 
double into irreducible representations.

An immediate consequence is the direct relation between the 
kernel representing the
fusion operation from quantum Teichm\"uller theory and the 
b-6j symbols of the modular double \cite{PT1,PT2}. 
We also find, not surprisingly, that the R-operator of the modular double 
\cite{FMD,BT1} is directly related to the braiding operation
of the  quantum Teichm\"uller theory. 

The Clebsch-Gordan maps of the modular double have previously been constructed in 
\cite{PT2} as an integral operator with an explicit kernel.
However, especially the proof of the completeness 
for the Clebsch-Gordan decomposition given in \cite{PT2}
was quite complicated. The construction of the Clebsch-Gordan operator
given in this paper will allow us to re-derive the main
results of \cite{PT2} on 
the Clebsch-Gordan decomposition 
in a simpler, and hopefully more transparent
way. The explicit construction of the Clebsch-Gordan operator
presented below reduces the proof of completeness to the
results of \cite{Ka3,Ka4} on the spectral decomposition of the 
geodesic length operators in Teichm\"uller theory. The proof
of this result given in \cite{Ka4} is much simpler than
the proof of the corresponding result on the Casimir
operators of $\CU_{q}(\fsl(2,\BR))$ given in \cite{PT2}.

\noindent
{\bf Acknowledgements:}
J.T. would like to thank L. Faddeev for interesting discussions 
and correspondence. I.N. gratefully acknowledges 
support by a EPFL Excellence Fellowship
of the \'Ecole Polytechnique F\'ed\'erale de Lausanne.

\section{Some notations and conventions}
\setcounter{equation}{0}

The special function $e_b(U)$ can be defined in the strip 
$|\Im z|<|\Im c_b|$, $c_b\equiv i(b+b^{-1})/2$ by means of the 
integral representation
\begin{equation}
\log e_b(z)\;\equiv\;\frac{1}{4}
\int\limits_{i0-\infty}^{i0+\infty}\frac{dw}{w}
\frac{e^{-2{\mathsf i}zw}}{\sinh(bw)
\sinh(b^{-1}w)}.
\end{equation}
Closely related is the function $w_b(x)$ defined via
\begin{equation}\label{eb-wb}
e_b(x)=e^{-\frac{\pi i}{\1\2}(1+2c_b^2)}
{e^{\frac{\pi i}{2}x^2}}{(w_b(x))^{-1}}\,.
\end{equation}
Another useful combination is the function $D_\al(x)$, defined
as
\begin{equation}
D_\al(x):=\frac{w_b(x+\al)}{w_b(x-\al)}\,. 
\end{equation}

For tensor products we will be using the following 
leg-numbering notation. Let us first define, as usual,
\begin{equation}
\SX_r:=1\ot\dots\ot \underset{{\rm r-th}}{\SX}\ot\dots \ot 1\,.
\end{equation}
We are using the slightly unusual convention to label tensor
factors from the right to the left, as, for example, in 
$\CH_2\ot\CH_1$.



\section{Modular double}

\setcounter{equation}{0}
\subsection{Prinicpal series representations of $\CU_q(\fsl(2,\BR))$}

We will be considering the Hopf-algebra $\CU_q(\fsl(2,\BR))$
which has generators $E$, $F$ and $K$ subject to the  
relations,
\begin{equation}
\begin{aligned}
&KE\,=\,q\,EK\,,\\
&KE\,=\,q\,EK\,,
\end{aligned}
\qquad [\,E\,,\,F\,]\,=\,-\frac{K^2-K^{-2}}{q-q^{-1}}\,.
\end{equation}
The algebra $\CU_q(\fsl(2,\BR))$ has the central element
\begin{equation}
Q=(q-q^{-1})^2\,FE-qK^2-qK^{-2}\,.
\end{equation}
The co-product is given as
\begin{align}
 \De(K)\,=\,&K\ot K\,,\qquad 
\begin{aligned} & \De(E)=E\ot K+K^{-1}\ot E\,,\\
 & \De(F)=F\ot K+K^{-1}\ot F\,.
\end{aligned}
\label{De}
\end{align}
This implies
\begin{align}
 \De(Q)\,=\,&K^{-1}F\ot EK
+K^{-1}E\ot F K \notag \\
& +Q\ot K^2+K^{-2}\ot Q+(q+q^{-1})K^{-2}\ot K^2\,.
\end{align}

This algebra has a one-parameter
family of representations $\CP_{\al}'$
\begin{equation}\label{Paldef}
\begin{aligned}
 \SE_s'
\equiv\pi_{s}(E):=e^{+\pi b\sq}\frac{\cosh\pi b (\spp-s)}{\sin\pi b^2}
e^{+\pi b\sq}\,,\\
\SF_s'\equiv\pi_{s}(F):=e^{-\pi b\sq}\frac{\cosh\pi b (\spp+s)}{\sin\pi b^2}
e^{-\pi b\sq}\,,
\end{aligned}
\qquad \SK_s'\equiv\pi_s(K):=e^{-\pi b\spp}\,,
\end{equation}
where $\spp$ and $\sq$ are operators acting on functions $f(q)$ as
$\spp f(q)=(2\pi i)^{-1}\frac{\pa}{\pa q}f(q)$ and $\sq f(q)=qf(q)$,
respectively.
In the definitions \rf{Paldef} we are parameterizing $q$ as
$q=e^{\pi \textup{i} b^2}$. There is a maximal dense 
subspace $\CP_{s}'\subset L^2(\BR)$ on which all polynomials formed out of 
$\SE_s'$, $\SF_s'$ and $\SK_s'$ are well-defined 
\cite[Appendix B]{BT}.

\subsection{Modular duality}

These representations are distinguished by
a remarkable self-duality property: It is automatically a
representation of the quantum group 
$\CU_{\tilde{q}}(\fsl(2,\BR))$, where 
$\tilde{q}=e^{\pi \textup{i}/b^2}$ if $q=e^{\pi \textup{i} b^2}$. 
This representation is generated from operators 
$\tilde{\SE}_\al$, $\tilde{\SF}_\al$ and  $\tilde{\SK}_\al$ 
which are defined by formulae obtained from 
those in \rf{Paldef} by replacing $b\ra b^{-1}$.
The subspace $\CP_\al$ is simultaneously 
a maximal domain for the polynomial functions of 
$\tilde{\SE}_\al$, $\tilde{\SF}_\al$ and  $\tilde{\SK}_\al$ 
\cite[Appendix B]{BT}.

This phenomenon was observed independently in \cite{PT1} and in 
\cite{FMD}. 
It is closely related to the fact that  $\SE_\al$, $\SF_\al$ and $\SK_\al$
are  {\it positive} 
self-adjoint generators which allows one to construct
$\tilde{\SE}_\al$, $\tilde{\SF}_\al$ and  $\tilde{\SK}_\al$
via \cite{BT1}
\begin{equation}
\tilde{\se}\,=\,\se^{1/b^2}\,,\qquad\tilde{\sf}\,=\,\sf^{1/b^2}\,,\qquad
\tilde{\SK}\,=\,\SK^{1/b^2}
\end{equation}
using the notations
\begin{equation}
\begin{aligned}
&\se:=2\sin(\pi b^2)\,\SE\,,\\
&\tilde{\se}:=2\sin(\pi b^{-2})\,\tilde{\SE}\,,
\end{aligned}
\qquad
\begin{aligned}
&\sf:=2\sin(\pi b^2)\,\SF\,,\\
&\tilde{\sf}:=2\sin(\pi b^{-2})\,\tilde{\SF}\,.
\end{aligned}
\end{equation}

It was proposed in \cite{PT1,BT1} to construct a
noncompact quantum group which has as {\it complete} set 
of tempered representations the self-dual representations
$\CP_{\al}$. 
It's gradually becoming clear how to realize this suggestion
precisely. 
Relevant steps in this direction were taken in \cite{BT1} by 
defining co-product, R-operator and Haar-measure of such a 
quantum group. Further important progress in this direction was
recently made in \cite{Ip}. 
Following \cite{FMD}, we will in the following call 
this noncompact quantum 
group the modular double of $\CU_q(\fsl(2,\BR))$.

\subsection{The Whittaker model for $\CD\CU_q(\fsl_2)$}

A unitarily equivalent family of representations of the modular double
is
\begin{subequations} \label{Whitt}
\begin{align}
&2\sin\pi b^2\,\SE_s\,=\,e^{\pi b(2\sq-\spp)}\,,\qquad
\SK_s\,=\,e^{-\pi b \spp}\,,\\
&2\sin\pi b^2\,\SF_s\,=\,
e^{\pi b(\sq-\spp/2)}(2\cosh(2\pi b s)+2\cosh(2\pi b\spp))
e^{\pi b(\sq-\spp/2)}\,,
\end{align} 
\end{subequations}
A joint domain of definition is the space $\CP$ of entire functions
which decay faster than any polynomial when going to infinity along
the real axis.
It is easy to see that this representation is unitarily equivalent to 
the one defined in \rf{Paldef},
$ 
\SX_s'=\SU_s^{}\cdot \SX\cdot\SU_s^{-1},
$ 
with
\begin{equation}
\SU_s:=e^{-\frac{\pi i}{2}\spp^2}w_b(\spp-s)\,.
\end{equation}

In any representation in which $\SE_r$ are invertible
we may represent the action of the Casimir 
on the tensor product 
of two representations, defined as, 
\begin{equation}
\SQ_{\2\1}\equiv(\pi_{s_\2}\ot\pi_{s_\1})(Q)\,,
\end{equation} 
by the formula
\begin{align}\label{Q21}
 \SQ_{\2\1}\,=\,&\SK^{-1}_\2\,
(q\SK^2_\2+q^{-1}\SK^{-2}_\2+\SQ_\2^{})\,\SE^{-1}_\2 \SE_\1^{}\SK_\1^{}
+\SK^{-1}_\2\SE_\2^{}\, 
(q\SK^2_\1+q^{-1}\SK^{-2}_\1+\SQ_\1^{})\,\SE^{-1}_\1 \SK_\1^{}\notag \\
& +\SQ_\2 \SK^2_\1+\SK^{-2}_\2 \SQ_\1+(q+q^{-1})\SK^{-2}_\2 \SK^2_\1\,.
\end{align}
Our main task is to diagonalize this operator.

\subsection{The model space}

It will be useful for us to introduce a space which contains 
all the irreducible representations of the modular double
with multiplicity one.

Let us consider the space $\CM:=\CP\ot L^2(\BR_+,d\mu)$. 
We'll choose the measure $d\mu\equiv d\mu(s)$ as
\begin{equation}
d\mu(s):=\,ds\;4\sinh(2\pi bs) \sinh(2\pi b^{-1}s) \,.
\end{equation}
This space may be
identified with the space of functions of two variables, taken
to be functions of $f(p_\2,s_\1)$. We will consider an operator
$\SQ\equiv\pi_{\CM}(Q)$ 
which will represent the action of the Casimir $Q$ on
$\CM$. Its action is 
\begin{equation}
\SQ\cdot f(p_\2,s_\1):=\,
2\cosh(2\pi b s_\1)\, f(p_\2,s_\1)\,.
\end{equation}
The space $\CM$ becomes a representation of the modular double
generated by the operators
$\SE\equiv\pi_{\CM}^{}(E)$,  $\SF\equiv\pi_{\CM}^{}(F)$
and $\SK\equiv\pi_{\CM}^{}(K)$ which are defined as
\begin{subequations} 
\begin{align}
&2\sin\pi b^2\,\SE\,=\,e^{\pi b(2\sq-\spp)}\,,\qquad
\SK\,=\,e^{-\pi b \spp}\,,\\
&2\sin\pi b^2\,\SF\,=\,
e^{\pi b(\sq-\spp/2)}(\SQ+2\cosh(2\pi b\spp))
e^{\pi b(\sq-\spp/2)}\,,
\end{align} 
\end{subequations}
It is clear by definition that $\CM$ decomposes into irreducible
representations of the modular double as
\begin{equation}
\CM \,\simeq \,
\int^{\oplus}d\mu(s)\;\CP_{s}\,\,,
\end{equation}
with action of the generators defined above. 

\subsection{The R-operator}

Let us introduce the rescaled generators $\se$ and 
$\sf$ via
\begin{equation}
\se:=2\sin\pi b^2\SE\,,\qquad
\sf:=2\sin\pi b^2\SF\,.
\end{equation}
Let us furthermore introduce an anti-self-adjoint element $\sh$ such
that $\SK=q^\sh$. We will then define the following operator on $\CM\ot\CM$:
\begin{equation}\label{Rdef}
 \SR\,=\, q^{\sh\ot\sh} \,
 E_b(\se\ot
 \sf)\, q^{\sh\ot\sh} \,.
\end{equation}
$\SR$ coincides with the R-operator proposed 
by L.~Faddeev in~\cite{FMD}. 
Notice that $|g_b(x)|=1$ for $x\in{\BR}^+$. This implies that
$\SR$ is manifestly unitary.

\begin{thm}\label{T1}
The operator $\SR$ has the following properties:
\begin{eqnarray}
 & {\rm (i)} &  \SR\,\De(\SX) \,=\, 
     \De'(\SX)\,\SR \,,  \label{i} \\
 & {\rm (ii)} & (\id\ot\De)\SR\,=\, \SR_{13}\SR_{12} ,\qquad
 (\De\ot\id)\SR\,=\, \SR_{13}\SR_{23} \,, \label{ii} \\
 & {\rm (iii)} & (\sigma \ot \id)\SR\,=\, \SR^{-1} , \qquad 
 (\id \ot \sigma)\SR\,=\, \SR^{-1} \,.  \label{iii}
\end{eqnarray}
\end{thm}

\begin{lem}\label{Qexp}
Let $\SU$ and $\SV$ be positive self-adjoint operators such that
$\SU\SV = q^2 \SV\SU$ where $q=e^{i \pi b^2}$. The function $E_b(x)$
satisfies the identities
\begin{eqnarray}
 \label{qexp}
 E_b(\SU) \ E_b(\SV) &=& E_b(\SU+\SV) \,, \\
 \label{pent}
 E_b(\SV) \ E_b(\SU) &=& E_b(\SU) \ E_b(q^{-1} \SU\SV) \ E_b(\SV) \,.
\end{eqnarray}
Furthermore, (\ref{qexp}) $\Leftrightarrow$ (\ref{pent}).
\end{lem}
In the literature, eqs.~(\ref{qexp}) and (\ref{pent}) 
are often referred to 
as the quantum exponential and the quantum pentagon
relations.

To prove the first formula in (\ref{ii}), we use 
the quantum exponential relation (\ref{qexp}) {}from 
Lemma~\ref{Qexp} with identification
\hbox{$\SU= \se_\1 \SK^{-1}_\2 \sf_\3$} 
and \hbox{$\SV= \se_\1 \sf_\2 \SK_\3$},
\begin{eqnarray*} 
 (\id\ot\De)\SR & \={Rdef} &  
 (\id\ot\De) \bigl( q^{ \sh_\1 \sh_\2} \,
 g_b ( \se_\1 \sf_\2) \, q^{ \sh_\1 \sh_\2} \bigr) \\
 & \={De} & q^{ \sh_\1 \sh_\2 + \sh_\1 \sh_\3} \,
 E_b \bigl( \se_\1 \sf_\2 \SK_\3 + 
  \se_\1 \SK^{-1}_\2 \sf_\3  \bigr)
 \, q^{ \sh_\1 \sh_\2 + \sh_\1 \sh_\3 }   \\ 
 & \={qexp} & q^{ \sh_\1 \sh_\2 + \sh_\1 \sh_\3 } \, 
 E_b( \se_\1^{} \SK^{-1}_\2 \sf_\3 ) \,
 E_b( \se_\1^{} \sf_\2^{} \SK_\3^{} ) \,
 q^{ \sh_\1 \sh_\2 + \sh_\1 \sh_\3 } \\
 & = & q^{ \sh_\1 \sh_\3} \, E_b( \se_\1 \sf_\3 ) \,
 q^{ \sh_\1 \sh_\3} \cdot q^{ \sh_\1 \sh_\2} \, E_b( \se_\1 \sf_\2 ) 
 \, q^{ \sh_\1 \sh_\2} =  \SR_{\1\3}\SR_{\1\2} \,. 
\end{eqnarray*}
The second formula in (\ref{ii}) is proved in the same way. 

The R-operator allows us to introduce the braiding of tensor 
products of the representations $\CP_{s}$. Specifically, let the 
operator $\SB:\CP_{s_{\two}}\ot \CP_{s_\one}\ra \CP_{s_{\one}}\ot 
\CP_{s_\two}$
be defined by $\SB_{s_{\two},s_\one}
\equiv \SP\SR_{s_{\two},s_\one}$, where $\SP$ is the 
operator that permutes the two tensor factors. Property 
(i) {}from Theorem~\ref{T1} implies as usual that 
$\SB_{s_{\two},s_\one}\circ \De(\SX)=
 \De(\SX)\circ\SB_{s_{\two},s_\one}$.



\section{The Clebsch-Gordan maps of the modular double}

\setcounter{equation}{0}

In this section we are going to re-derive the main results of \cite{PT2} 
on the Clebsch-Gordan decomposition of tensor products of representations
of  $\CD\CU_q(\fsl(2,\BR))$ in a completely new way. 
The most difficult part in \cite{PT2}
was to prove the completeness of the eigenfunctions of the 
Casimir operator $\SQ_{\2\1}$ acting
on the tensor product of two representations. This result will now
be obtained by first constructing an explicit unitary operator which maps
$\SQ_{\2\1}$ to a simple standard form $\SQ_\1''$, 
and then applying the 
result of Kashaev \cite{Ka4} on the completeness of the eigenfunctions
of $\SQ_\1''$. The resulting proof is much shorter than the 
one given in \cite{PT2}.


\subsection{Definition of Clebsch-Gordan maps}

The goal is to construct the Clebsch-Gordon projection maps, 
\begin{equation}
\SC^{s_\3}_{s_\2s_\1}:\CP_{s_\2}\ot\CP_{s_\1}\,\ra\, \CP_{s_\3}\,,
\end{equation}
that satisfy
\begin{equation}
\SC^{s_\3}_{s_\2s_\1}\cdot
(\pi_{s_\2}\ot\pi_{s_\1})(X)\,=\,\pi_{s_\3}(X)\cdot\SC^{s_\3}_{s_\2s_\1} \,.
\end{equation}

It will be convenient to consider the 
unitary operators
\begin{equation}
\SC_{s_\2s_\1}:\CP_{s_\2}\ot\CP_{s_\1}\,\ra\, \CM\,\simeq\, 
\int^{\oplus}d\mu(s_\3)\;\CP_{s_\3}\,,
\end{equation}
related to $\SC^{s_\3}_{s_\2s_\1}$ as
\begin{equation}
\SC_{s_\2s_\1}\,=\,\int^{\oplus}d\mu(s_\3)\;\SC^{s_\3}_{s_\2s_\1}\,.
\end{equation}

We note that $\SC_{s_\2s_\1}$ is characterized by the properties
\begin{subequations}\label{intertw}
\begin{align}
&(\SC_{s_\2s_\1})^{-1}\cdot \SE\cdot \SC_{s_\2s_\1}\,=\,
\SE_\2\SK_\1+\SK^{-1}_\2\SE_\1\,,\\
&(\SC_{s_\2s_\1})^{-1}\cdot \SK\cdot \SC_{s_\2s_\1}\,=\,\SK_\2\SK_\1\,,\\
&(\SC_{s_\2s_\1})^{-1}\cdot \SQ\cdot \SC_{s_\2s_\1}\,=\,\SQ_{\2\1}\,.
\end{align}
\end{subequations}
The ``missing'' property
\[ 
(\SC_{s_\2s_\1})^{-1}\cdot \SF\cdot \SC_{s_\2s_\1}\,=\,
\SF_\2\SK_\1+\SK^{-1}_\2\SF_\1\,,
\]
is an easy consequence of \rf{intertw}, since invertibility of $\SE$ allows
us to express
$\SF$ in terms of $\SE$, $\SK$ and $\SQ$ in our representations.

\subsection{Factorization of Clebsch-Gordan maps}

We will construct the Clebsch-Gordan maps in the following factorized 
form: 
\begin{equation}\label{CGMconstr}
\SC_{s_\2s_\1}:=\,\nu^{\mss_{\2\1}}_{s_\2s_\1}\cdot\SS_\1\cdot\SC_\1\cdot (\ST_{\1\2})^{-1}\,,
\end{equation}
where
\begin{itemize}
\item The operator $\ST_{\1\2}$ satisfies
\begin{subequations}\label{ST12}
\begin{align}
\label{ST12a}
&\ST_{\1\2}\cdot \SE_\2\cdot (\ST_{\1\2})^{-1}\,=\,\SE_\2\SK_\1+
\SK^{-1}_\2\SE_\1\,,\\
\label{ST12b}
&\ST_{\1\2}\cdot \SK_\2\cdot (\ST_{\1\2})^{-1}\,=\,\SK_\2\SK_\1\,,\\
\label{ST12c}
&\ST_{\1\2}\cdot {\SQ}_{\2\1}'\cdot (\ST_{\1\2})^{-1}\,=\,\SQ_{\2\1}\,,
\end{align}
\end{subequations}
where
\begin{equation}\label{SQ1def}
{\SQ}_{\1}':=\,2\cosh 2\pi b(\sq_\1-\spp_\1)+e^{-2\pi b\sq_\1}Q_\1
+e^{-2\pi b\spp_\1}Q_\2+e^{-2\pi b(\spp_\1+\sq_\1)}\,.
\end{equation}
This means that $\ST_{\1\2}$ generates the representation of the
co-product in the 
representation of the Borel-subalgebra generated by $E$ and $K$ on 
$\CP_{s_\2}\ot\CP_{s_\1}$, and it simplifies $\SQ_{\2\1}$ to an operator
that acts nontrivially only on one tensor factor.
\item The operator $\SC_\1$ maps $L^2(\BR^2)$ to itself, 
commutes with $\SK_\2$ and $\SE_\2$ 
and maps $\SQ_\1'$ to a simple form,
\begin{equation}\label{SC1}
(\SC_\1)^{-1}\cdot\SQ_\1''\cdot\SC_\1\,=\,{\SQ}_{\1}'\,,
\end{equation}
with $\SQ_\1''$ being defined as 
\begin{equation}\label{SL1}
\SQ_\1''\,=\,2\cosh 2\pi b\spp_\1+e^{-2\pi b\sq_\1}\,.
\end{equation}
\item $\SS_\1$ maps $L^2(\BR^2)$ to $\CM$ in such a way
that $\SQ_\1''$ is mapped to the multiplication operator $\SQ$, 
\begin{subequations}
\begin{align}
& \SS_\1^{-1}\cdot \SE\cdot \SS_\1^{}\,=\,
\SE_\2\,,\\
&\SS_\1^{-1}\cdot \SK\cdot \SS_\1^{}\,=\,\SK_\2\,,\\
&\SS_\1^{-1}\cdot \SQ\cdot \SS_\1^{}\,=\,\SQ_\1''\,.
\end{align}
\end{subequations}
\item $\nu^{\mss_{\2\1}}_{s_\2s_\1}$ is a normalization factor that 
may depend on the positive self-adjoint operator
$\mss_{\2\1}$ defined by $\SQ_{\2\1}=2\cosh(2\pi b\mss_{\2\1})$.
A convenient choice for $\nu^{\mss_{\2\1}}_{s_\2s_\1}$ will be defined later. 
\end{itemize}
It follows easily that the operator defined in \rf{CGMconstr}
satisfies \rf{intertw}.

\subsection{Construction of the Clebsch-Gordan maps}\label{CGconstr}

The operators $\ST_{\1\2}$ and $\SC_1$ in \rf{CGMconstr} 
can be constructed explicitly as
\begin{subequations}
\begin{align}
&\ST_{\1\2}:=e_b(\sq_\1+\spp_\2-\sq_\2)e^{-2\pi i\spp_\1\sq_\2}\,,\\
&\SC^{-1}_\1:=e_b(\sq_\1-s_\2)e^{2\pi i s_\1\spp_\1}
\frac{e_b(s_\1-\spp_\1)}{e_b(s_\1+\spp_\1)} 
e^{2\pi i s_\2\sq_\1}
\,.
\end{align}
\end{subequations}
The operator $\SS_\1$ essentially coincides with the operator
that maps $\SL_\1$ to diagonal form. This operator can be
represented by the integral kernel
\begin{equation}\label{L1-eigenf}
\langle\,p_\1\,
|\,s_{\1}\,\rangle\,=\,\frac{e_b(s_{\1}+p_1+c_b-i0)}{e_b(s_{\1}-p_1-c_b+i0)}
e^{-2\pi i s_1(p_1+c_b)}\,.
\end{equation}
The functions $\phi_{s_\1}(p_\1):=
\langle\,p_\1\, |\,s_{\1}\,\rangle$
are nothing but the eigenfunctions of the operator $\SL_\1$ in the
representation where $\spp_1$ is diagonal. 
It was shown in \cite{Ka4} that the eigenfunctions $\phi_{s_\1}(p_\1)$
are delta-function orthogonalized and 
complete in $L^2(\BR)$,
\begin{subequations}
\begin{align}
&\int_{\BR_+}dp_\1\;\langle\,s_\1\,|\,p_\1\,\rangle
\langle\,p_\1\,|\,s_{\1}'\,\rangle\,=\,\de(s_\1^{}-s_\1')\,.\\
&\int_{\BR_+}d\mu(s_\1)\;\langle\,p_\1\,|\,s_{\1}\,\rangle\langle\,s_\1\,|\,p_\1'\,\rangle\,=\,\de(p_\1^{}-p_\1')\,.
\end{align}
\end{subequations}
This is equivalent to unitarity of the operator $\SS_\1$.

\subsection{Verification of intertwining property}

We want to demonstrate that the operator $\ST_{\1\2}$ satisfies \rf{ST12}. In order to see 
this, let us
calculate
\begin{align*}
\ST_{\1\2}\cdot\SE_{2}\cdot(\ST_{\1\2})^{-1}\,=& \,
e_b(\sq_\1+\spp_\2-\sq_\2)\cdot e^{\pi b(2\sq_\2-(\spp_\2+\spp_\1))}\cdot
(e_b(\sq_\1+\spp_\2-\sq_\2))^{-1} \\
=& e^{\frac{\pi b}{2}(2\sq_\2-(\spp_\2+\spp_\1))}\cdot
\frac{e_b(\sq_\1+\spp_\2-\sq_\2-ib/2)}{e_b(\sq_\1+\spp_\2-\sq_\2+ib/2)}\cdot
e^{\frac{\pi b}{2}(2\sq_\2-(\spp_\2+\spp_\1))}\\
=& e^{\frac{\pi b}{2}(2\sq_\2-(\spp_\2+\spp_\1))}\cdot
(1+e^{2\pi b(\sq_\1+\spp_\2-\sq_\2)})\cdot
e^{\frac{\pi b}{2}(2\sq_\2-(\spp_\2+\spp_\1))}\\
=& \SE_\2\SK_\1+(\SK_\2)^{-1}\SE_\1\,.
\end{align*}
Equation \rf{ST12c} us verified as follows: Let us write $\ST_{\1\2}=
\mst_{\1\2}e^{-2\pi i \spp_1\sq_2}$, and calculate
\begin{align}\label{adT12}
& \ST_{\1\2}\cdot \tilde{\SL}_\1\cdot (\ST_{\1\2})^{-1}\,= \,\\
&=\,\mst_{\1\2}\cdot\big( 2
\cosh 2\pi b(\sq_1-\spp_1-\sq_2)+e^{-2\pi b(\sq_1-\sq_2)}Q_1
+e^{-2\pi b\spp_1}Q_2+e^{-2\pi b(\sq_1+\spp_1-\sq_2)}\big)\cdot\mst_{\1\2}^{-1} \notag\\
&=\,2
\cosh 2\pi b(\sq_1-\spp_1-\sq_2)\notag\\
&\qquad+Q_1 e^{-\pi b(\sq_1-\sq_2)}
(1+e^{2\pi b(\sq_1+\spp_2-\sq_2)})e^{-\pi b(\sq_1-\sq_2)}\notag\\
& \qquad + Q_2e^{-\pi b\spp_1}(1+e^{2\pi b(\sq_1+\spp_2-\sq_2)})
e^{-\pi b\spp_1}\notag\\
&\qquad +e^{-\pi b(\sq_1+\spp_1-\sq_2)}
(1+e^{2\pi b(\sq_1+\spp_2-\sq_2+ib/2)})(1+e^{2\pi b(\sq_1+\spp_2-\sq_2-ib/2)})
e^{-\pi b(\sq_1+\spp_1-\sq_2)}
\notag\end{align}
Comparison of this expression with \rf{Q21} proves \rf{ST12c}.

The calculations needed to verify \rf{SC1} are very similar.

\subsection{The b-Clebsch-Gordan coefficients}

The b-Clebsch-Gordan coefficients are defined as the matrix elements of the
Clebsch-Gordan operator,
\begin{equation}
\big(\,{}^{s_3}_{p_3}\,|\,{}^{s_\2}_{p_\2}\,{}^{s_\1}_{p_\1}\big)_b^{}:=
\langle\,s_\3,p_\3\,|\,\SC_{s_\2s_\1}\,|\,p_\2,p_\1\,\rangle\,.
\end{equation}
We have
\begin{propn}\label{CGpropn:T} $\quad$\\
There exists a choice of coefficients $\nu^{\mss_{\2\1}}_{s_\2s_\1}$ such that 
the following statements are true:
\begin{enumerate} 
\item The b-Clebsch-Gordan coefficients are 
explicitly given by the formula
\begin{align}
\big(\,{}^{s_3}_{p_3}\,|\,{}^{s_\2}_{p_\2}\,{}^{s_\1}_{p_\1}\big)_b^{}=&\,
\de(p_{\2\1}-p_\2-p_\1)\,\left(\frac{w_b(s_\1+s_\2-s_\3)w_b(s_\1+s_\3-s_\2)
w_b(s_\2+s_\3-s_\1)}{w_b(s_\1+s_\2+s_\3)}\right)^{\frac{1}{2}}
\notag\\
&\times e^{\frac{\pi i}{2}(p_3^2-p_\1^2-p_\2^2)}\,
\frac{w_b(p_\1-s_\1)w_b(p_\2-s_\2)}{w_b(p_{\2\1}-s_{\2\1})}
e^{\pi i(p_\2(s_\1+c_b)-p_\1(s_\2+c_b))}\notag\\
& \times\int_{\BR}dp\;e^{\pi ip(s_\1+s_\2-s_{\2\1}+c_b)} 
D_{\frac{1}{2}(s_\2-s_\1-s_{\2\1}-c_b)}(p+p_\2)
D_{\frac{1}{2}(s_\1-s_\2-s_{\2\1}-c_b)}(p-p_\1)
\notag
\\
& \hspace{4cm} \times D_{\frac{1}{2}(s_\1+s_\2+s_{\2\1}-c_b)}^{}(p)\,.
\label{b-3j}\end{align}
\item The following Weyl-symmetries hold:
\begin{equation}\label{Weyl}
\big(\,{}^{s_3}_{p_3}\,|\,{}^{s_\2}_{p_\2}\,{}^{s_\1}_{p_\1}\big)_b^{}
=\big(\,{}^{s_3}_{p_3}\,|\,{}^{s_\2}_{p_\2}\,{}^{-s_\1}_{\;\;\,p_\1}\big)_b^{}
=\big(\,{}^{s_3}_{p_3}\,|\,{}^{-s_\2}_{\;\;\,p_\2}\,{}^{s_\1}_{p_\1}\big)_b^{}
=\big(\,{}^{-s_3}_{\;\;\,p_3}\,|\,{}^{s_\2}_{p_\2}\,{}^{s_\1}_{p_\1}\big)_b^{}\,.
\end{equation}
\item The b-Clebsch-Gordan coefficients are real,
\begin{equation}\label{reality}
\big[\big(\,{}^{s_3}_{p_3}\,|\,{}^{s_\2}_{p_\2}\,{}^{s_\1}_{p_\1}\big)_b^{}\big]^*
=\big(\,{}^{s_3}_{p_3}\,|\,{}^{s_\2}_{p_\2}\,{}^{s_\1}_{p_\1}\big)_b^{}\,.
\end{equation}
\end{enumerate}
\end{propn}
The proof can be found in Appendix \ref{proofCGpropn}.

The unitarity of of the Clebsch-Gordan maps $\SC_{s_\2s_\1}$ 
is equivalent to the following orthogonality
and completeness relations for the Clebsch-Gordan coefficients,
\begin{align}
& \int_{\BR^2} dp_{\2}dp_{\1}\;
\big[\big(\,{}^{s_{\2\1}}_{p_{\2\1}}\,|\,{}^{s_\2}_{p_\2}\,
{}^{s_\1}_{p_\1}\,\big)\big]^*
\big(\,{}^{s_{\2\1}'}_{p_{\2\1}'}\,|\,
{}^{s_\2}_{p_\2}\,{}^{s_\1}_{p_\1}\,\big) = 
\de(p_{\2\1}-p_{\2\1}')\de(s_{\2\1}-s_{\2\1}')\\
& \int_{\BR_+}d\mu(s_{\2\1})\int_\BR dp_{\2\1}\;
\big[\big(\,{}^{s_{\2\1}}_{p_{\2\1}}\,|\,{}^{s_\2}_{p_\2}\,
{}^{s_\1}_{p_\1}\,\big)\big]^*
\big(\,{}^{s_{\2\1}}_{p_{\2\1}}\,|\,
{}^{s_\2}_{p_\2'}\,{}^{s_\1}_{p_\1'}\,\big) = 
\de(p_\1-p_\1')\de(p_\2-p_\2')\,.
\end{align}

We finally want to compare our results with those of \cite{PT2}.
In this reference 
the authors 
constructed Clebsch-Gordan maps $\SC'{}^{s_3}_{s_\2,s_\1}:
\CP_{s_\2}'\ot\CP_{s_\1}'\ra\CP_{s_\3}'$ as integral operators
of the form
\begin{equation}\label{CGdef}
(\SC'{}^{s_\3}_{s_\2,s_\1}\psi)(x_3)\,=\,
\int_{\BR^2}dx_\1dx_\2\;
\big(\,{}^{s_\3}_{x_\3}\,|\,{}^{s_\2}_{x_\2}\,{}^{s_\1}_{x_\1}\big)_b^{}
\,\psi(x_\2,x_\1)\,,
\end{equation}
where
\begin{align} \label{b-3ja}
\big(\,{}^{s_3}_{x_3}\,|\,{}^{s_\2}_{x_\2}\,{}^{s_\1}_{x_\1}\big)_b^{} &
 = 
N(s_\3,s_\2,s_\1)D_{-\frac{1}{2} (s_\1+s_\2+s_3+c_b)}
\big(x_\2-x_\1-\fr{s_3+c_b}{2}\big)
 \\ 
& \quad\times D_{-\frac{1}{2}(s_\2-s_3-s_\1+c_b)}^{}
\big(x_\2-x_3- \fr{s_\1+c_b}{2}\big)
D_{-\frac{1}{2}(s_\1-s_3-s_\2+c_b)}^{}
\big(x_3-x_\1- \fr{s_\2+c_b}{2}\big)\,. 
\notag\end{align}
The normalization factor will be chosen as
\begin{align}
N(s_\3,s_\2,s_\1)=\left(\frac{w_b(s_\1+s_\2+s_\3)w_b(s_\1+s_\2-s_\3)}
{w_b(s_\1+s_\3-s_\2)w_b(s_\2+s_\3-s_\1)}\right)^{\frac{1}{2}}\,.
\end{align}
\begin{propn}\label{Comppropn}
We have
\begin{equation}
\SC'{}^{s_3}_{s_\2,s_\1}\,=\,
\SU_{s_\3}\cdot\SC{}^{s_3}_{s_\2,s_\1}\cdot(\SU_{s_\2}^{-1}\ot\SU_{s_\1}^{-1})\,.
\end{equation}
\end{propn}
The proof is given in Appendix \ref{PTvsNT}.

\subsection{The fusion operation}

Let us now consider tensor products of three representations.
There are two natural ways to construct unitary operators
\begin{equation}\label{321target}
\SC_{s_\3s_\2s_\1}:\CP_{s_\3}\ot\CP_{s_\2}\ot\CP_{s_\1}\,\ra\, \CM\ot
\int^{\oplus}d\mu(s)\; e_s\,,
\end{equation}
that satisfy
\begin{equation}
\SC_{s_\3s_\2s_\1}\cdot
(\pi_{s_\3}\ot\pi_{s_\2}\ot\pi_{s_\1})(X)\,=\,(\pi_{\CM}(X)\ot 1)
\cdot\SC_{s_\3s_\2s_\1} \,.
\end{equation}
In  \rf{321target} we used the notation $e_s$ for the one-dimensional
module of the algebra of functions $f:\BS\ra\BC$ with action
given as $f\cdot e_s=f(s) e_s$. 
The variable $s$ represents the multiplicity 
with which the representation $\CM$ appears in the triple 
tensor product $\CP_{s_\3}\ot\CP_{s_\2}\ot\CP_{s_\1}$.
Two such operators can be constructed as
\begin{align}
\SC_{s_\3(s_\2s_\1)}:=
\,\int^{\oplus}d\mu(s_{\2\1})\;
\SC_{s_\3s_{\2\1}}\cdot(1\ot\SC^{s_{\2\1}}_{s_\2s_\1})\,,\\
\SC_{(s_\3s_\2)s_\1}:=
\,\int^{\oplus}d\mu(s_{\3\2})\;
\SC_{s_{\3\2}s_\1}\cdot(\SC^{s_{\3\2}}_{s_\3s_\2}\ot 1)\,.
\end{align}
The fusion operator $\SA_{s_\3s_\2s_\1}:
\int^{\oplus}d\mu(s_{\3\2})\int^{\oplus}d\mu(s_\4)\;\CP_{s_\4}\ra
\int^{\oplus}d\mu(s_{\2\1})\int^{\oplus}d\mu(s_\4)\;\CP_{s_\4}$ 
is defined as
\begin{align}
\label{fusdef}
& \SA_{s_\3s_\2s_\1}:=\SC_{s_\3(s_\2s_\1)}\cdot\big[\SC_{(s_\3s_\2)s_\1}\big]^{\dagger}\,.
\end{align}
This operator commutes with $\pi_{s_\4}$ and is therefore of the form
\begin{equation}
 \SA_{s_\3s_\2s_\1}\,=\,\int^{\oplus}d\mu(s_\4)\;\SA^{s_\4}_{s_\3s_\2s_\1}\,,
\end{equation}where $\SA^{s_\4}_{s_\3s_\2s_\1}$ is a unitary operator
$\SA^{s_\4}_{s_\3s_\2s_\1}:
\int^{\oplus}d\mu(s_{\3\2})\,e_{s_{\3\2}}\longrightarrow
\int^{\oplus}d\mu(s_{\2\1})\,e_{s_{\2\1}}\simeq
L^2(\BS,d\mu).
$

\subsection{The b-j symbols}

The b-6j symbols are defined as the matrix elements of the operator
$\SA^{s_\4}_{s_\3s_\2s_\1}$,
\begin{equation}
\big\{\,{}^{s_\1}_{s_\3}\;{}^{s_\2}_{s_\4}\;{}^{s_{\2\1}}_{s_{\3\2}}\big\}_b^{}:=
\langle\,s_{\2\1}\,|\,\SA^{s_\4}_{s_\3s_\2s_\1}\,|\,s_{\3\2}\,\rangle\,.
\end{equation}
Proposition \ref{Comppropn} allows us to use the results from 
\cite{PT2,TV} for the calculation of these matrix elements. The
result is
\begin{align}\label{6j3}
\big\{\,{}^{s_1}_{s_4}\;{}^{s_2}_{s_5}\;{}^{s_3}_{s_6}\big\}_b^{}
&=\Delta(\al_3,\al_2,\al_1)\Delta(\al_5,\al_4,\al_3)\Delta(\al_6,\al_4,\al_2)
\Delta(\al_6,\al_5,\al_\1)\\
&\qquad \times\int\limits_{\CC}du\;
S_b(u-\alpha_{321}) S_b(u-\al_{543}) S_b(u -\alpha_{642}) 
S_b(u-\alpha_{651})
\notag \\[-1.5ex] & \hspace{2cm}\times 
S_b( \alpha_{4321}-u) S_b(\alpha_{6431}-u) 
S_b(\alpha_{6532}-u) S_b(2Q-u)\,.
\notag\end{align}
The expression involves the following ingredients:
\begin{itemize}
\item We have used the notations $\al_i=\frac{Q}{2}+{\mathrm i}s_i$, as
well as
$\al_{ijk}=\al_i+\al_j+\al_k$,
$\al_{ijkl}=\al_i+\al_j+\al_k+\al_l$ for $i,j,k,l\in\{1,2,3,4,5,6\}$.
\item $\Delta(\al_3,\al_2,\al_1)$ is defined as 
\begin{align}
&\Delta(\al_3,\al_2,\al_1)=\bigg(\frac{S_b(\alpha_1+\alpha_2+\alpha_s-Q)}{S_b(\alpha_1+\alpha_2-\alpha_s) 
S_b(\alpha_1+\alpha_s-\alpha_2) S_b(\alpha_2+\alpha_s-\alpha_1)} 
\bigg)^{\frac{1}{2}}\,.
\notag\end{align}
\item The integral is defined in the cases
that $\al_k\in Q/2+\textup{i}\BR$ by a contour $\CC$ which 
approaches $2Q+\textup{i}\BR$ near infinity,
and passes the real axis in the interval $(3Q/2,2Q)$.
For other values of the variables $\al_k$ it is defined by analytic 
continuation.
\end{itemize}




\section{Quantum Teichm\"uller theory}
\setcounter{equation}{0}

This section presents the definitions and results from the 
quantum Teichm\"uller theory that will be needed in this paper.
We will use the formulation introduced by R. Kashaev \cite{Ka1}, 
see also \cite{T05} for a more detailed exposition and a discussion
of its relation to the framework of Fock \cite{F97} and
Chekhov and Fock \cite{CF}. The formulation from \cite{Ka1}
starts from the quantization of a somewhat enlarged 
space $\hat{\CT}(C)$. The usual Teichm\"uller space $\CT(C)$ 
can then be characterized as subspace of $\hat{\CT}(C)$ using 
certain linear constraints. This is motivated by the observation 
that the spaces $\hat{\CT}(C)$ have natural polarizations,
which is not obvious in the formulation of \cite{F97,CF}.

\subsection{Algebra of operators and its representations}

For a given surface $C$ with 
constant negative curvature metric and at least one puncture one considers
ideal triangulations $\tau$. Such ideal triangulations are defined
by maximal collection of non-intersecting open geodesics which start and
end at the punctures of $C$. We will assume that the triangulations
are decorated, which means that a distinguished corner is chosen in
each triangle.

We will find it convenient
to parameterize triangulations by their dual graphs which are called
fat graphs $\vf_\tau$. The vertices of $\vf_\tau$ are in one-to-one correspondence
with the triangles of $\tau$, and the edges of $\vf_\tau$ are in 
one-to-one correspondence
with the edges of $\tau$. The relation between a triangle $t$ in 
$\tau$ and the fat graph $\vf_{\tau}$ is depicted in Figure \ref{decor}.
$\vf_\tau$ inherits a natural decoration of its
vertices from $\tau$, as is also indicated in Figure \ref{decor}.
\begin{figure}[htb]
\epsfxsize5cm
\centerline{\epsfbox{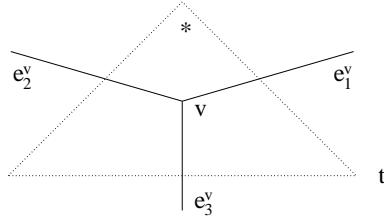}}
\caption{\it Graphical representation of the vertex $v$ dual
to a triangle $t$. The marked corner defines a
corresponding numbering of the edges that emanate at $v$.}
\label{decor}
\end{figure}

The quantum theory associated to the Teichm\"uller space $\CT(C)$
is defined on the kinematical level by associating to each vertex $v\in\vf_0$,
$\vf_0=\{{\rm vertices\;\,of\;\,}\vf\}$,
of $\vf$ a pair of generators $p_v,q_v$ which are supposed to satisfy
the relations
\begin{equation}
\big[\,p_v\,,\,q_{v'}\,\big]\,=\,\frac{\de_{vv'}}{2\pi i}\,.
\end{equation}
There is a natural representation of this algebra on the 
Schwarz space $\hat{\CS}(C)$ of rapidly decaying smooth functions $\psi(q)$, 
$q:\vf_0\ni v\ra q_v$, generated from 
$\pi_\vf(q_v):=\sq_v$, $\pi_\vf(p_v):=\spp_v$, where
\begin{equation}
\sq_v\,\psi(q):=\,q_v\psi(q)\,,\qquad
\spp_v\,\psi(q):=\,\frac{1}{2\pi i}\frac{\pa}{\pa q_v}\psi(q)\,.
\end{equation}
For each surface $C$ 
we have thereby defined an algebra $\hat{\CA}(C)$ together with 
a family
of representations $\pi_\vf$ of $\hat{\CA}(C)$ on the 
Schwarz spaces $\hat{\CS}_\vf(C)$ which are dense subspaces of
the Hilbert space $\CK(\vf)\simeq L^2(\BR^{4g-4+2n})$. 
The next step is to show that the choice of fat graph $\vf$ is
inessential by constructing unitary operators 
$\pi_{\vf_\2\vf_\1}:\CK(\vf_1)\ra\CK(\vf_2)$
intertwining the 
representations  $\pi_{\vf_1}$ and $\pi_{\vf_\2}$.

\subsection{The projective representation of the 
Ptolemy groupoid on $\CK(\vf)$ \index{Ptolemy groupoid}}
\label{qPtolemy}

The groupoid generated by the changes from 
one fat graph to another is called the Ptolemy groupoid.
It can be described in terms of generators and relations,
see e.g. \cite[Section 3]{T05} for a summary of the relevant results 
and further references.

Following \cite{Ka3} closely we shall define a 
projective unitary representation of the Ptolemy groupoid in 
terms of the following set of 
unitary operators 
\begin{equation}\begin{aligned}
\SA_v\;\equiv\;& e^{\frac{\pi i}{3}}
e^{-\pi i (\spp_v+\sq_v)^2}e^{-3\pi i \sq_v^2}\\
\ST_{vw}\;\equiv\;& e_b(\sq_v+\spp_w-\sq_w)e^{-2\pi i\spp_v\sq_w} ,
\end{aligned}\qquad
\text{where}\;\;v,w\in\vf_{\zero}\,.
\label{qPtgens}\end{equation}
The special function $e_b(U)$ can be defined in the strip 
$|\Im z|<|\Im c_b|$, $c_b\equiv i(b+b^{-1})/2$ by means of the 
integral representation
\begin{equation}
\log e_b(z)\;\equiv\;\frac{1}{4}
\int\limits_{i0-\infty}^{i0+\infty}\frac{dw}{w}
\frac{e^{-2{\mathsf i}zw}}{\sinh(bw)
\sinh(b^{-1}w)}.
\end{equation}
These operators are unitary for $(1-|b|)\Im b=0$. They satisfy the 
following relations \cite{Ka3}
\begin{subequations}
\begin{align}
{\rm (i)}& 
\qquad\ST_{vw}\ST_{uw}\ST_{uv}\;=\;\ST_{uv}\ST_{vw},\label{pentrel}\\
{\rm (ii)}& \qquad\SA_{v}\ST_{uv}\SA_{u}\; =\; 
\SA_{u}\ST_{vu}\SA_{v},
\label{symrel}\\
{\rm (iii)}& \qquad\ST_{vu}\SA_{u}\ST_{uv}
\;=\;\zeta\SA_{u}\SA_{v}\SP_{uv},
\label{invrel}\\
{\rm (iv)}& \qquad\SA_u^3\;=\;\id,\label{cuberel}
\end{align}
\end{subequations}
where $\zeta=e^{\pi i c_b^2/3}$, $c_b\df\frac{i}{2}(b+b^{-1})$.
The relations 
\rf{pentrel} to \rf{cuberel} allow us to define a projective 
representation of the Ptolemy groupoid as follows.  
\begin{itemize}
\item Assume that $\om_{uv}\in[\vf',\vf]$.
To $\om_{uv}$ let us associate
the operator
\[ 
\su(\omega_{uv})\,\df\,\ST_{uv}\;:\;\CK(\vf)\ni \fv\;\,\ra\;\,
\ST_{uv}\fv\in\CK(\vf').
\]
\item For each fat graph $\vf$ and vertices $u,v\in \vf_{\zero}$
let us define the following operators
\[ \begin{aligned}
{}& \SA_{u}^{\vf}\;:\;\CK(\vf)\ni \fv\;\,\ra\;\,
\SA_{u}^{}\fv\in\CK(\rho_{u}\circ\vf).\\
{}& \SP_{uv}^{\vf}\;:\;\CK(\vf)\ni \fv\;\,\ra\;\,
\SP_{uv}^{}\fv\in\CK((uv)\circ\vf).
\end{aligned}
\]
\end{itemize}
It follows immediately from \rf{pentrel}-\rf{cuberel}
that the operators $\ST_{uv}$, $\SA_{u}$ 
and $\SP_{uv}$ can be used to generate
a unitary projective representation of the Ptolemy groupoid.

The corrsponding automorphisms of the algebra $\CA(C)$ are
\begin{equation}
\sa_{\vf_\2\vf_\1}(\SO):=
\sa\sd[\SU_{\vf_\2\vf_\1}](\SO):=
\SU_{\vf_\2\vf_\1}^{}\cdot\SO\cdot\SU_{\vf_\2\vf_\1}^{}\,.
\end{equation}
The automorphism $\sa_{\vf_\2\vf_\1}$
generate the canonical quantization
of the changes of coordinates for $\hat{\CT}(C)$ 
from one fat 
graph to another \cite{Ka1}.

\subsection{The reduction to the Teichm\"uller spaces}
\label{constraints}

Recall that 
the quantum theory defined in this way is not quite the one we
are interested in. It is the 
quantum theory of an enlarged space $\hat{\CT}(C)$
which is the product of the 
Teichm\"uller space with the first homology of $C$, both 
considered as real vector spaces \cite{Ka1,T05}.
The embedding of the 
Teichm\"uller space $\CT(C)$ into  $\hat{\CT}(C)$
can be described classically 
in terms of a certain set of constraints $z_c=0$ which
characterize the locus of  $\CT(C)$ within  $\hat{\CT}(C)$.

To define the quantum representatives of the constraints 
let us introduce an embedding  
of the first homology $H_1(\Sigma,\BR)$ into $\hat{\CT}(C)$ as follows.
Each graph geodesic  $g_\ga$ which 
represents an element $\ga\in H_1(\Sigma,\BR)$ may 
be described by an ordered sequence of vertices
$v_i\in\vf_\0$, and edges $e_{i}\in\vf_\1$,
$i=0,\dots,n$, where $v_0=v_n$, $e_0=e_n$, and we assume 
that $v_{i-1}$, $v_{i}$ 
are connected by the single edge $e_i$. We will define $\omega_i=1$
if the arcs connecting $e_i$ and $e_{i+1}$ turn around the 
vertex $v_i$ in the counterclockwise sense, $\omega_i=-1$ otherwise.
The edges emanating from $v_i$ will be numbered
$e^i_{j}$, $j=1,2,3$ according to the convention introduced in Figure
\ref{decor}. 
To each $c\in H_1(\Sigma,\BR)$ we will assign
\begin{equation}\label{constrdef}
z_c\df\sum_{i=1}^{n} u_i,\qquad 
u_i:\;=\;\omega_i\left\{ \begin{array}{ll} 
-q_{v_i} &\;\text{if $\{e_{i},e_{i+1}\}=\{e_3^i,e_1^i\}$,}\\ 
\phantom{-}p_{v_i}   &\;\text{if $\{e_{i},e_{i+1}\}=\{e_2^i,e_3^i\}$,}\\ 
q_{v_i}-p_{v_i} & \; \text{if $\{e_{i},e_{i+1}\}=\{e_1^i,e_2^i\}$.}
\end{array}\right. 
\end{equation}
Let $C_{\vf}$ be the
subspace in $\hat{\CT}(C)$ that is spanned by the $z_{c}$, 
$c\in H_1(\Sigma,\BR)$.

\begin{lem} $\frac{\quad}{}$ \cite{Ka1}
The mapping $H_1(\Sigma,\BR)\ni c\mapsto z_{c}\in C_\vf$
is an isomorphism of Poisson vector spaces.
\end{lem}

Replacing $q_v$ by $\sq_v$ and $p_v$ by $\spp_v$ in the definition
above gives the definition of the operators $\sz_c\equiv \sz_{\vf,c}$ 
which represent
the constraints in the quantum theory.  
Let us note that the constraints transform under a
change of fat graph as
$\sa_{\vf_\2\vf_\1}(\sz_{\vf_\1,c})=\sz_{\vf_\2,c}$.

\subsection{Length operators}

A particularly important class of coordinate functions on the
Teichm\"uller spaces are the geodesic length functions.
The quantization of these observables was studied in 
\cite{CF,CF2,T05}.

Such length operators can be constructed in general 
as follows \cite{T05}. 
We will first define
the length operators for a case in which the 
choice of fat graph $\vf$ simplifies the representation of the 
curve $c$. 
We then explain how to
generalize the definition to all other cases.

Let $A_c$ be an annulus embedded in the surface $C$
containing the curve $c$, and let 
$\vf$ be a fat graph which looks inside of $A_c$ 
as depicted in Figure \ref{annfig}.
\begin{equation}\label{annfig}
\lower.9cm\hbox{\epsfig{figure=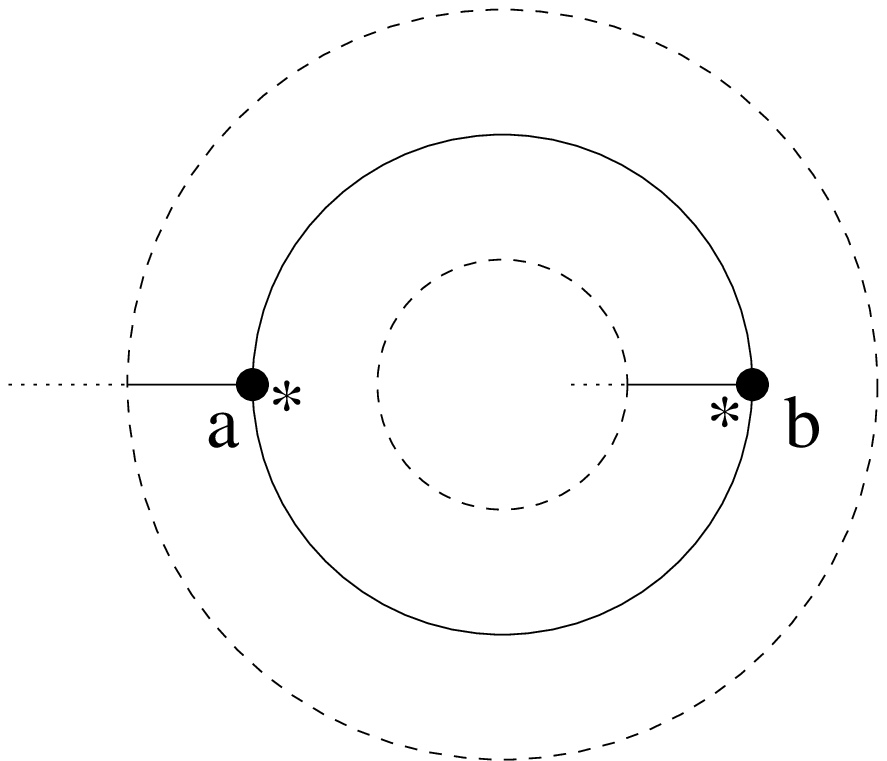,height=3cm}}
\qquad
\begin{aligned}
& \text{Annulus $A_c$: Region bounded} \\
& \text{by the two dashed circles,}\\
& \text{and part of $\vf_\si$ contained in $A_c$.}
\end{aligned}
\end{equation}
Let 
\begin{equation}\label{SLsimple}
\SL_{\vf,c}:=\,2\cosh 2\pi b\spp_c+
e^{-2\pi b\sq_c}\,,
\end{equation}
where $\spp_c:=\frac{1}{2}(\spp_a-\sq_a-\spp_b)$, 
$\sq_c:=\frac{1}{2}(\sq_a+\spp_a+\spp_b-2\sq_b)$.

In all remaining cases we will define the length operator
$\SL_{\vf,c}$ as follows: There always exists
a fat graph $\vf_0$ for which the definition 
above can be used to define 
$\SL_{\vf_0,c}$. Let then
\begin{equation}\label{genlength}
\SL_{\vf,c}:=\sa_{\vf,\vf_0}^{}(\SL_{\vf_0,c})\,.
\end{equation}
It can be shown that the length operators $\SL_{\vf,c}$ are unambigously
defined in this way \cite{T05}.

The length operators satisfy the following properties:
\begin{enumerate}
\item[(a)] {\bf Spectrum:} $\SL_{\vf,c}$ is self-adjoint.
The spectrum of $\SL_{\vf,c}$ is simple and equal
to $[2,\infty)$. This is necessary and sufficient for the existence
of an operator $\sll_{\vf,c}$ - the {\it geodesic length operator} - 
such that 
$\SL_{\vf,c}=2\cosh\frac{1}{2}\sll_c$.
\item[(b)] {\bf Commutativity:} 
\[
\big[\,\SL_{\vf,c}\,,\,\SL_{\vf,c'}\,\big]\,=\,0\quad
{\rm if}\;\; c\cap c'=\emptyset.
\]
\item[(c)] {\bf Mapping class group invariance:}
\[ 
\sa_\mu(\SL_{\vf,c})\,=\,\SL_{\mu.\vf,c},
\quad\sa_\mu\equiv\sa_{[\mu.\vf,\vf]},\quad
\text{for all}\;\;\mu\in{\rm MC}(\Sigma).
\]
\end{enumerate}
It can furthermore be shown that this definition reproduces the classical
geodesic length functions on $\CT(C)$ in the classical limit.

As an example for the use of 
\rf{genlength} that will be important for the following let us
assume that the curve $c$ is the boundary component
of a trinion $P_c$ embedded in $C$ within which the fat graph $\vf'$ 
looks as follows:
\begin{equation}\label{trinfig}
\lower.9cm\hbox{\epsfig{figure=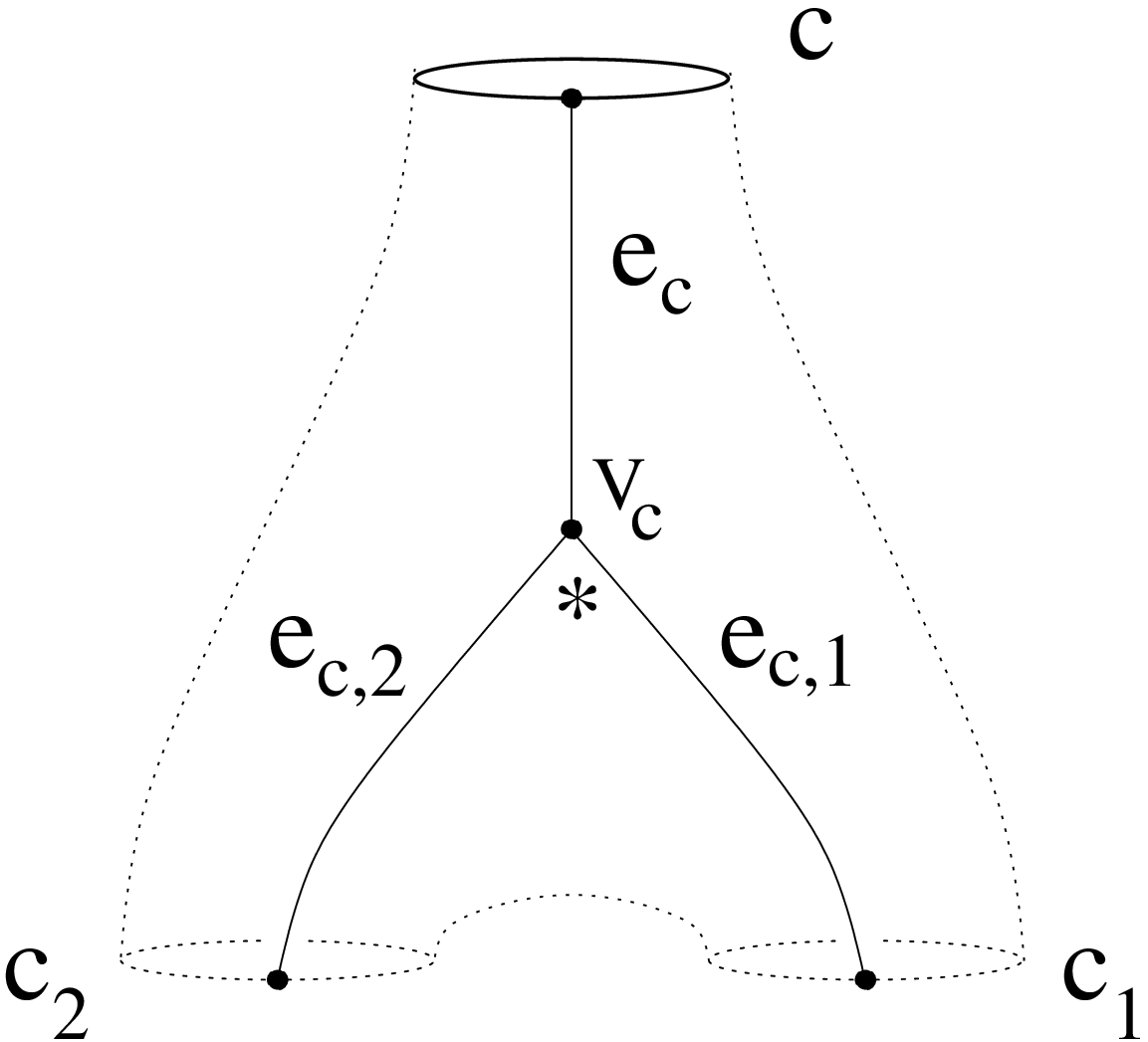,height=3cm}}\,.
\end{equation}
Let $c_\ep$, $\ep=1,2$ be the curves
which represent the other boundary components
of $P_c$ as 
indicated in Figure \ref{trinfig}.

\begin{propn}\label{Lreppropn}
$\SL_c$ is given by
\begin{equation}\label{modLlem2}
\SL_{\vf',c}\;=\;2\cosh(\sy_c^{\2}+\sy_c^{\1})
+e^{-\sy_c^{{\2}}}\SL_{c_\1}^{}
+e^{\sy_c^{\1}}\SL_{c_\2}^{}
+e^{\sy_c^{\1}-\sy_c^{\2}}\, ,
\end{equation}
where $\sy_c^{{\ep}}$, $\ep=1,2$
are defined as
$\sy_c^{\2}=2\pi b ({\sq}_{c}+\sz_{c_\2})$, 
$\sy_c^{\1}=- 2\pi b ({\spp}_{c}-\sz_{c_\1})$.
\end{propn}
The proof of \rf{modLlem2} can be found in Appendix 
\ref{Proofapp}.

\subsection{The annulus}

As a basic building block let us develop the quantum 
Teichm\"uller theory of an annulus in some detail.
To the simple closed curve $c$ that can be embedded into $A$ we associate
\begin{itemize}
\item the constraint 
\begin{equation}
\sz:=\frac{1}{2}(\spp_a-\sq_a+\spp_b)\,,
\end{equation}
\item the length operator $\SL$ is defined as in \rf{SLsimple}.
\end{itemize}
The operator $\SL$ is positive-self-adjoint, and its spectral 
decomposition \cite{Ka4} was recalled in
the above. 

For later use let us construct the change of representation from 
the representation in which $\spp_a$ and $\spp_b$ are diagonal
to a representation where $\sz$ and $\SL$ are diagonal.
To this aim let us introduce $\sd:=\frac{1}{2}(\sq_a+\spp_a-\spp_b+2\sq_b)$.
We have 
\[
\begin{aligned}
{[}{\sz},\sd]=(2\pi i)^{-1}\,,\\
[\spp,\sq]=(2\pi i)^{-1}\,,
\end{aligned}\qquad
\begin{aligned}
&[\sz,\spp]=0\,,\quad
[\sz,\sq]=0\,,\\
&[\sd,\spp]=0\,,\quad
[\sd,\sq]=0\,.
\end{aligned}
\]
Let $\langle\,p,z\,|$ be an eigenvector of $\spp$ and $\sz$ with
eigenvalues $p$ and $z$, respectively, and
$|\,p_a,p_b\,\rangle$ an eigenvector of $\spp_a$ and $\spp_b$
with eigenvalues $p_a$ and $p_b$, respectively.
It follows easily that 
\begin{equation}
\langle\,p,z\,|\,p_a,p_b\,\rangle\,=\,\de(p_b-z+p)e^{\pi i(p+z-p_a)^2}\,.
\end{equation}
The transformation 
\begin{equation}\label{holediag}
\psi(s,z)\,=\,\int_{\BR^2}dpdp_a\;\frac{w_b(s-p+c_b-i0)}{w_b(s+p-c_b+i0)}
e^{\pi i(p+z-p_a)^2}\Psi(p_a,z-p)\,,
\end{equation}
will then map a wave function $\Psi(p_a,p_b)$ in the representation
which diagonalizes $\spp_a$, $\spp_b$ to the 
corresponding wave function $\psi(s,z)$ in the representation
which diagonalizes $\SL$ and $\sz$.

\subsection{Teichm\"uller theory for surfaces with holes}\label{ssec:hole}

The formulation of quantum Teichm\"uller theory introduced above
has only punctures (holes with vanishing geodesic circumference) 
as boundary components. In order to generalize to holes of 
non-vanishing geodesic circumference one may represent each hole
as the result of cutting along a geodesic surrounding a pair of
punctures.
\[
\lower.9cm\hbox{\epsfig{figure=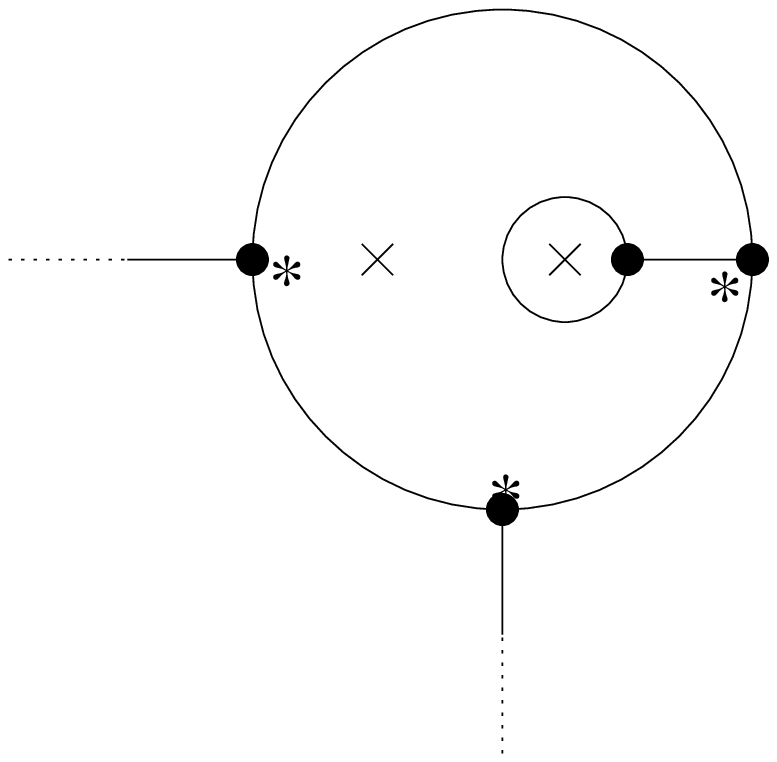,height=3cm}}
\quad
\begin{aligned}
& \text{Example for a} \\
& \text{fat graph in the} \\
& \text{vicinity of two} \\
& \text{punctures (crosses)}\\
& \text{}
\end{aligned}
\qquad
\lower.9cm\hbox{\epsfig{figure=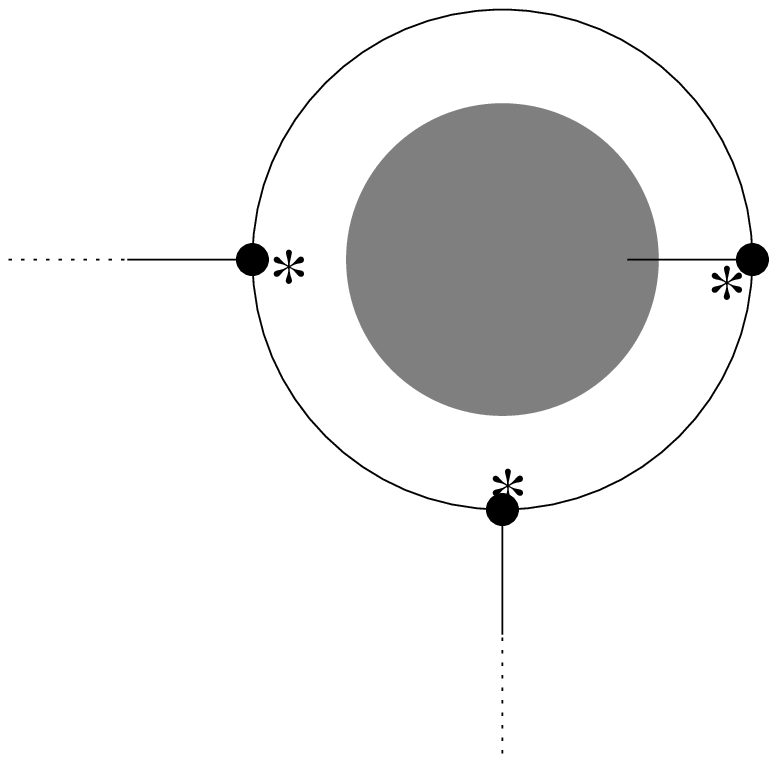,height=3cm}}
\quad
\begin{aligned}
& \text{The same fat graph} \\
& \text{after cutting} \\
& \text{out the hole}\\
& \text{}
\end{aligned}
\]
On a surface $C$ with $n$ holes one may choose $\vf$ to have the following
simple standard form near at most $n-1$ of the holes:
\begin{equation}\label{holefig}
\lower.9cm\hbox{\epsfig{figure=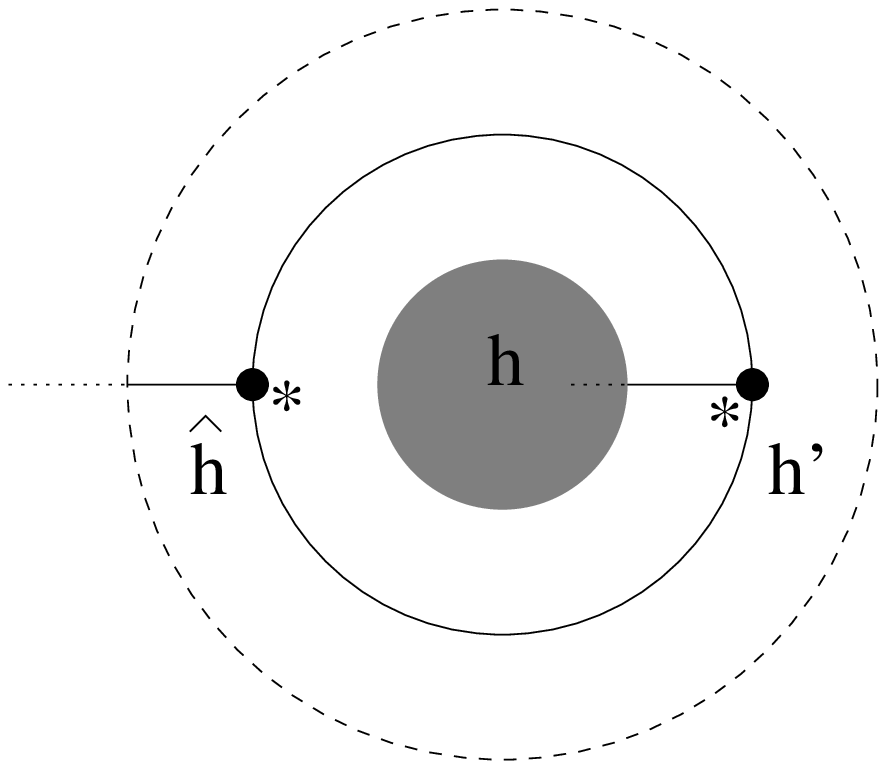,height=3cm}}
\end{equation}
The price to pay is a fairly complicated representation of the 
closed curves which surround the remaining holes.

The simple form of the fat graph near the incoming 
boundary components allows us to use the 
transformation \rf{holediag} 
to pass to a representation where the 
length operators and constraints
associated to these holes are diagonal.
In order to describe the resulting hybrid representation
let us denote by $s_b$ and $z_b$ the assignments of values 
$s_h$ and $z_h$ to each 
incoming hole $h$, while $p$ assigns real numbers $p_v$ to all
vertices $v$ of $\vf$ which do not coincide with any vertex
$\hat{h}$ or $h'$ associated to an incoming hole $h$.
The states will then be described by wave-functions 
$\psi(p;s_b,z_b)$ on which the operators $\SL_{h}$ and $\sz_{h}$
act as operators of multiplication by $2\cosh2\pi bs_h$ and 
$z_h$, respectively. 

For a given hole $h$ one may define a projection $\CH(C_{h(s,z)})$ 
of $\CH(C)$ to the eigenspace with fixed eigenvalues $2\cosh 2\pi bs$ 
and $c$ of $\SL_h$ and $\sz_h$. States in $\CH(C_{h(s,z)})$ 
are represented by wave-functions $\psi_h(p_h)$, where $p_h$ assigns real
values to all vertices in $\vf_0\setminus\{\hat{h},h'\}$.
The mapping class
action on $\CH(C)$ commutes with $\SL_h$ and $\sz_h$. It follows that
the operators $\SM(\mu)$ representing 
the mapping class group action on $\CH(C)$ project to 
operators  $\SM_{s,z}(\mu)$ generating an action of
${\rm MCG}(C)$ on $\CH(C_{h(s,z)})$.


\subsection{The cutting operation}

Cutting $C$ along the curve $c$ embedded in an annulus as considered
above will produce two surfaces $C''$ and $C'$ with boundary 
containing copies of the curve $c$. We may regard $C''$ and $C'$ 
as subsurfaces of $C$. The mapping class groups 
${\rm MCG}(C'')$ and ${\rm MCG}(C')$ thereby get embedded
as subgroups into ${\rm MCG}(C)$. The images of 
${\rm MCG}(C'')$ and ${\rm MCG}(C')$  are generated by 
the Dehn twist along $c$ together with diffeomorphisms
of $C''$ and $C'$ which act trivially on $A$, respectively. 

The spectral decomposition of $\SL_{c}$ and $\sz_c$
defines a natural counterpart of the operation to cut $C$ into $C''$ and
$C'$ within the quantum Teichm\"uller theory. It produces an 
isomorphism
\begin{equation}\label{cutmapdef}
\SS_c:\CH(C)\,\mapsto\,\int^{\oplus}_{\BR_+}ds\int^{\oplus}_\BR dc\;
\CH\big(C''_{h''(s,z)}\big)\ot\CH\big(C'_{h'(s,z)}\big)\,.
\end{equation}
The explicit form of the operator 
$\SS_c$ is easily found with the help of the integral transformation 
\rf{holediag}. To this aim it is sufficient to split the set $\vf_0$ of 
vertices of $\vf$ as $\vf_0=\vf_0''\cup\{a,b\}\cup\vf_0'$, 
where $a$ and $b$ are the vertices lying inside $A$, and the set
$\vf_0'$ contains the vertices in $\vf_0\setminus\{a,b\}$
located in $C'$. Writing accordingly
$\Psi(p)=\Psi(p'',p_a,p_b,p')$, with $p'':\vf_{0}''\mapsto \BR$
and $p':\vf_{0}'\mapsto \BR$,
we may use the integral transformation \rf{holediag} to map 
$\Psi(p)=\Psi(p'',p_a,p_b,p')$ to a function 
$\psi(p'',s,z,p')$ which represents an element of the Hilbert space 
on the right of \rf{cutmapdef}.


\section{Relation between the modular double and 
quantum Teichm\"uller theory}
\setcounter{equation}{0}


We are now ready to address our main aim. Recall that the modular
double is characterized by the following main objects:
The operators $ \SC_{s_\2s_\1}$ which generate the co-product,
and the R-operator $\SR$. We are going to show that these 
operators have very natural interpretations in within 
the quantum Teichm\"uller theory.

\subsection{The hole algebra}

Recall that the representation $\pi_\CM$
of the modular double $\CD\CU_q(\fsl_2)$
has positive self-adjoint
generators $\SE$, $\SK$, $\SF$. It will 
again be convenient
to replace 
the generator $\SF$ by the Casimir $\SQ$ 
\begin{equation}
\SF=(q-q^{-1})^{-2}\big(\SQ+q\SK^2+q^{-1}\SK^{-2}\big)\SE^{-1}\,.
\end{equation}
 
We will identify the representation $\pi_\CM$
of the algebra  $\CD\CU_q(\fsl_2)$ with the hole algebra
which is associated to the following subgraph of a fat graph $\si$:
\[
\lower.9cm\hbox{\epsfig{figure=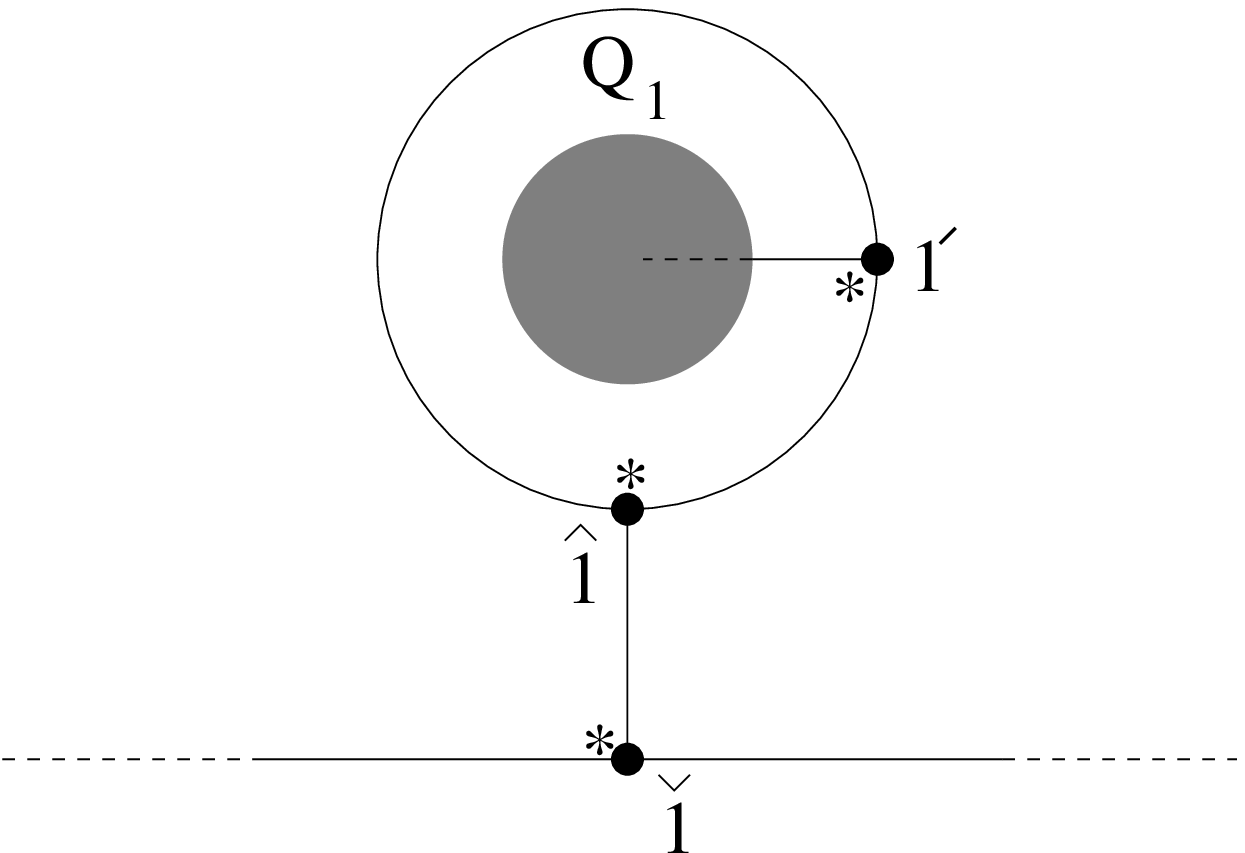,height=3cm}}
\]
The identification is such that
\begin{equation}
\begin{aligned}
&\SE\,\mapsto\,e^{\pi b(2\sq_{\check{\1}}-\spp_{\check{\1}})}\,,\\
&\SK\,\mapsto\,e^{-\pi b\spp_{\check{\1}}}\,,
\end{aligned}
\qquad \SQ\,\mapsto\,\SL\,.
\end{equation}

We furthermore 
note that local changes of the fat graph are naturally mapped to unitary
equivalence transformations 
of the representation $\pi_\CM$. A particularly important one is the 
equivalence transformation corresponding to the automorphism $w$. We have
\begin{propn}\label{W-propn}
The automorphism $\sw$ coincides with the automorphism
associated to the following move $W_\1$ 
\begin{equation}\label{Wmovedef}
\lower.9cm\hbox{\epsfig{figure=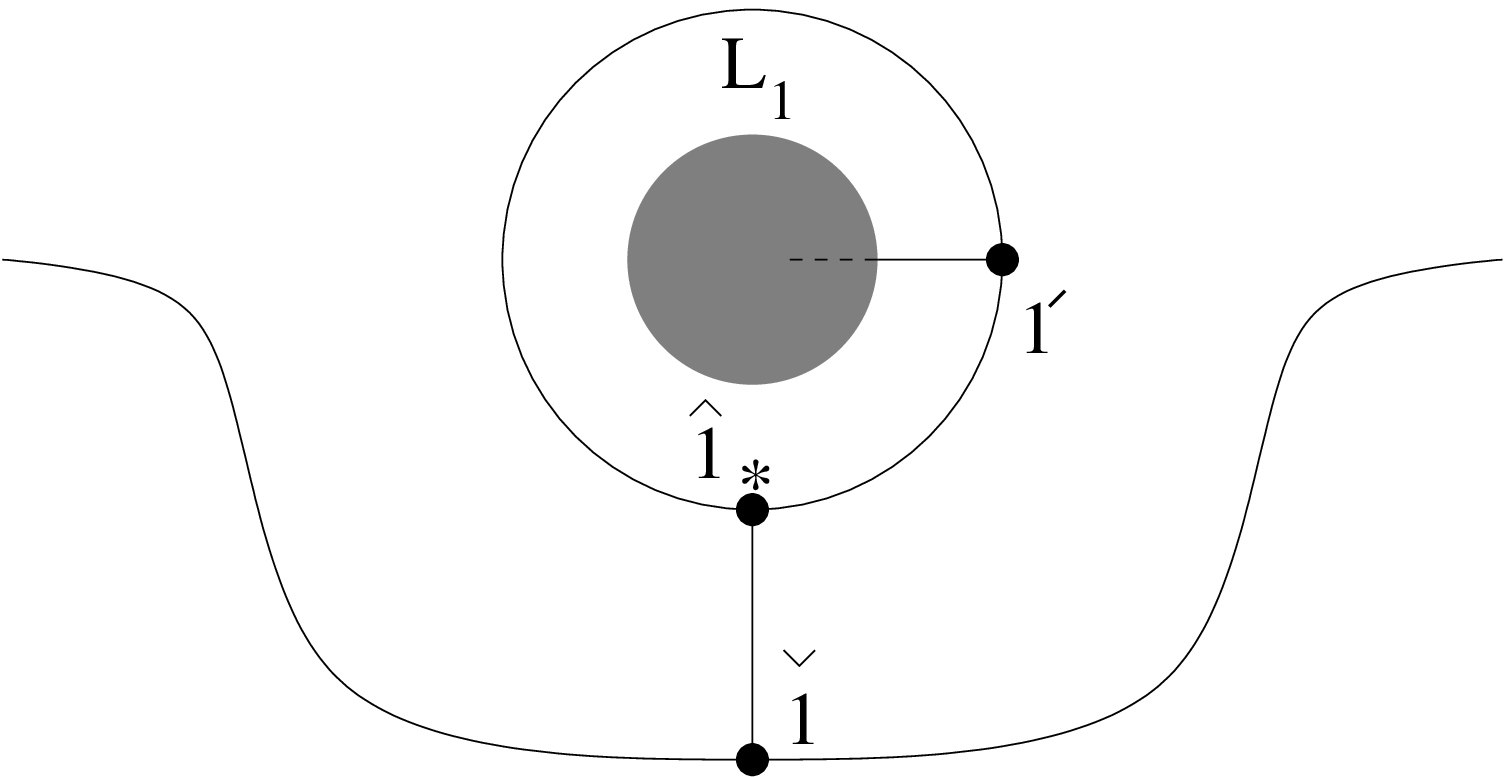,height=2.8cm}}
\quad \overset{W_\1}{\longrightarrow}\quad
\lower.9cm\hbox{\epsfig{figure=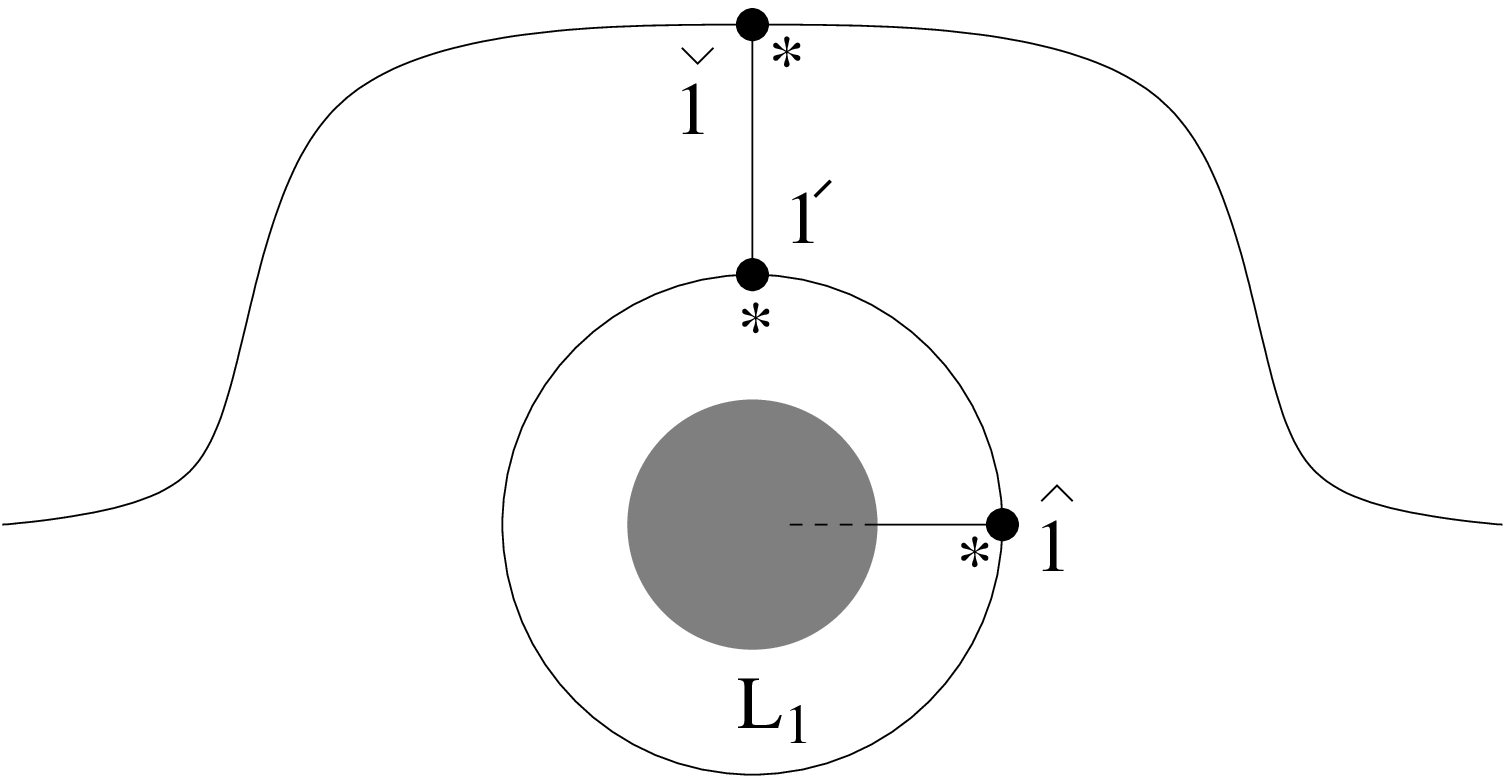,height=2.8cm}}
\end{equation}
after setting the constraint $\sz_\1$ to zero.
\end{propn}
The proof is given in Appendix \ref{Proofapp}.

\subsection{Tensor products of representations}

It is clearly natural to identify the tensor product of two representations
with the following subgraph
\begin{equation}\label{tensorfig}
\lower.9cm\hbox{\epsfig{figure=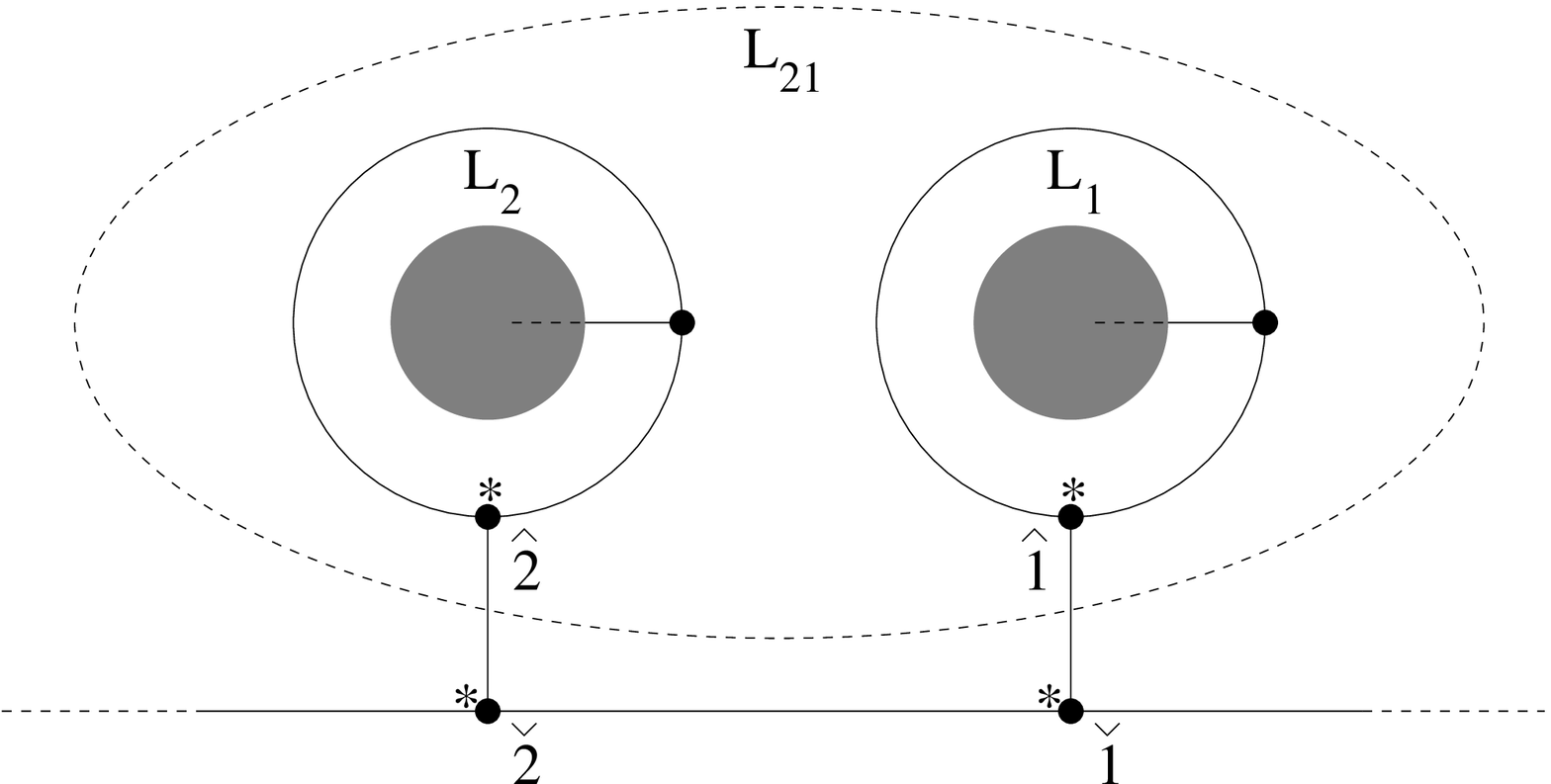,height=3cm}}
\end{equation}
Let $\SL_{\2\1}$ be the operator which represents the geodesic length
observable in the representation corresponding to the fat graph 
above.

The key observation to be made 
is formulated in the following proposition:
\begin{propn}\label{lengthCas} 
The projection of the length operator $\SL_{\2\1}$ 
onto the subspace of vanishing constraints becomes equal to the
Casimir $\SQ_{\2\1}$,
\begin{equation}
\SL_{\2\1}\,\mapsto\,\SQ_{\2\1}\,.
\end{equation}
\end{propn}
\begin{proof}
In order to calculate the explicit form of the length operator
in the representation associated to the fat graph \rf{tensorfig},
we may take Proposition \ref{Lreppropn} as a starting point.
It remains to calculate the change of representation induced
by the move $\omega_{\check\1\check\2}^{}$ which is diagrammatically 
represented as
\begin{equation}\label{fus1move}
\lower.9cm\hbox{\epsfig{figure=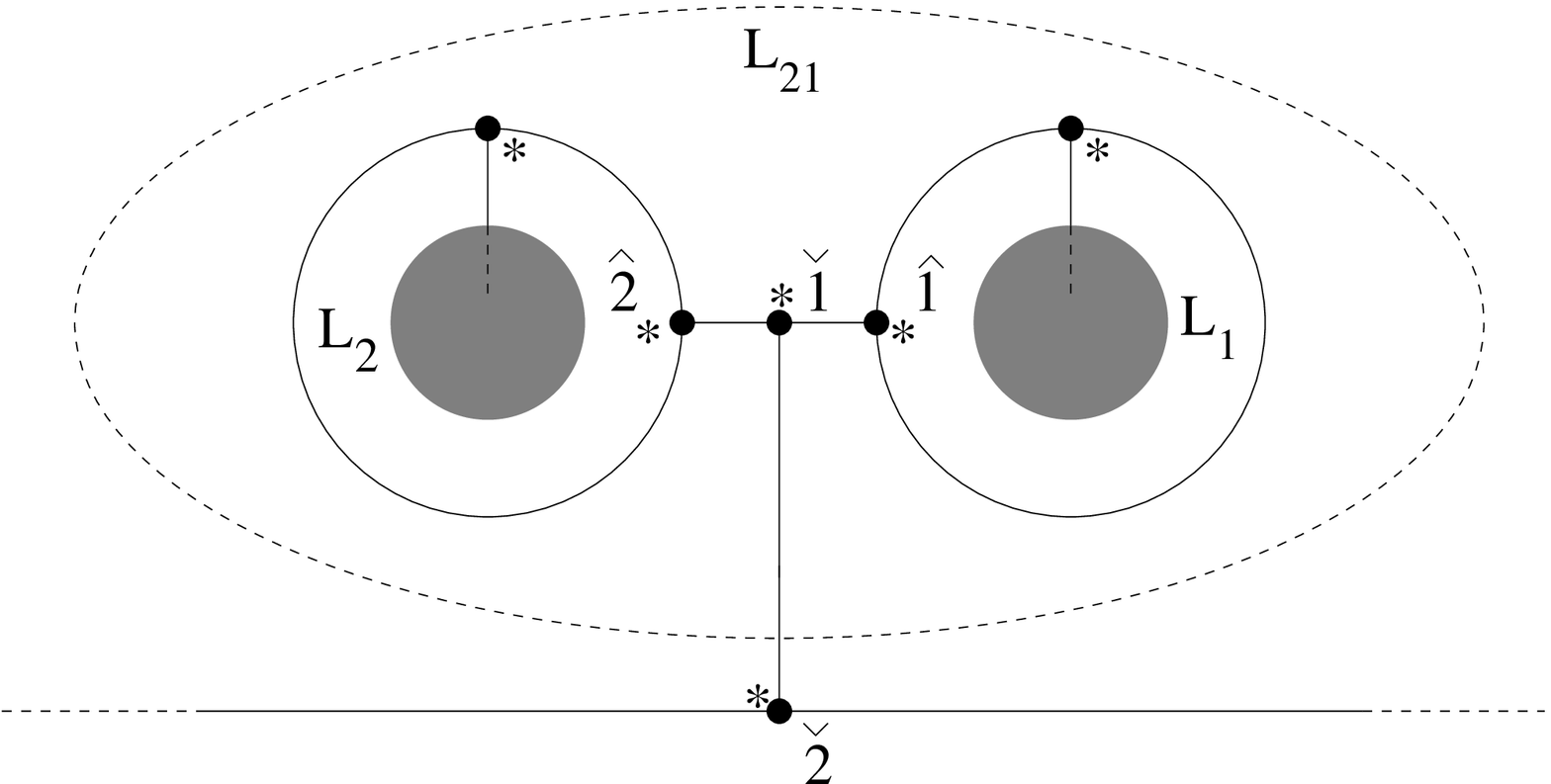,height=3cm}}
\quad \overset{\omega_{\check\1\check\2}^{}}{\longrightarrow}\quad
\lower.9cm\hbox{\epsfig{figure=caslength.eps,height=3cm}}
\end{equation}
This move is represented by the operator $\ST_{\check\1\check\2}$.
This calculation is obtained from the
one described above in \rf{adT12} by simple 
substitutions, resulting in the expression
\begin{align}
\SL_{\2\1}\,=\,
& \,e^{+2\pi b\sz_\1}
:e^{2\pi b(\sq_\1-\sq_\2+\spp_\2-\spp_\1)}
(2\cosh 2\pi b(\spp_\2-\sz_\2)+\SL_\2):  \notag \\
+&\,e^{-2\pi b\sz_\2}
:e^{-\pi b(\sq_\1-\sq_\2)}
(2\cosh 2\pi b(\spp_\1-\sz_\1)+\SL_\1):\notag\\
+& \,e^{2\pi b(\spp_\2-\sz_\2)}\SL_\1+e^{-2\pi b(\spp_\1-\sz_\1)}\SL_\2
+2\cos\pi b^2 e^{2\pi b(\spp_\2-\sz_\2-\spp_\1+\sz_1)}\,.
\end{align}
Setting the 
constraints to zero and comparing with 
\rf{Q21} yields the claimed result.
\end{proof}

\subsection{The Clebsch-Gordan maps}

Note that the operator $\SQ_\1'$ defined in \rf{SQ1def} 
essentially 
coincides with the particular 
representation of a length operator given 
in \rf{modLlem2}
after setting the constraints to zero.
It follows immediately from this observation 
that one may without loss 
of generality assume that the projection of the 
operator $\SC_v^{\CT}$ defined in \rf{SCTdef}
to $\sz_v^i=0$, $i=\1,\2$ coincides with the operator
$\nu^{\mss_{\2\1}}_{s_\2s_\1}\cdot\SS_\1\cdot\SC_\1$ which appears
as a building block in the construction of the 
Clebsch-Gordan maps $\SC_{s_\2s_\1}$
given in equation \rf{CGMconstr}. More precisely:
\begin{itemize}
\item The operator $\SC_\1$ corresponds to the following move:
\[
\lower.9cm\hbox{\epsfig{figure=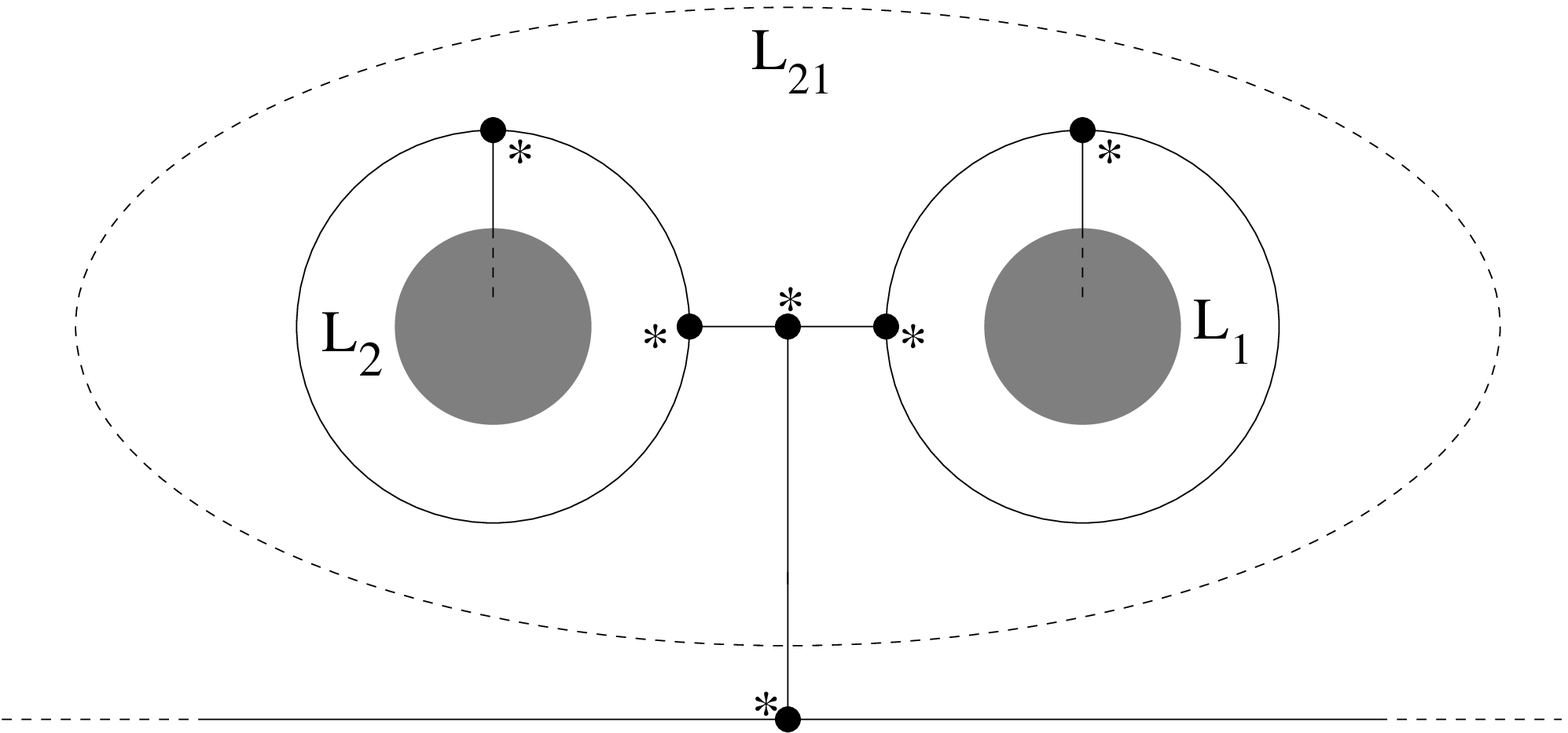,height=3cm}}
\quad \longrightarrow\quad
\lower.9cm\hbox{\epsfig{figure=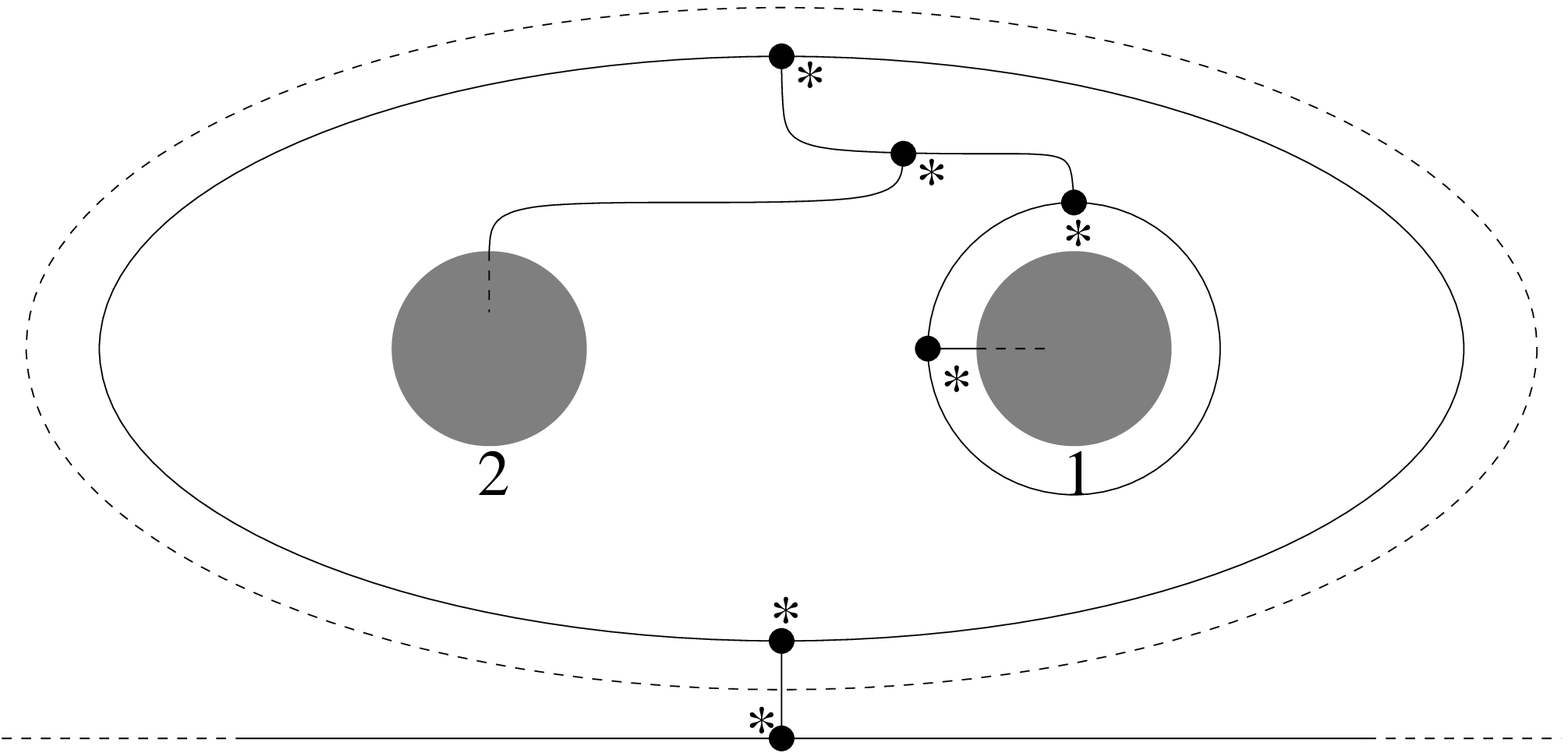,height=3cm}}
\] 
\item The operator $\SS_\1$ corresponds to the cutting operation:
\[
\lower.9cm\hbox{\epsfig{figure=leng4c.eps,height=3cm}}
\quad \longrightarrow\quad
\lower.9cm\hbox{\epsfig{figure=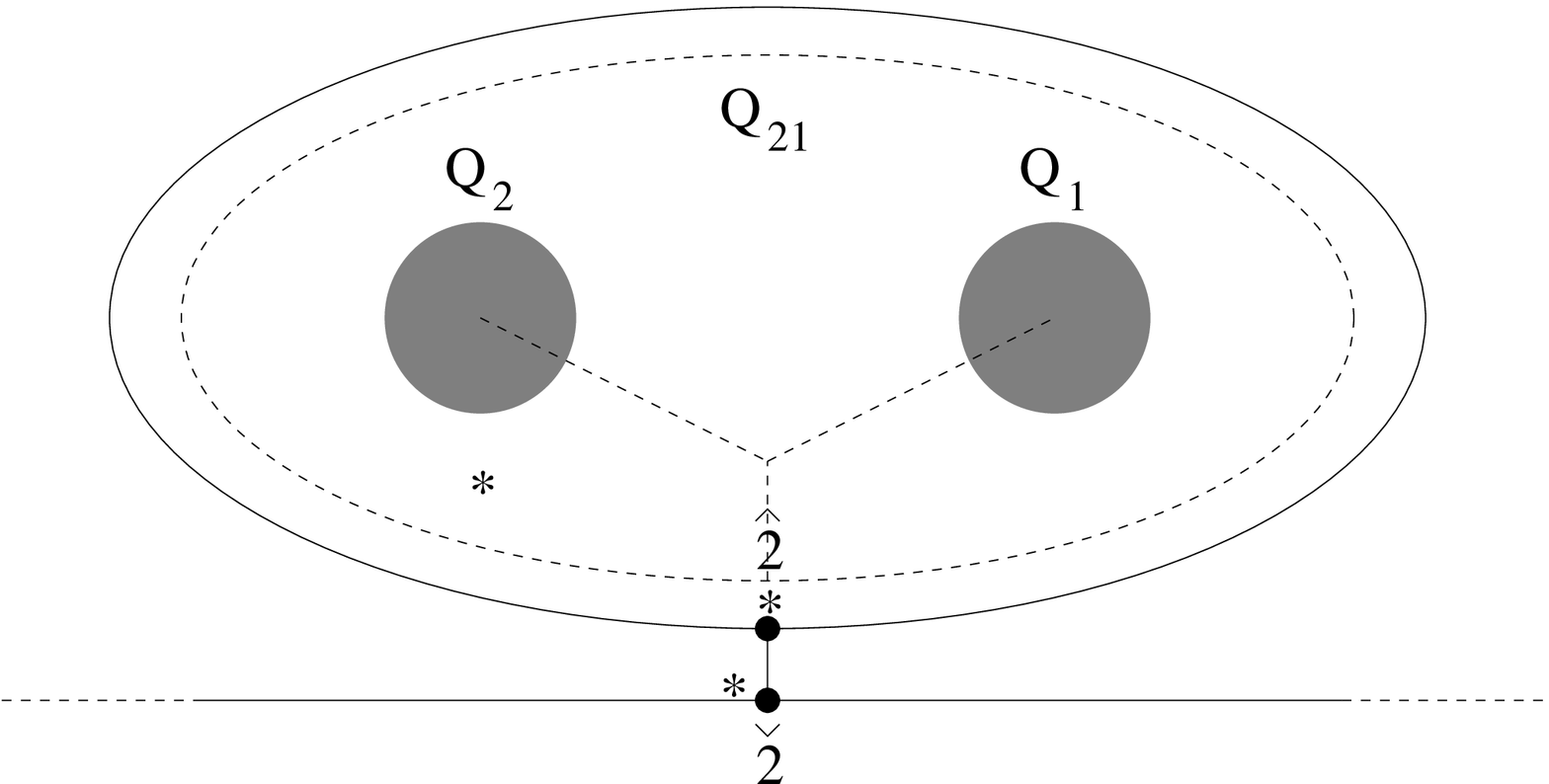,height=3cm}}
\] 
\end{itemize}
It remains to notice that the operator $ (\ST_{\1\2})^{-1}$ 
which appears in the
factorized representation of the full Clebsch-Gordan maps,
\begin{equation}\label{CGMconstr'}
\SC_{s_\2s_\1}:=\,\SS_\1\cdot\SC_\1\cdot (\ST_{\1\2})^{-1}\,,
\end{equation}
corresponds to the move 
$\omega_{\check\1\check\2}^{}$ depicted in \rf{fus1move}.

These observations may be summarized by saying 
that the Clebsch-Gordan maps of the
modular double represent the quantum 
cutting operation associated to the curve $c_{\2\1}$ 
surrounding holes $h_\2$ and $h_\1$ within the hybrid 
representation assigned to the graph in \rf{tensorfig}.

\subsection{The R-operator}\label{SSsec:Rop}

Let us consider the following move $r_{\2\1}$
\[
\lower.9cm\hbox{\epsfig{figure=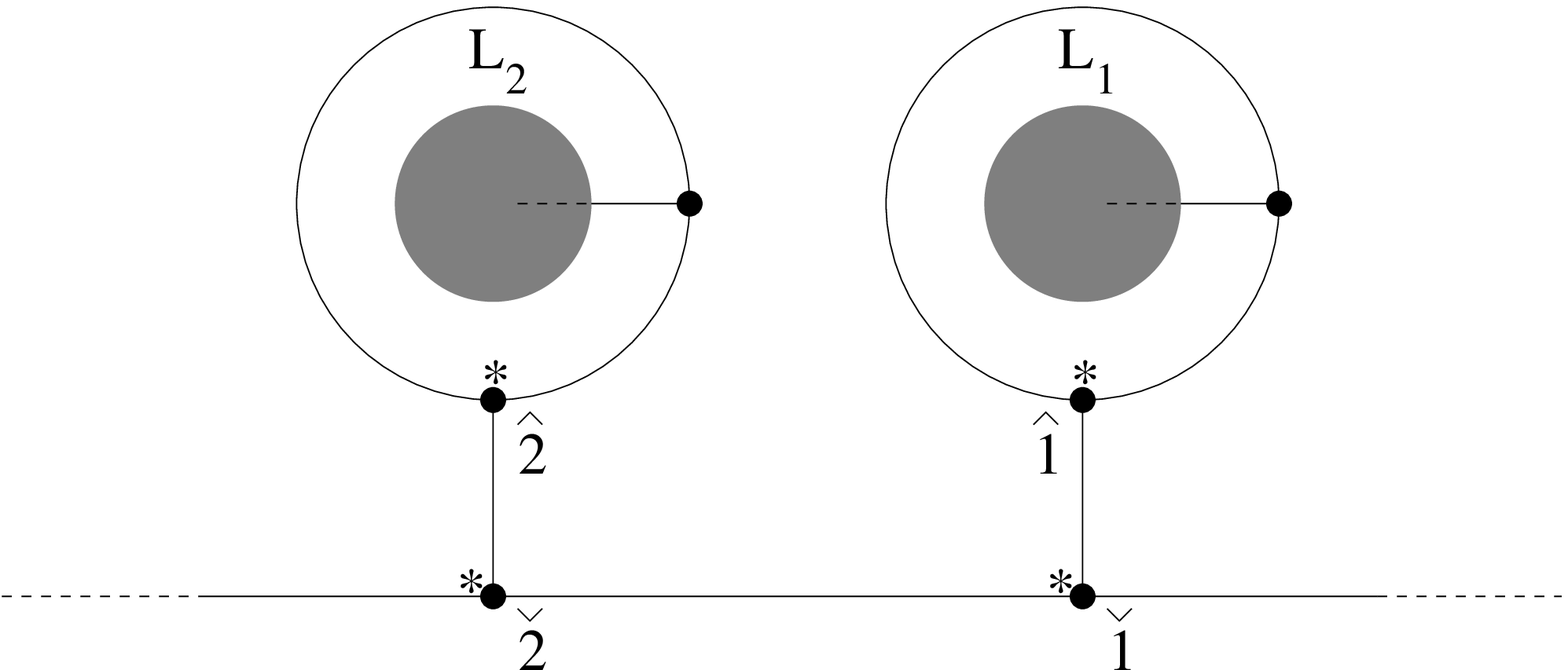,height=3cm}}
\quad \overset{r_{\2\1}}{\longrightarrow}\quad
\lower.9cm\hbox{\epsfig{figure=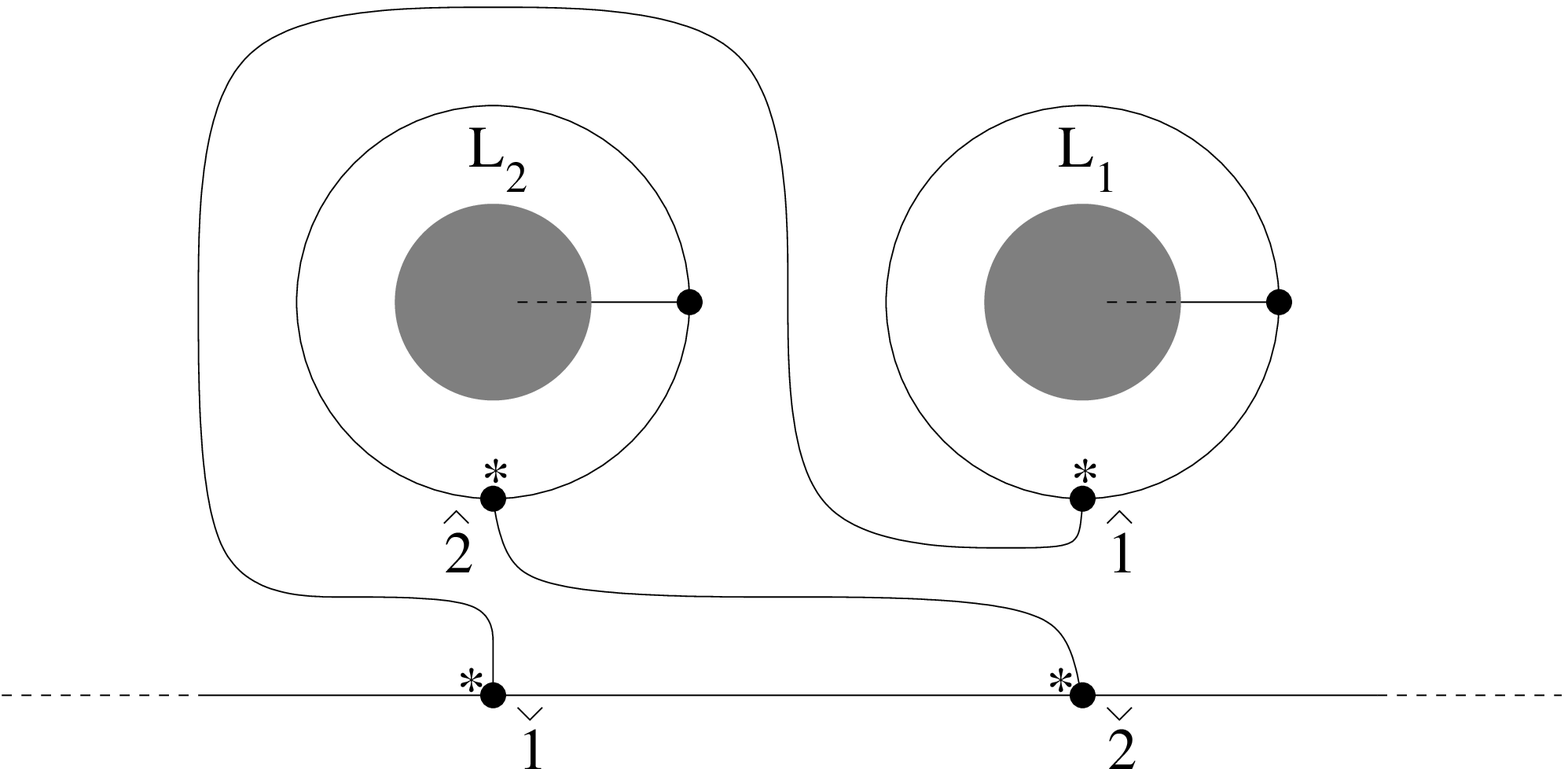,height=3cm}}
\] 
When this move is composed with the operation of exchanging all indices
$\1$ and $\2$ one gets the braid move representing the clockwise rotation
of holes $\2$ and $\1$ around each other until the positions have been
exchanged. The operator $\sr_{\2\1}$ which represents this move
within the quantum Teichm\"uller theory is easily found to be 
\begin{align}
\sr_{\2\1}=\SW_\2^{-1}\cdot\SA_{\ctwo}^{}\,
\ST_{\cone\ctwo}^{-1}\,\SA_{\ctwo}^{-1}
\cdot\SW_\2^{}\,.
\end{align}
This is easily seen by noting that
the operator $\SA_{\check\2} (\ST_{\check\1\check\2})^{-1}\SA_{\check\2}^{-1}$
represents the following move:
\[
\lower.9cm\hbox{\epsfig{figure=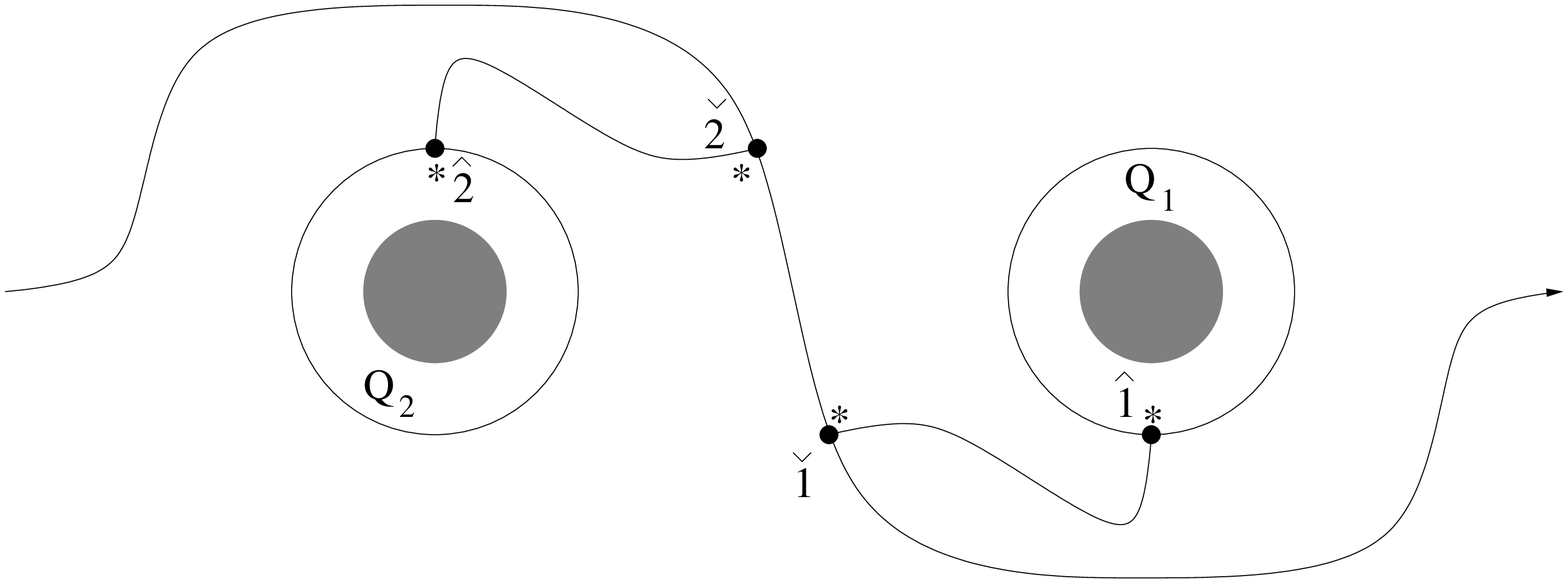,height=2.8cm}}
 \rightarrow
\lower.9cm\hbox{\epsfig{figure=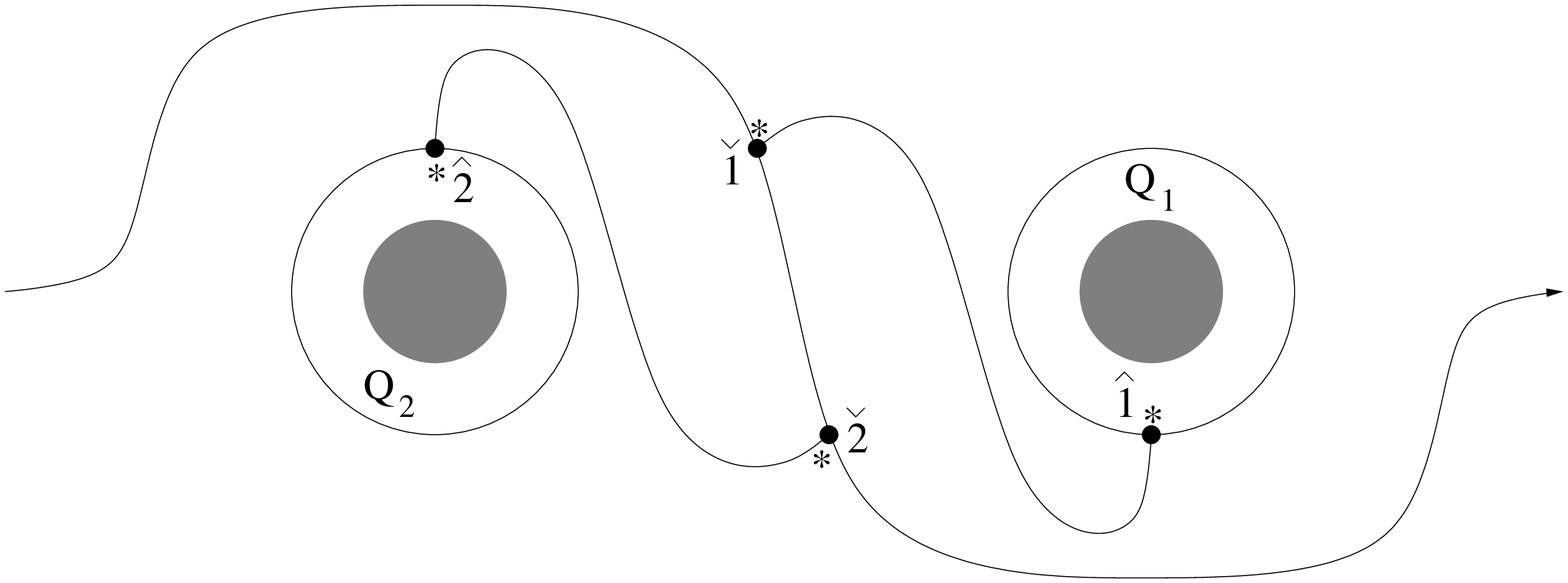,height=2.8cm}}
\] 

\begin{propn} The operator $\sr_{\2\1}$ gets mapped to the 
R-operator $\SR$.
\end{propn}
\begin{proof}
We have 
\begin{align*}
\SA_{\ctwo}^{}\,
\ST_{\cone\ctwo}^{-1}\,\SA_{\ctwo}^{-1}\,=\,e^{-\pi i\spp_{\cone}\spp_{\ctwo}}
\,e_b\big(\sq_{\cone}-\fr{1}{2}\spp_{\cone}+\sq_{\ctwo}-\fr{1}{2}
\spp_{\ctwo}\big)\,
e^{-\pi i\spp_{\cone}\spp_{\ctwo}}\,.
\end{align*}
This is identified as the operator
\[
q^{-\sh\ot\sh} \,
 E_b(\se\ot
 \se)\, q^{-\sh\ot\sh}\,.
\]
Using $\sa\sd[\SW_\2](\se)=\sf$ yields $\SR$, as claimed.
\end{proof}

\newpage

\section{A representation of the Moore-Seiberg groupoid in genus $0$}

\setcounter{equation}{0}
We will now present an important application of the results above. 
It was shown in \cite{T05} that the quantum Teichm\"uller theory 
defines a representation of the Moore-Seiberg groupoid, which is
important for understanding relations to conformal field theory.
The results of this paper will allow us to calculate the
operators which generate the representation of the Moore-Seiberg groupoid
explicitly.

\subsection{Pants decompositions}

Let us consider hyperbolic surfaces $C$ of 
genus $0$ with $n$ holes. We will assume that the holes
are represented by geodesics in the hyperbolic metric.
A pants decomposition of a hyperbolic surface 
$C$ is defined by a cut system which in this context
may be represented by a collection 
$\CC=\{\ga_1,\dots,\ga_{n-3}\}$ of non-intersecting simple closed geodesics
on $C$.
The complement 
$C\setminus\CC$ is a disjoint union $\bigsqcup_{v}C_{0,3}^v$ of
three-holed spheres (trinions).
One may reconstruct $C$ from the resulting collection 
of trinions by pairwise gluing of boundary 
components. 

For given lengths of the three boundary geodesics there is a unique
hyperbolic metric on each trinion $C_{0,3}^v$. 
Introducing a numbering of the boundary geodesics $\ga_i(v)$, $i=1,2,3$,
one gets three distinguished geodesic arcs $\ga_{ij}(v)$, $i,j=1,2,3$
which connect the boundary components pairwise.  
Up to homotopy there are exactly 
two tri-valent graphs $\Ga_{\pm}^v$ on $C_{0,3}^v$ 
that do not intersect any $\ga_{ij}(v)$.
We may assume that these graphs glue to two connected graphs 
$\Ga_\pm$ on $C$. 
The pair of data $\si=(\CC_\si,\Ga_\si)$, where $\Ga_\si$ is one
of the MS graphs $\Ga_{\pm}$ associated 
to a hyperbolic pants decomposition, can be used to distinguish
different pants decompositions in hyperbolic geometry.
The role of the graph $\Ga_\si$ is to distinguish pants decompositions
obtained from each other by means of Dehn twists, rotations of 
one boundary component by $2\pi$ before gluing.

\subsection{The Moore-Seiberg groupoid}

Let us note \cite{MS,BK} that
any two 
different pants decompositions 
$\si_2$, $\si_1$ can be connected by a sequence of
elementary moves localized in subsurfaces of $C_{g,n}$ of type
$C_{0,3}$, $C_{0,4}$. These will be called the 
$B$, $S$ and $F$, respectively.
Graphical representations for the elementary 
moves $B$, $Z$,  and $F$ are given in 
Figures \ref{bmove},  \ref{zmove} and \ref{fmove}, respectively.
\begin{figure}[t]
\epsfxsize9cm
\centerline{\epsfbox{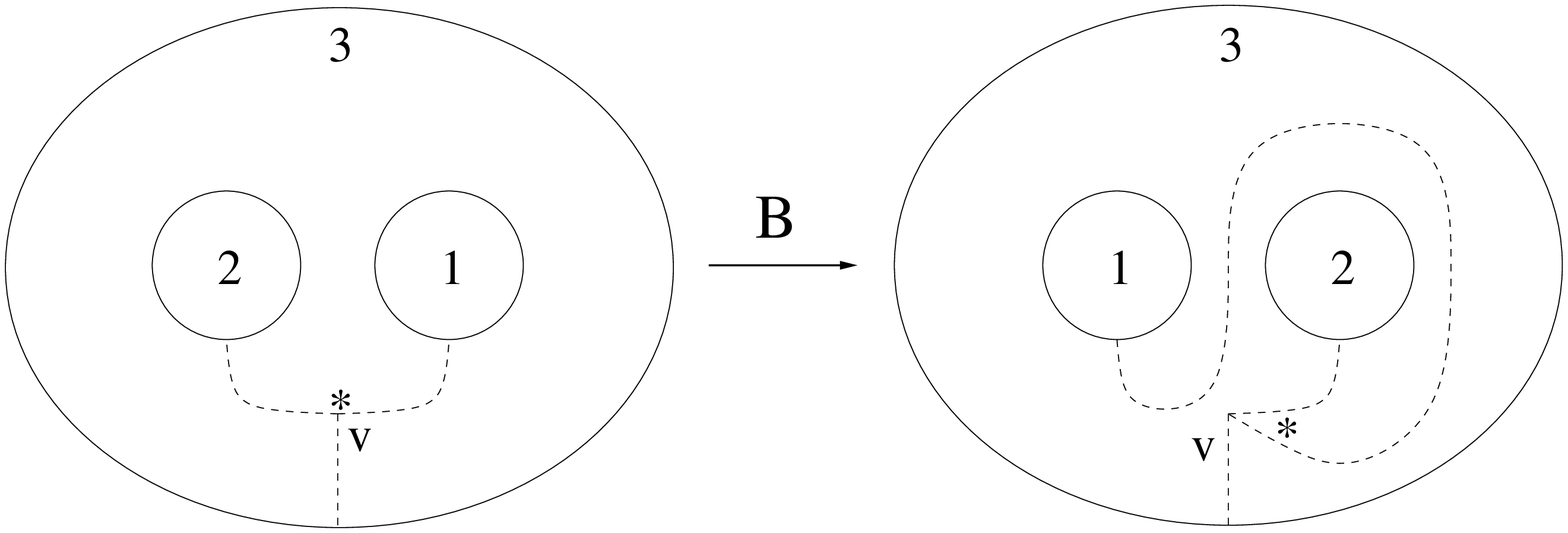}}
\caption{The move $B_v:\si\ra B_v\si$}\label{bmove}\vspace{-.03cm}
\end{figure}
\begin{figure}[t]
\epsfxsize9cm
\centerline{\epsfbox{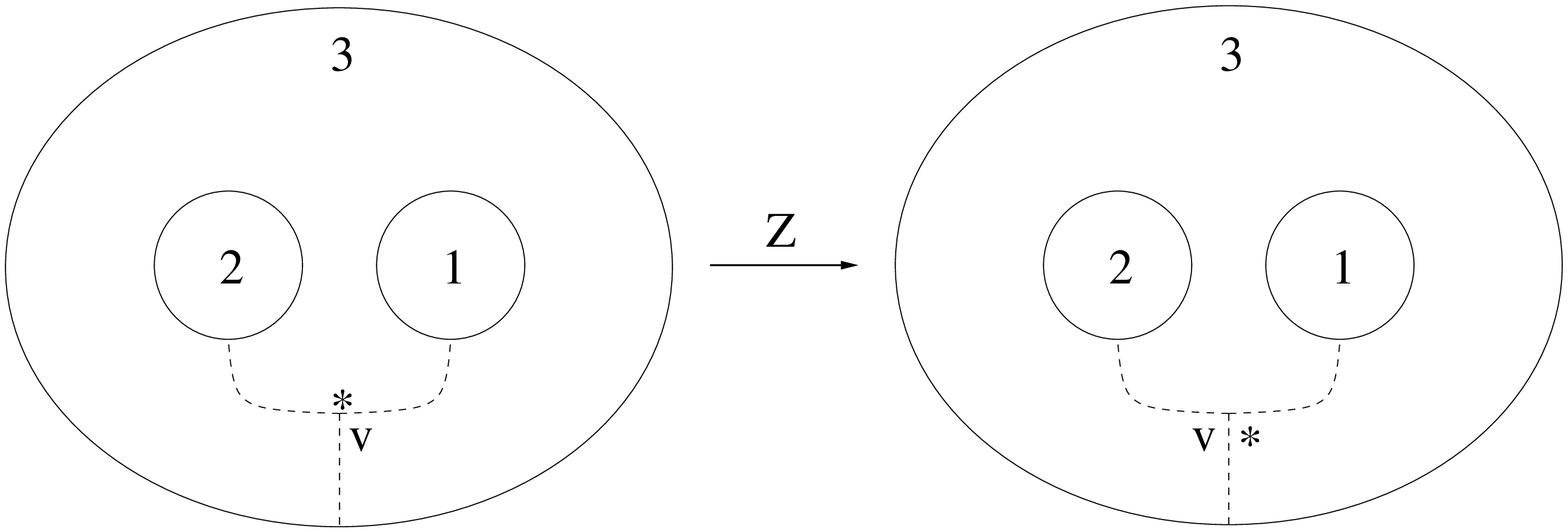}}
\caption{The move $Z_v:\si\ra Z_v\si$}\label{zmove}
\end{figure}
\begin{figure}[t]
\epsfxsize4.5cm
\centerline{\epsfbox{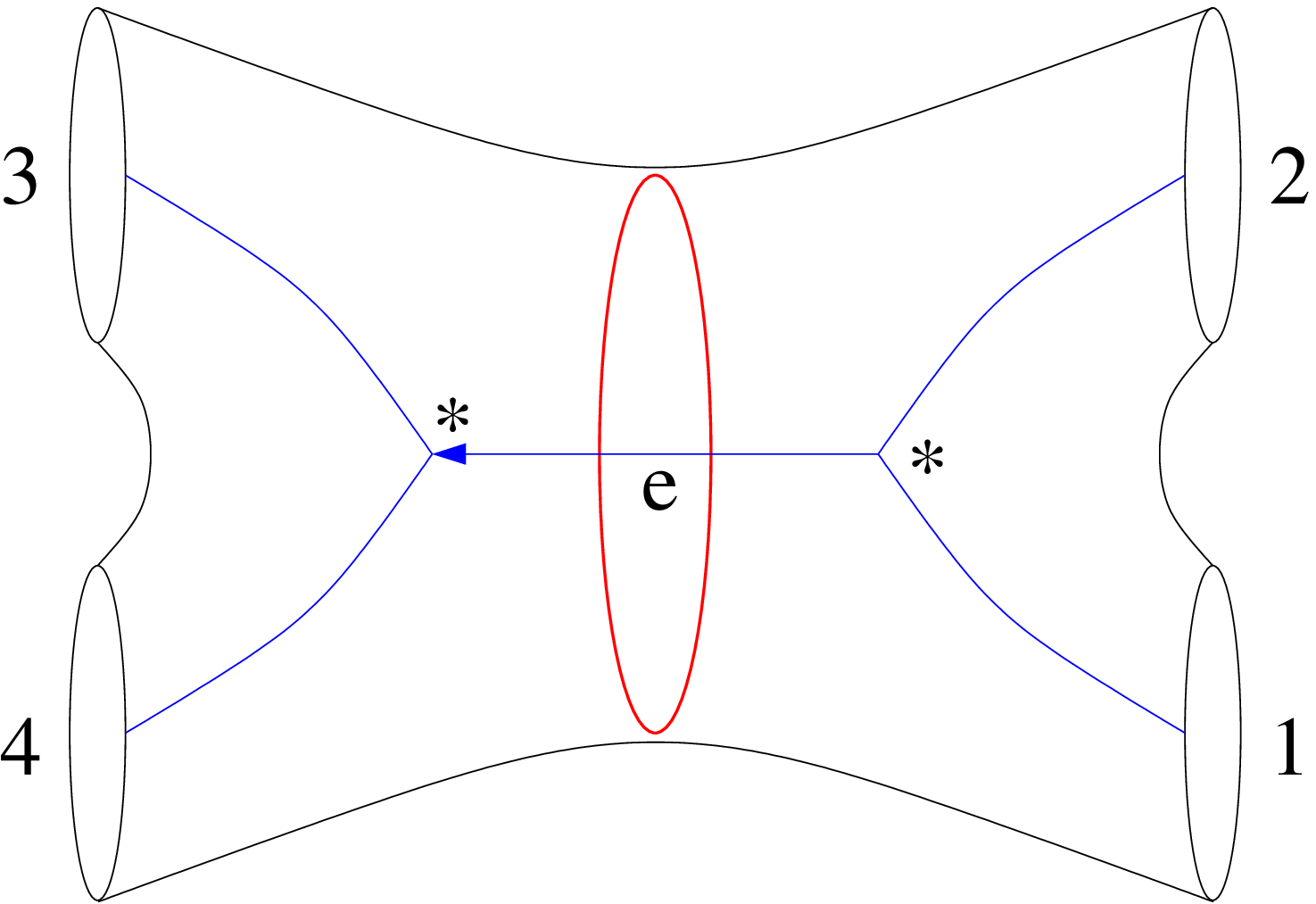}\hspace{.5cm}
$\displaystyle\overset{F_{e}}{\Longrightarrow}$\hspace{.5cm}
\epsfxsize4.5cm\epsfbox{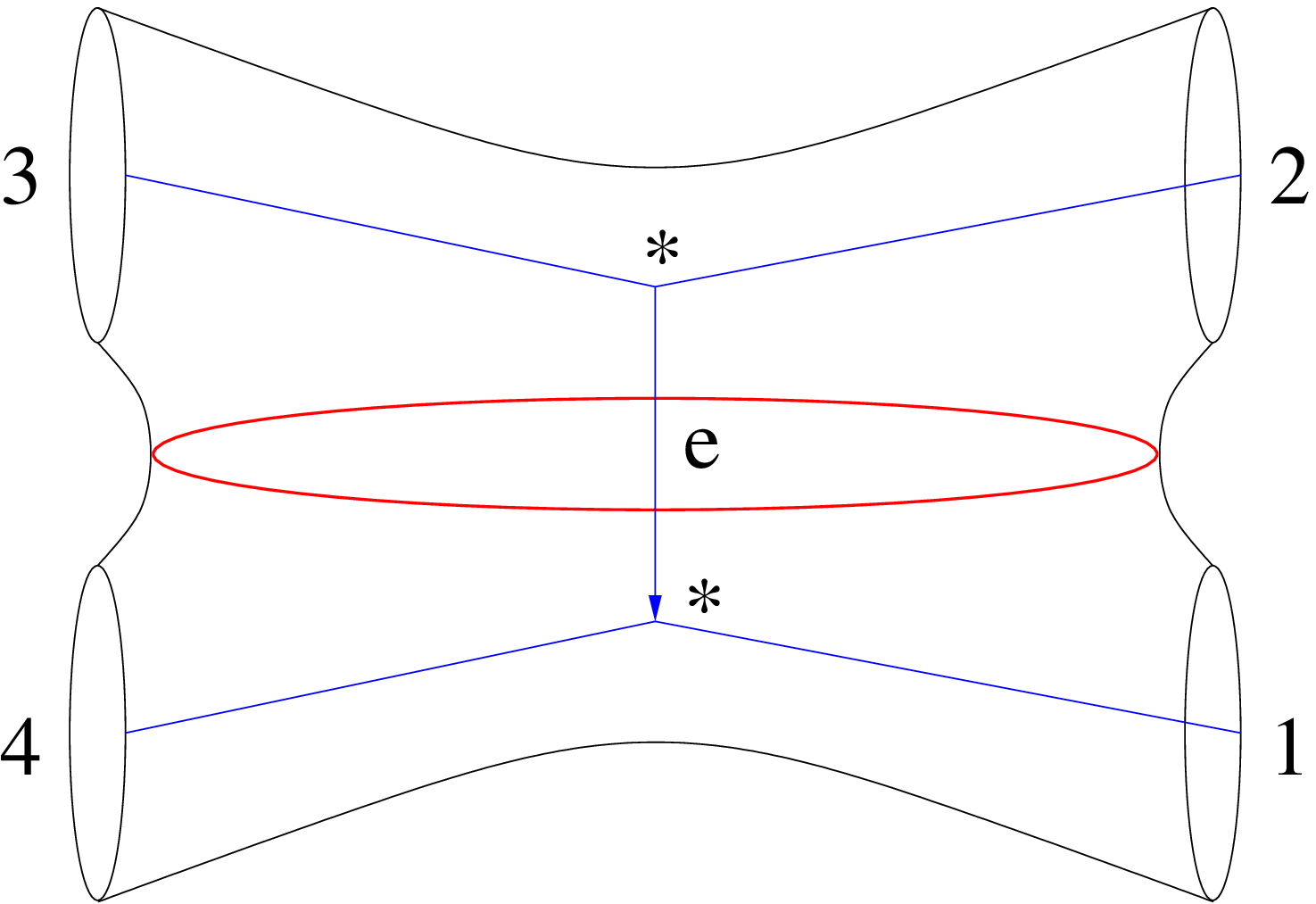}}
\caption{The move $F_e:\si_s\ra \si_t\equiv 
F_e\si_s$}\label{fmove}\vspace{.3cm}
\end{figure}

One may formalize
the resulting structure by introducing a two-dimensional 
CW complex $\CM(C)$ with set of vertices $\CM_\0(C)$
given by the
pants decompositions $\si$, and a 
set of edges $\CM_\1(C)$ associated to the elementary
moves. 

The Moore-Seiberg groupoid is defined to be the path groupoid of 
 $\CM(C)$. It can be described in terms of generators and relations,
the generators being associated with the edges in  $\CM_\1(C)$, 
and the relations associated with the faces of $\CM(C)$. 
The classification of the relations was first presented in \cite{MS},
and rigorous mathematical proofs have been presented in \cite{FG,BK}.
The relations are all represented by sequences of moves localized
in subsurfaces $C_{g,n}$ with genus $g=0$ and $n=3,4,5$ punctures. 
Graphical representations of the
relations can be found in \cite{MS,FG,BK}.

\subsection{Representation of the Moore-Seiberg groupoid}

A representation of the Moore-Seiberg groupoid can be obtained 
from the quantum Teichm\"uller theory as follows \cite{T05}. 

\subsubsection{}

The starting point is a construction which produces a fat graph 
$\vf_\si$ associated to pants decompositions $\si$. 
This 
construction depends on a choice of decoration for a pants
decomposition $\si$ which is the choice of a distinguished 
boundary component for each trinion. The distinguished 
boundary component will be called outgoing, the other
boundary component incoming.
The decoration is 
indicated by an asterisk in the Figures \ref{bmove}, 
\ref{zmove} and \ref{fmove}. We identify the Z-move as the 
elementary change of decorations. In the following we will
use the notation $\si$ for {\it decorated} pants decompositions.

The construction described in \cite{T05} 
can be applied for a subset of decorated 
pants decompositions which is defined by the condition that
outgoing boundary components are never glued to another.
Such pants decompositions will be called admissible.

There is
a natural fat graph $\vf_\si$ associated to $\si$ which is 
defined by 
gluing the following pieces:
\begin{itemize}
\item For each curve $c$ separating two 
incoming boundary component let us insert an annulus $A_c$,
with fat graph locally of the form depicted in Figure \rf{annfig}.
\item Trinions: See Figure \rf{trinfig}.
\item Holes: See Figure \rf{holefig}.
\end{itemize}
Gluing these pieces in the obvious way will produce the
connected graph $\vf_\si$ associated to the 
Moore-Seiberg graph $\si$ we started from.

\subsubsection{} 

Following \cite{T05}, we will now describe how to map a maximal commuting 
family of length operators to diagonal form. 
We will start from the hybrid representation described above in which the
length operators and constraints associated to the incoming holes
are diagonal. Recall that states are represented by wave-functions
$\psi(p;s_b,z_b)$ in such a representation, where $p:\tilde{\vf}_0\mapsto\BR$,
and $\tilde{\vf}_0$ is the subset of $\vf_0$ that does not 
contain $\hat{h}$ nor $h'$ for any incoming hole $h$.
A maximal commuting 
family of length operators is associated to the cut system
$\CC_\si$ of a pants decompostion.

To each vertex $v$ of $\Ga_\si$ 
assign the length operator $\SL^\2_v$ and $\SL^\1_v$
to the incoming and $\SL_v$ to the outgoing boundary components of the 
pair of pants $P_v$ containing $v$.
The main tool is the operator $\SC_v^{\CT}$ which maps $\SL_v$ to a simple 
standard form,
\begin{equation}
\SC_v^{\CT}\cdot\SL_v\cdot(\SC_v^{\CT})^{-1}\,=\,2\cosh 2\pi b \spp_v+e^{-2\pi b \sq_v}\,.
\end{equation}
Such an operator can be constructed explicitly as
\begin{equation}\label{SCTdef}
\SC_v^{\CT}(\mss_v^\2,\mss_v^\1):=\,e^{-2\pi i s_\2\sq_v} \,\frac{e_b(\mss^\1_v+\spp_v)}{e_b(\mss^\1_v-\spp_v)} \,
e^{-2\pi i \mss^\1_v\spp_v}\,(e_b(\sq_v-\mss^\2_v))^{-1}\,
e^{-2\pi i(\sz_v^\2\spp_v+\sz_v^\1\sq_v)},
\end{equation}
where $\mss_v^\imath$, $\imath=\1,\2$ are the 
positive self-adjoint operators defined by 
$\SL_v^\imath=2\cosh2\pi b\mss_v^{\imath}$, and $\sz^\2_v$, $\sz^\1_v$
are the constraints associated to the incoming boundary components
of $P_v$. The operator $\SC_v$ is clearly related to the operator
$\SC_\1$ that appeared as a key ingredient of the Clebsch-Gordan maps
for the modular double in the previous part.

The map to the length representation is then constructed 
as the ordered
product over the operators $\SC_v$, $v\in\vf_{\si,0}$. 
The resulting operator may be represented as the following explicit 
integral transformation: Let $s$ be 
the assignment $v:\tilde{\vf}_0\mapsto\BR_+$. Define
\begin{align}
\Phi(s,z_b)\,=\,\int\limits_{\BR^{n-3}} 
\Big({\prod}_{v\in\tilde\vf_0} dp_v\; 
K_{s^{\2}_vs^{\1}_v}^{z^{\2}_vz^{\1}_v}(s_v,p_v)\Big)\,
\psi(p;s_b,z_b)\,.
\end{align}
The kernel $K_{s_{\2}s_{\1}}^{z_{\2}z_{\1}}(s,p)$ has the explicit form
\begin{align}
K_{s_{\2}s_{\1}}^{z_{\2}z_{\1}}(s,p)& \,=\,
\langle\,s\,|\,\SC^{\CT} \,|\,p\,\rangle \,\\
&\,=\,
\langle\,s\,|\,e^{-2\pi i s_\2\sq} \,\frac{e_b(s_\1+\spp)}{e_b(s_\1-\spp)} \,
e^{-2\pi i s_\1\spp}\,(e_b(\sq-s_\2))^{-1}\,
e^{-2\pi i(z_\2\spp+z_\1\sq)}
\,|\,p\,\rangle \notag \\
& \,=\,\int_\BR dp'\;\langle\,s\,
|\,p'\,\rangle\,\frac{w_b(s_\1-p'-s_\2)}{w_b(s_\1+p'+s_\2)}\,
\langle\,p'\,|\,e^{-2\pi i s_\2\sq_\1}(e_b(\sq-s_\2))^{-1}
e^{-2\pi i(z_\2\spp+z_\1\sq)}
\,|\,p\,\rangle
\notag\\
& \,=\,\zeta_{\0}\int_\BR dp'\;
{e^{-2\pi i(s_\2-z_b)(s_\2+p'-p+z_\1)}}{e_b(p-z_\1-s_\2-p'+z_b)}
\notag\\
&\hspace{4cm}\times\frac{w_b(s_\1-p'-s_\2)}{w_b(s_\1+p'+s_\2)}
\frac{w_b(s+p'-z_b)}{w_b(s-p'+z_b)}e^{-2\pi i z_\2(2p-z_\1)}
 \,.
\notag\end{align}
In the last step we have used the complex conjugate of equation
\rf{matel2} in Appendix \ref{CGapp} below.

\subsubsection{}

The construction above canonically defines operators $\SU_{\si_\2\si_\1}$ 
intertwining between the representations 
$\pi_{\si_\1}$ and $\pi_{\si_\2}$ as
\begin{equation}\label{SWtoSU}
\SU_{\si_\2\si_\1}:=\SC_{\si_\2}^{}\cdot\SW_{\vf_{\si_\2}\vf_{\si_\1}}\cdot\SC_{\si_\1}^{-1}\,,
\end{equation}
where $\SW_{\vf_{\si_\2}\vf_{\si_\1}}$ is any operator representing the
move $[\vf_{\si_\2},\vf_{\si_\1}]$ between the fat graph associated to 
$\si_\1$ and $\si_\2$, respectively. In this way one defines
operators  $\hat\SB_v$, $\hat\SA_e$, and $\hat\SZ_v$  associated to the
elementary moves $B_v$, $F_e$ and $Z_v$ 
between different MS-graphs, respectively. These operators
satisfy operatorial versions of the Moore-Seiberg consistency conditions
\cite{T05}, which follow immediately from the relations of the 
Ptolemy groupoid using \rf{SWtoSU}.

One should note that the definition \rf{SWtoSU} can be applied only
if the decorated pants 
decompositions $\si_\2$ and $\si_\1$ are both admissible.
However, this restriction will quickly turn out to be inessential.
To begin with, let us note that the definition
\rf{SWtoSU} can indeed be applied to all operators that appear
in the relations of the Moore-Seiberg groupoid. A quick inspection 
of the relations listed in \cite{BK,T05} shows that all the 
decorated pants decompositions appearing therein are 
admissible.

The wave-function $\phi(s):=\Phi(s,z_b)|_{z_b=0}$ represents
the projection of the wave-function $\Phi$ 
to the subspace defined by vanishing constraints $\sz_{\vf,c}$.
The operators $\SU_{\si_\2\si_\1}$ commute with the constraints
$\sz_{\vf,c}$, in the sense
that
\begin{equation}
\SU_{\si_\2\si_\1}\cdot \sz_{\vf_{\si_\1},c}\,=\,
\sz_{\vf_{\si_\2},c}\cdot\SU_{\si_\2\si_\1}\,.
\end{equation}
The projections of the operators  $\hat\SB_v$, $\hat\SA_e$, and $\hat\SZ_v$
define operators $\SB_v$, $\SA_e$, and $\SZ_v$
which satisfy the relations of 
the Moore-Seiberg groupoid up to possible projective phases.

\subsubsection{}\label{SSec:locality}

The representation of the Moore-Seiberg groupoid defined via
\rf{SWtoSU} has nice locality properties in the sense that
the operator representing a move localized in a subsurface $C'$ of $C$
will only act on the variables $s_e$ associated to the edges $e$
of $\Ga_\si$ that have nontrivial intersection with $C'$.
 In order to make this precise and easily visible in the notations,
let us introduce the one-dimensional
Hilbert space $\CH_{s_2s_1}^{s_3}$ associated to  a
three-holed sphere $C_{0,3}$ with parameters $s_i$, $i=1,2,3$ 
associated to the boundary components 
according the numbering convention indicated on the left of 
Figure \ref{bmove}. Note that edges $e$ of the MS graph determine
curves $c_e$ in the cut system. The eigenvalues $L_e$ of the 
operators $\SL_{c_e}$ will be parameterized, as before, in terms
of real numbers $s_e$ such that
$L_e=2\cosh 2\pi bs_e$.
To a pants decomposition we may then associate 
the direct integral
of Hilbert spaces
\begin{equation}\label{Hilbertsp}
\CH_\si\,\simeq\,\int^{\oplus}_{\BR_+^h}
\prod_{e\in\si_1}
d\mu(s_e) \;\bigotimes_{v\in\si_0} \CH_{s_2(v)s_1(v)}^{s_3(v)}\,.
\end{equation}
We denoted the set of internal edges of the MS graph $\si$ by $\si_1$,
and the set of vertices by $\si_0$.

As a useful notation let us introduce ``basis vectors'' $\langle \,s\,|$ 
for $\CH_\si$, 
more precisely distributions on dense subspaces of $\CH_\si$ such that
the wave-function $\psi(s)$ of a state $|\,\psi\,\rangle$ 
is represented as $\psi(s)\,=\,\langle \,s\,|\,\psi\,\rangle$.
Representing  $\CH_\si$ as in \rf{Hilbertsp} 
one may identify  
\begin{equation}\label{basis}
\langle \,s\,|\,\simeq\,
\bigotimes_{v\in\si_0} v_{s_2(v)s_1(v)}^{s_3(v)}\,,
\end{equation}
where $v_{s_2s_1}^{s_3}$ is understood as an element of the 
dual $\big(\CH_{s_2,s_1}^{s_3}\big)^{t}$ of the 
Hilbert space $\CH_{s_2s_1}^{s_3}$.

The operators $\SB_v$, $\SZ_v$ and $\SA_e$ will be 
represented in the following 
form:
Given functions $B_{s_\2s_\1}^{s_\3}$ and $Z_{s_\2s_\1}^{s_\3}$ of 
three variables one may define multiplication operators $\SB$ and
$\SZ$ as
\begin{subequations}\label{BZFdef}
\begin{align}
& \SB\cdot v_{s_\2s_\1}^{s_\3}\,=\,B_{s_\2s_\1}^{s_\3}
v_{s_\2s_\1}^{s_\3}\,,\\
& \SZ\cdot v_{s_\2s_\1}^{s_\3}\,=\,Z_{s_\2s_\1}^{s_\3}
v_{s_\1s_\3}^{s_\2}\,.
\end{align}
For each vertex $v$ of $\Ga_{\si}$ the representation
\rf{Hilbertsp} suggests an obvious way to lift $\SB$ and
$\SZ$ to operators $\SB_v$ and $\SZ_v$ mapping $\CH_{\si}$ to 
$\CH_{B_v \si}$ and $\CH_{Z_v \si}$, respectively.

Let furthermore $\si_s$ and $\si_t$ be the pants decompositions
of $C_{0,4}$ depicted on the left and right of Figure 
\ref{fmove} respectively. The 
operators $\SA:\CH_{\si_s}\ra\CH_{\si_t}$ can be represented as
\begin{align}
& \SA\cdot v^{s_\4}_{s_\3s_{\2\1}}\ot 
v^{s_{\2\1}}_{s_\2s_\1} \;\,=\,
\int^{\oplus}_\BS d\mu(s_{\3\2}) \;
\Fus{s_\1}{s_\2}{s_\3}{s_\4}{s_{\2\1}}{s_{\3\2}}\;
v^{s_\4}_{s_{\3\2}s_\1}\ot 
v^{s_{\3\2}}_{s_\3s_\2}\,.
\end{align}
\end{subequations}
For each pants decomposition $\si$ and each edge $e$ of $\Ga_\si$ 
one may then use $\SA$ to
define operators $\SA_e:\CH_{\si}
\ra\CH_{F_e \si}$.

\subsubsection{}

Indeed, it is easy to see that the operators $\SB_v$, $\SA_e$, and $\SZ_v$ 
defined via \rf{SWtoSU} are of the form described in  Subsection
\ref{SSec:locality}.
The fact that $\SB_v$ and $\SZ_v$ act as multiplication
operators on $\CH_\si$ follows from the observation that
these operators, as can easily be checked, commute with the 
length operators associated to the boundary components of a trinion.
The form claimed for $\SA_e$ follows from the fact that 
this operator commutes with length operators
associated to the boundary components of the four-holed sphere $C_{0,4}$
containing $e$.
 
It is now clear how the representation of the Moore-Seiberg groupoid 
is extended from pairs $(\si_\2,\si_\1)$ of admissible pants
decompositions to all pairs $(\si_\2,\si_\1)$ of pants decompositions.

\subsection{Explicit form of the generators}

We now come to one of the main applications of the connection
between the modular double and the quantum Teichm\"uller theory:
It will allow us to calculate the explicit representation 
of the operators $\SB_v$, $\SA_e$, and $\SZ_v$. As explained 
above, it suffices to find the corresponding 
operators $\SB$, $\SZ$ and $\SA$ 
which take the form specified in equations \rf{BZFdef} above.
The result will be
\begin{subequations}
\begin{align}
&B_{s_\2s_\1}^{s_\3}\,=\,e^{\pi i (s_\3^2-s_\2^2-s_\1^2+c_b^2)}\,,\\
&\Fus{s_\1}{s_\2}{s_\3}{s_\4}{s_{\2\1}}{s_{\3\2}}=
\fus{s_\1}{s_\2}{s_\3}{s_\4}{s_{\2\1}}{s_{\3\2}}\,,\\
&Z_{s_\2s_\1}^{s_\3}\,=\,1\,.
\end{align}
\end{subequations}
This result will be a rather easy consequence of the
the relations between the modular double and quantum Teichm\"uller 
theory observed above.

\subsubsection{The operator $\SA$}

It follows from the relation  between Clebsch-Gordan maps and 
operators $\SC_v$ observed above that
the operators $\SA$ can be expressed in terms of the operators 
$\SA^{\rm T}(s_\3,s_{\2},s_\1)$ defined as
\begin{equation}\label{SFprform}
\SA^{\rm T}(s_\3,s_{\2},s_\1)\,=\,\SC_\2(s_\3,\mss_{\2\1})\SC_{\1}(s_\2,s_\1)
\cdot\ST_{{\cone}{\ctwo}}^{-1}
\cdot\big[\SC_\1(\mss_{\3\2},s_{\1})\SC_{\2}(s_\3,s_\2)\big]^{-1}\,.
\end{equation}
We have given a diagrammatic
representation for the the operator $\SA^{\rm T}$ in Figure \ref{fpr}. 
\begin{figure}[t]
\epsfxsize11cm
\centerline{\epsfbox{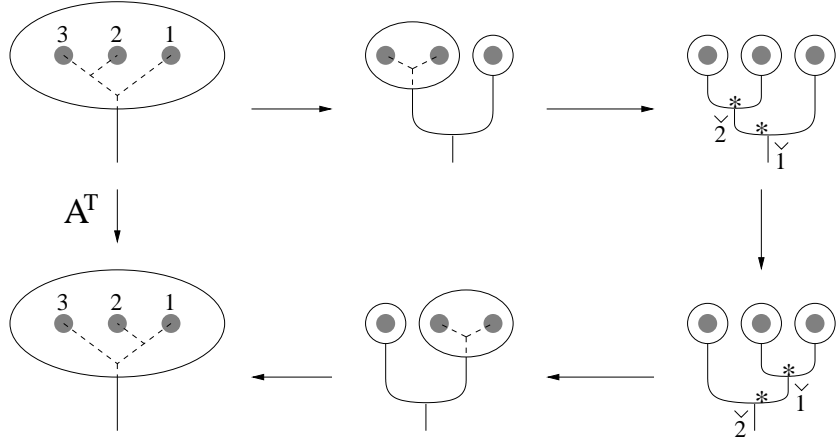}}
\caption{Graphical representation for Equation \ref{SFprform}.}
\label{fpr}\end{figure}

By projecting 
$\SA^{\rm T}(s_\3,s_{\2},s_\1)$
to vanishing constraints one gets
an operator $\SA^{\rm T}_{s_\3s_\2s_\1}:\CH_{\si_s}\ra \CH_{\si_t}$.
It is not hard to see that
we have
$\SA^{\rm T}_{s_\3s_\2s_\1}=\SA_{s_\3s_\2s_\1}$, where 
$\SA_{s_\3s_\2s_\1}$  is the operator defined in \rf{fusdef}.
Indeed, we may express $\SA_{s_\3s_\2s_\1}$ in the following form
\[
\SA_{s_\3s_\2s_\1}\,=\, 
\SC_\2(s_\3,\mss_{\2\1})\SC_{\1}(s_\2,s_\1)
\cdot\ST_{{\2}{\3}}^{-1}\,\ST_{{\1}{\2}}^{-1}\,\ST_{{\2}{\3}}^{}\,
\ST_{{\1}{\3}}^{}
\cdot\big[\SC_\1(\mss_{\3\2},s_{\1})\SC_{\2}(s_\3,s_\2)\big]^{-1}\,.
\]
By using \rf{pentrel} one easily simplifies this expression to the
form given in \rf{SFprform}.
This result allows us to conclude that the matrix elements of the 
fusion operator $\SA$ are given by the 
b-6j symbols
$\fus{s_\1}{s_\2}{s_\3}{s_\4}{s_{\2\1}}{s_{\3\2}}$.

\subsubsection{The operator $\SB$}

It follows from our main result in Subsection \ref{SSsec:Rop} 
that the operator $\SB_v$ is represented by $\SP_{\2\1}\sr_{\2\1}$,
where $\SP_{\2\1}$ is the permutation operator. It was shown in \cite{BT1}
that the Clebsch-Gordan maps diagonalize this operator, with eigenvalue
being $e^{\pi i (s_\3^2-s_\2^2-s_\1^2+c_b^2)}$.

\subsubsection{The operator $\SZ$}

Let us, on the one hand, consider the relation in  the Moore-Seiberg groupoid
drawn in Figure \ref{MSrel}.

\begin{figure}[htb]
\epsfxsize9cm
\centerline{\epsfbox{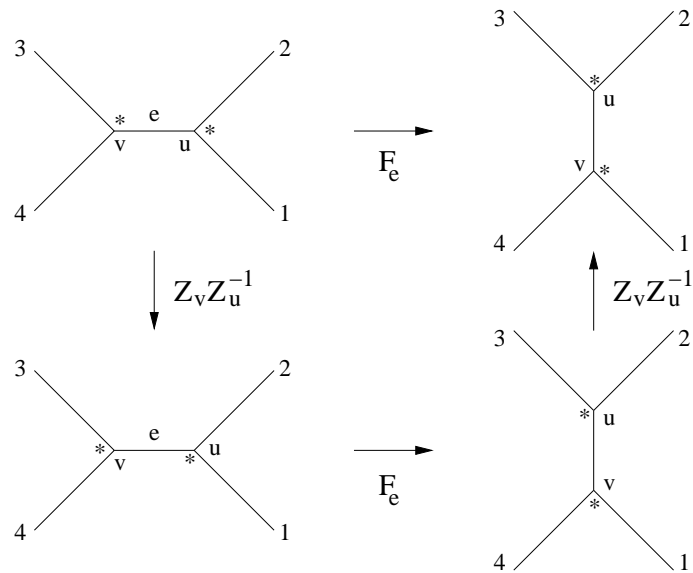}}
\caption{A relation in the Moore-Seiberg groupoid.}
\label{MSrel}\end{figure}

This relation implies the symmetry relation
\begin{equation}
\Fus{s_\1}{s_\2}{s_\3}{s_\4}{s_{\2\1}}{s_{\3\2}}\,=\,
(Z^{s_{\2\1}}_{s_\2s_\1})^{-1}Z^{s_\4}_{s_\3s_{\2\1}}\,
\Fus{s_\3}{s_\4}{s_\1}{s_\2}{s_{\2\1}}{s_{\3\2}}\,
(Z^{s_{\2}}_{s_{\3\2}s_\3})^{-1}Z^{s_{\3\2}}_{s_\4s_\1}\,.
\end{equation}

Note, on the other hand, that 
the coefficients $\big\{\,{}^{\al_1}_{\al_3}\,{}^{\al_2}_{\al_4}\,{}^{\al_s}_{\al_t}\big\}_b^{}$ satisfy the tetrahedral symmetries
\begin{equation}
\big\{\,{}^{\al_1}_{\al_3}\,{}^{\al_2}_{\al_4}\,{}^{\al_s}_{\al_t}\big\}_b^{}
=\big\{\,{}^{\al_2}_{\al_4}\,{}^{\al_1}_{\al_3}\,{}^{\al_s}_{\al_t}\big\}_b^{}
=\big\{\,{}^{\al_2}_{\al_4}\,{}^{\al_s}_{\al_t}\,{}^{\al_1}_{\al_3}\big\}_b^{}
=\big\{\,{}_{\al_1}^{\al_3}\,{}_{\al_2}^{\al_4}\,{}_{\al_t}^{\al_s}\big\}_b^{}\,,
\end{equation}
as follows easily from the integral representation \rf{6j3}.
From the comparison 
it is easy to see that we must have  $Z_{s_\2s_\1}^{s_\3}=1$,
as claimed.
\newpage
\appendix

\section{Calculation of the Clebsch-Gordan coefficients} \label{CGapp}
\setcounter{equation}{0}
\subsection{Proof of Proposition \ref{CGpropn:T}}\label{proofCGpropn}

As the most tedious step let us calculate the 
matrix elements of the operator 
$\ST_{\1\2}^{}\,\SC^{-1}_\1$, which is defined as
\begin{equation}\label{mateldef}
C\big(\,{}^{s_{\2\1}}_{p_{\2\1}}\,|\,{}^{s_\2}_{p_\2}\,{}^{s_\1}_{p_\1}\,\big):=
\langle \,p_\2,p_\1\,|\,\ST_{\1\2}^{}\,\SC^{-1}_\1\,|\,p_{\2\1},s_{\2\1}\,\rangle\,.
\end{equation}
\begin{propn}\label{CGpropn} The matrix elements of  
$\ST_{\1\2}^{}\,\SC^{-1}_\1$ are 
explicitly given by the formula
\begin{align}\label{bCGnew}
C\big(\,{}^{s_{\2\1}}_{p_{\2\1}}\,|\,{}^{s_\2}_{p_\2}\,{}^{s_\1}_{p_\1}\,\big)=&\,
\de(p_{\2\1}-p_\2-p_\1)\,
e^{\frac{\pi i}{2}(\De_{s_\1}+\De_{s_\2}-\De_{s_{\2\1}})} 
w_b(s_\1+s_\2-s_{\2\1})w_b(s_{\2\1}+s_\2-s_\1)\notag\\
&\times e^{\frac{\pi i}{2}(p_3^2-p_\1^2-p_\2^2)}\,
\frac{w_b(p_\1-s_\1)w_b(p_\2-s_\2)}{w_b(p_{\2\1}-s_{\2\1})}
e^{\pi i(p_\2(s_\1+c_b)-p_\1(s_\2+c_b))}\notag\\
& \times\int_{\BR}dp\;e^{\pi ip(s_\1+s_\2-s_{\2\1}+c_b)} 
D_{\frac{1}{2}(s_\2-s_\1-s_{\2\1}-c_b)}(p+p_\2)
D_{\frac{1}{2}(s_\1-s_\2-s_{\2\1}-c_b)}(p-p_\1)
\notag
\\
& \hspace{4cm} \times D_{\frac{1}{2}(s_\1+s_\2+s_{\2\1}-c_b)}^{}(p)\,.
\end{align}
\end{propn}

\begin{proof}
By using 
\begin{equation}
e_b(\sq_\1+\spp_\2-\sq_\2)e^{-2\pi i\spp_\1\sq_\2}\,=\,
e^{-2\pi i\spp_\1\sq_\2}e_b(\sq_\1+\spp_\2-\spp_\1)\,,
\end{equation}
it is easy to see that 
\begin{equation}\label{C-matel}
C\big(\,{}^{s_{\2\1}}_{p_{\2\1}}\,|\,{}^{s_\2}_{p_\2}\,{}^{s_\1}_{p_\1}\,\big)=
\de(p_{\2\1}-p_\2-p_\1)
C^{s_{\2\1}}_{s_\2s_\1}(p_\1,p_{\2\1})\,,
\end{equation}
where
\begin{equation}
C^{s_{\2\1}}_{s_\2s_\1}(p_\1,p_{\2\1})=
\langle \,p_\1\,|\,e_b(\sq_\1-\spp_\1+p_{\2\1})\,
\SC^{-1}_\1\,|\,s_{\2\1}\,\rangle\,.
\end{equation}
As a preparation it will be convenient to  rewrite $C_\1^{-1}$ using \rf{eb-wb}
in the form
\begin{equation}
\SC^{-1}_\1:=e_b(\sq_\1-s_\2)e^{2\pi i s_\2\sq_\1}
\frac{w_b(s_\1+\spp_\1+s_\2)}{w_b(s_\1-\spp_\1-s_\2)}
\,.
\end{equation}
The matrix element \rf{C-matel} may then 
be calculated by inserting two resolutions of the identity as follows,
\begin{align}\label{Cintrep0}
C^{s_{\2\1}}_{s_\2s_\1}(p_\1,p_{\2\1})=
\int_{\BR^2}dp'dp''\; & 
\langle \,p_\1\,|\,e_b(\sq_\1-\spp_\1+p_{\2\1})\,|p'\,\rangle\times\\
\times & \langle\,p'\,|\,e_b(\sq_\1-s_\2)e^{2\pi i s_\2\sq_\1}\,|\,p''\,\rangle
\frac{w_b(s_\1+p''+s_\2)}{w_b(s_\1-p''-s_\2)}\langle\,p''\,
|\,s_{\2\1}\,\rangle\,.
\notag\end{align}
The ingredients of the kernel are the following:\\[1ex]
{\it 1. The matrix element $\langle p_\1|e_b(\sq_\1-\spp_\1+p_{\2\1})|p'\rangle$}:
\\[1ex]
We may use the integral identity 
\begin{equation}\label{FTeb1}
e_b(x)\,=\,\zeta_{\0}\int_{\BR-i0}dy\; e^{-2\pi i xy}e^{-\pi i y^2}
e_b(y+c_b)\,,
\end{equation}
where $\zeta_\0=e^{\frac{\pi i}{12}(1-4c_b^2)}$, 
in order to represent the function $e_b$ in this matrix element.
We get 
\begin{align}
\langle p_\1|e_b(\sq_\1-\spp_\1+p_{\2\1})|p'\rangle &\,=\,
\zeta_{\0}\int_{\BR-i0}dy\; e^{-\pi i y^2}
e_b(y+c_b)\,\langle p_\1|e^{-2\pi i(\sq_\1-\spp_\1+p_{\2\1})y}|p'\rangle\notag\\
&\,=\,\zeta_{\0}\int_{\BR-i0}dy\; e^{-\pi i y^2}
e_b(y+c_b)\,e^{\pi i y(p_\1+p'-2p_{\2\1})}\de(p_\1+y-p')\notag\\
&\,=\,\zeta_{\0}\,e^{2\pi i(p'-p_\1)(p_\1-p_{\2\1})}e_b(p'-p_\1+c_b)\,.
\label{matel1}\end{align}
{\it 2. The matrix element 
$\langle p'|e_b(\sq_\1-s_\2)e^{2\pi i s_\2\sq_\1}|p''\rangle$}:\\[1ex]
We now use a variant of the integral identity \rf{FTeb1} which takes the form
\begin{equation}\label{FTeb2}
e_b(x)\,=\,\zeta_{\0}^{-1}
\int_{\BR-i0}dy\; e^{-2\pi i xy}\frac{e^{2\pi i c_b y}}
{e_b(-y-c_b)}\,.
\end{equation}
A calculation similar to the one leading to \rf{matel1} gives now
\begin{align}
\langle\,p'\,|\,e_b(\sq_\1-s_\2)e^{2\pi i s_\2\sq_\1}\,|\,p''\,\rangle
&\,=\,
\zeta_{\0}^{-1}\int_{\BR-i0}dy\; \frac{e^{2\pi i (s_\2+c_b) y}}
{e_b(-y-c_b)}\,\langle \,p'+y-s_\2\,|\,p''\,\rangle\notag \\
& \,=\,\zeta_{\0}^{-1}
\frac{e^{2\pi i(s_\2+c_b)(s_\2+p''-p')}}{e_b(p'-s_\2-p''-c_b)}\,.
\label{matel2}\end{align}
{\it 3. The integral over $p'$}:\\[1ex]
Let us focus on the integral over $p'$ appearing in 
\rf{Cintrep0}:
\begin{equation}
\CI':=\int_{\BR}dp'\; 
\langle \,p_\1\,|\,e_b(\sq_\1-\spp_\1+p_{\2\1})\,|p'\,\rangle
\langle\,p'\,|\,e_b(\sq_\1-s_\2)e^{2\pi i s_\2\sq_\1}\,|\,p''\,\rangle\,.
\end{equation}
Inserting \rf{matel1} and \rf{matel2} yields
\begin{align}
\CI'\,=&\,e^{2\pi i(s_\2+c_b)(s_\2+p'')}  
e^{2\pi ip_\1(p_{\2\1}-p_{1})}\times\notag\\
&\times\int_{\BR+i0}
dp'\;e^{-2\pi ip'(p_{\2\1}+s_\2+c_b-p_\1)}
\frac{e_b(p'-p_\1+c_b)}{e_b(p'-s_\2-p''-c_b)}\,.
\end{align}
By using 
\begin{equation}
\int_{\BR}dz\;e^{-2\pi i z(u+c_b)}\frac{e_b(z+c_b)}{e_b(z-x-c_b)}\,=\,
\zeta_\0^{-1}\frac{e_b(u-x)}{e_b(-x-c_b)e_b(u)}\,,
\end{equation}
we may calculate 
\begin{align}\label{CI'}
\CI'\,=&\,\zeta_\0^{-1} 
\frac{e^{2\pi i(s_\2+c_b)(s_\2+p''-p_\1)} \,e_b(p_{\2\1}-p'')}{e_b(p_\1-s_\2-p''-c_b)e_b(p_{\2\1}+s_\2-p_\1)}\,.
\end{align}
It is now convenient to rewrite the resulting expression in terms
of the function $w_b(x)$ related to $e_b(x)$ via \rf{eb-wb}.
We find using $p_{\2\1}=p_\2+p_\1$
\begin{align}
\CI'=\,e^{\pi i(s_\2+c_b)(s_\2-p_\1)} 
e^{-\pi ip_\2s_\2}\,e^{\frac{\pi i}{2}(p_3^2-p_\1^2-p_\2^2)}
\frac{w_b(p_{2}+s_\2)w_b(p''-p_{\2\1})}{w_b(p''+s_\2-p_\1+c_b)}e^{\pi ip''(s_\2-p_\2+c_b)}.
\end{align}
Taking into account that 
\begin{equation}
\langle\,p''\,
|\,s_{\2\1}\,\rangle\,=\,\frac{w_b(s_{\2\1}-p''-c_b)}{w_b(s_{\2\1}+p''+c_b)}\,,
\end{equation}
we may use \rf{CI'} to get a single integral representation for
$C^{s_{\2\1}}_{s_\2s_\1}(p_\1,p_{\2\1})$, 
which takes the form
\begin{equation}
C^{s_{\2\1}}_{s_\2s_\1}(p_\1,p_{\2\1})=e^{\frac{\pi i}{2}(p_3^2-p_\1^2-p_\2^2)}
e^{\pi i(s_\2+c_b)(s_\2-p_\1)} 
e^{-\pi ip_\2s_\2}\,w_b(p_{2}+s_\2)\,\CI''\,,
\end{equation}
where the integral $\CI''$ is defined as
\begin{equation}
\CI'':=\int\limits_{\BR-i0}dp\;
\frac{w_b(p-p_{\2\1})}{w_b(p+s_\2-p_\1+c_b)}
\frac{w_b(p+s_\1+s_\2)}{w_b(p+s_{\2\1}+c_b)}
\frac{w_b(p+s_\2-s_\1)}{w_b(p-s_{\2\1}+c_b)}e^{\pi ip(s_\2-p_\2+c_b)}\,.
\end{equation}
We'll need to rewrite this integral further. Let us first introduce
the combination
\begin{equation}
D_a(x):=\,\frac{w_b(x+a)}{w_b(x-a)}\,.
\end{equation}
In terms of this function we may write $\CI''$ as
\begin{align}
\CI'':=\int_{\BR}dp\;e^{\pi ip(s_\2-p_\2+c_b)}\;
& D_{-\frac{1}{2}(s_\2+p_{2}+c_b)}^{}\big(p+\fr{1}{2}(s_\2-p_\1-p_{\2\1}+c_b)\big)
\\[-.75ex]
 \times &
D_{\frac{1}{2}(s_\1+s_\2-s_{\2\1}-c_b)}\big(p+\fr{1}{2}(s_\1+s_\2+s_{\2\1}+c_b)\big)
\notag\\
 \times &
D_{\frac{1}{2}(s_\2-s_\1+s_{\2\1}-c_b)}\big(p+\fr{1}{2}(s_\2-s_\1-s_{\2\1}+c_b)\big)
\notag\,.
\end{align}
By using the identities \cite[Equation (A.34)]{BT}
\begin{align}\label{startriang}
&\int dx\;D_{\al}(x+u)D_{\be}(x+v)D_{\ga}(x+w)e^{-2\pi i x\de}=\\
&=A_{\al\be\ga\de}D_{\al+\be+c_b}(u-v)e^{-2\pi i(\al+\be+c_b)w}
\int dx\;e^{-2\pi i x\ga^*}D_{\al^*}(x+v)D_{\be^*}(x+u)D_{\de^*}(x+w)\,,
\notag\end{align}
and 
\begin{equation}
D_{a}(x)D_{b}(y)\,=\,
D^{}_{\frac{1}{2}(a+b+x-y)}\big(\fr{1}{2}(a-b+x+y)\big)
D^{}_{\frac{1}{2}(a+b+y-x)}\big(\fr{1}{2}(b-a+x+y)\big)\,,
\end{equation}
we find that 
$C\big({}^{s_{\2\1}}_{p_{\2\1}}|{}^{s_\2}_{p_\2}{}^{s_\1}_{p_\1}\big)$ is
indeed represented by the formula \rf{bCGnew}, a claimed.
\end{proof}

With the help of the integral identity \rf{startriang}
it is straightforward to check that the expression given 
in \rf{b-3j} has the Weyl-symmetries \rf{Weyl}.
The reality \rf{reality} follows immediately
since
\begin{equation}\label{reality'}
\big[\big(\,{}^{s_3}_{p_3}\,|\,{}^{s_\2}_{p_\2}\,{}^{s_\1}_{p_\1}\big)_b^{}\big]^*
=\big(\,{}^{-s_3}_{\;\;\,p_3}\,|\,
{}^{-s_\2}_{\;\;\,p_\2}\,{}^{-s_\1}_{\;\;\,p_\1}\big)_b^{}\,.
\end{equation}
Keeping in mind that
\begin{equation}
\big[\langle\,s_\3,p_\3\,|\,\SC_{s_\2s_\1}\,|\,p_\2,p_\1\,\rangle\big]^* =
\langle \,p_\2,p_\1\,|\,\ST_{\1\2}^{}\,\SC^{-1}_\1\,|\,s_{\3},p_{\3}\,\rangle\,,
\end{equation}
one may complete the proof of
Proposition \ref{CGpropn:T} by comparing the expressions
\rf{b-3j} and \rf{bCGnew}.

\subsection{Proof of Proposition \ref{Comppropn}}\label{PTvsNT}

As the main technical step
let us calculate the Fourier-transformation of the b-Clebsch-Gordan
coefficients 
$\big(\,{}^{s_3}_{x_3}\,|\,{}^{s_\2}_{x_\2}\,{}^{s_\1}_{x_\1}\big)_b^{}$ defined 
in \rf{b-3ja}. We need to calculate the following integral:
\begin{equation}\label{k-x}
\big(\,{}^{s_3}_{k_3}\,|\,{}^{s_\2}_{k_\2}\,{}^{s_\1}_{k_\1}\big)_b^{\rm\sst PT}=
\int_{\BR}dx_3\;e^{-2\pi i k_3x_3}
\int_{\BR^2}dx_\2dx_\1\;e^{2\pi i (k_\2x_\2+k_\1x_\1)}
\big(\,{}^{s_3}_{x_3}\,|\,{}^{s_\2}_{x_\2}\,{}^{s_\1}_{x_\1}\big)_b^{}\,.
\end{equation}
\begin{propn} We have
\begin{align}\label{bCGold}
 \big(\,{}^{s_3}_{k_3}\,|\,{}^{s_\2}_{k_\2}\,{}^{s_\1}_{k_\1}\big)_b^{\rm\sst PT}= 
& \,\de(k_3-k_\2-k_\1)\,
\,e^{\pi i(k_\2(s_\1+c_b)-k_\1(s_\2+c_b))}\,\\
&\times N(s_\3,s_\2,s_\1) 
\,w_b(-s_\1-s_\2-s_3)w_b(s_\1+s_3-s_\2)w_b(s_\2+s_3-s_\1)\,\notag\\
&\times\int_\BR dy \;e^{-\pi i(s_\3-s_\2-s_\1-c_b)y}
D_{\frac{1}{2} (s_\1+s_\2+s_3-c_b)}(y)\notag\\
&\hspace{4cm}\times
D_{\frac{1}{2}(s_\2-s_3-s_\1-c_b)}^{}(y+k_\2)
D_{\frac{1}{2}(s_\1-s_3-s_\2-c_b)}^{}(y-k_\1)\,. 
\notag
\end{align}
\end{propn}

\begin{proof}
After using the integral transformation
\begin{equation}\label{D-FT}
D_a(x)=w_b(2a+c_b)\int_\BR dy\;e^{-2\pi ixy}D_{-a-c_b}(y)\,
\end{equation}
in order to express the function $D_a(x)$ which appears in the first
line of \rf{b-3ja}, we get the integral
\begin{align}
& \big(\,{}^{s_3}_{k_3}\,|\,{}^{s_\2}_{k_\2}\,{}^{s_\1}_{k_\1}\big)_b^{\rm\sst PT}= 
N(s_\3,s_\2,s_\1)\,w_b(-s_\1-s_\2-s_3)\\
& \times\!
\int_\BR dy \;e^{\pi i(s_3+c_b)y}D_{\frac{1}{2} (s_\1+s_\2+s_3-c_b)}(y)
\int_{\BR}dx_3\;e^{-2\pi i k_3x_3}\!
\int_{\BR^2}\!dx_\2dx_\1\;e^{2\pi i (k_\2x_\2+k_\1x_\1)}\,
e^{-2\pi i(x_\2-x_\1)y}\notag \\ 
& \hspace{3cm}\times D_{-\frac{1}{2}(s_\2-s_3-s_\1+c_b)}^{\rm\sst PT}
\big(x_\2-x_3- \fr{s_\1+c_b}{2}\big)
D_{-\frac{1}{2}(s_\1-s_3-s_\2+c_b)}^{}
\big(x_3-x_\1- \fr{s_\2+c_b}{2}\big)\,. 
\notag\end{align}
Substituting the variables of integration as
$x_\2=y_\2+x_3+(s_\1+c_b)/2$, $x_\1=y_\1+x_3-(s_\2+c_b)/2$ 
yields
\begin{align}
& \big(\,{}^{s_3}_{k_3}\,|\,{}^{s_\2}_{k_\2}\,{}^{s_\1}_{k_\1}\big)_b^{\rm\sst PT}= 
N(s_\3,s_\2,s_\1)\,w_b(-s_\1-s_\2-s_3)\,e^{\pi i(k_\2(s_\1+c_b)-k_\1(s_\2+c_b))}\\
& \qquad\qquad\times\int_\BR dy \;e^{\pi i(s_3-s_\2-s_\1-c_b)y}
D_{\frac{1}{2} (s_\1+s_\2+s_3-c_b)}(y)
\int_{\BR}dx_3\;
e^{-2\pi i (k_3-k_\2-k_\1)x_3}\notag\\
&\qquad\qquad\times \int_{\BR}dy_\2 \;e^{-2\pi iy_\2(y-k_\2)}\,
D_{-\frac{1}{2}(s_\2-s_3-s_\1+c_b)}^{}(y_\2)
\int_{\BR}dy_\1\;e^{2\pi i y_\1(y+k_\1)}\,
D_{-\frac{1}{2}(s_\1-s_3-s_\2+c_b)}^{}(y_\1). 
\notag
\end{align}
The integrals over $y_\2$ and $y_\1$ may be carried out using \rf{D-FT},
while the integral over $x_3$ yields a delta-distribution $\de(k_3-k_\2-k_\1)$.
We arrive at the formula
\begin{align}\label{bCGold'}
 \big(\,{}^{s_3}_{k_3}\,|\,{}^{s_\2}_{k_\2}\,{}^{s_\1}_{k_\1}\big)_b^{\rm\sst PT}= 
& \,\de(k_3-k_\2-k_\1)\,
N(s_\3,s_\2,s_\1)\,e^{\pi i(k_\2(s_\1+c_b)-k_\1(s_\2+c_b))}\,\\
&\times w_b(-s_\1-s_\2-s_3)w_b(s_\1+s_3-s_\2)w_b(s_\2+s_3-s_\1)\,\notag\\
&\times\int_\BR dy \;e^{\pi i(s_3-s_\2-s_\1-c_b)y}
D_{\frac{1}{2} (s_\1+s_\2+s_3-c_b)}(y)\notag\\
&\hspace{4cm}\times
D_{\frac{1}{2}(s_\2-s_3-s_\1-c_b)}^{}(y-k_\2)
D_{\frac{1}{2}(s_\1-s_3-s_\2-c_b)}^{}(y+k_\1)\,. 
\notag
\end{align}
Substituting $y\ra -y$ and using that $D_a(x)=D_{a}(-x)$
completes the proof.
\end{proof}
It remains to compare the resulting expression 
\rf{bCGold} with \rf{b-3j}. We find that
\begin{align}\label{3jcomp}
\big(\,{}^{s_{\3}}_{p_{\3}}\,|\,{}^{s_\2}_{p_\2}\,{}^{s_\1}_{p_\1}\,\big)
& \,=\,N(s_\3,s_\2,s_\1)
\frac{w_b(p_\1-s_\1)w_b(p_\2-s_\2)}{w_b(p_{\3}-s_{\3})}
e^{\frac{\pi i}{2}(p_3^2-p_\1^2-p_\2^2)}
\big(\,{}^{s_\3}_{p_\3}\,|\,
{}^{s_\2}_{p_\2}\,{}^{s_\1}_{p_\1}\big)_b^{\rm\sst PT}\,.
\end{align}
We observe that the terms in the first line of \rf{3jcomp} represent
the unitary 
transformation between the representation \rf{Paldef} and the Whittaker
model \rf{Whitt}. The prefactor in the second line depends only
on the triple of Casimir eigenvalues and represents a change of
normalization of the Clebsch-Gordan maps.

\section{Proofs of some technical results}\label{Proofapp}
\setcounter{equation}{0}

\subsection{Proof of Proposition \ref{W-propn}}

The move $W_1$ defined in \rf{Wmovedef}
may be factorized into the following 
three simple moves:

{\bf First move:} The move 
$\rho_{\check\1}^{}\circ\omega_{\hat\1\check\1}^{}$, diagrammatically 
represented as follows
\[
\lower.9cm\hbox{\epsfig{figure=w-op1.eps,height=2.8cm}}
\quad 
\overset{\rho_{\check\1}^{}\circ\,\omega_{\hat\1\check\1}^{}}{\longrightarrow}
\quad
\lower.9cm\hbox{\epsfig{figure=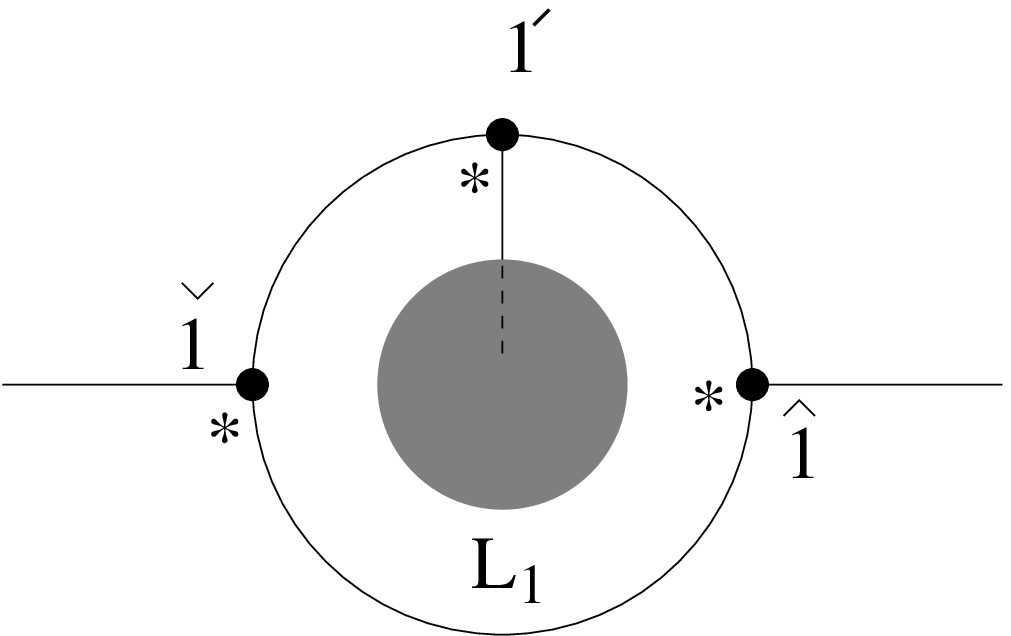,height=2.8cm}}
\]
{\bf Second move:} The move 
$\omega_{\1'\hat\1}^{}$, diagrammatically 
represented as follows
\[
\lower.9cm\hbox{\epsfig{figure=w-op1a.eps,height=2.8cm}}
\quad 
\overset{\omega_{\1'\hat\1}^{}}{\longrightarrow}
\quad
\lower.9cm\hbox{\epsfig{figure=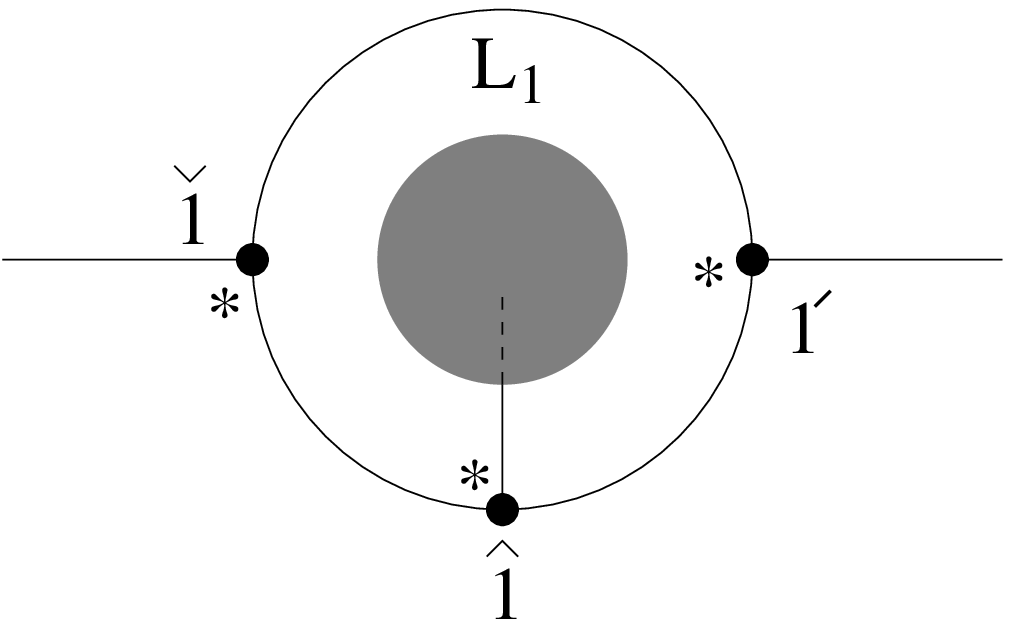,height=2.8cm}}
\]
{\bf Third move:} The move 
$\rho_{\check\1}^{}\circ\omega_{\check\1\1'}^{}$, diagrammatically 
represented as follows
\[
\lower.9cm\hbox{\epsfig{figure=w-op2a.eps,height=2.8cm}}
\quad 
\overset{\rho_{\check\1}^{}\circ\,\omega_{\check\1\1'}^{}}{\longrightarrow}
\quad
\lower.9cm\hbox{\epsfig{figure=w-op2.eps,height=2.8cm}}
\]

We have 
\begin{equation}
\sa\sd[\SW_\1](e^{-\pi b \spp_{\check{1}}})\,=\,e^{-\pi b(-\spp_{\check{1}}+2\sz_\1)}\,,
\end{equation}
where $\sz_\1:=\frac{1}{2}(\spp_0-\sq_0+\spp_{\hat{1}})$. This is equivalent
to $\sa\sd[\SW_\1](\SK_\1^{})=\SK_\1^{-1}$.

Furthermore
\begin{equation}
\sa\sd[\SW_\1](e^{\pi b(2\sq_{\check{\1}}-\spp_{\check{\1}})})
=e^{-\frac{\pi b}{2}(2\sq_{\check{\1}}-\spp_{\check{\1}})}
\big(2\cosh2\pi b(\spp_{\check{\1}}-\sz_\1)+\SL_\1\big)
e^{-\frac{\pi b}{2}(2\sq_{\check{\1}}-\spp_{\check{\1}})}
\end{equation}
This is equivalent
to $\sa\sd[\SW_\1](\SE_\1^{})=\SF_\1$.

\subsection{Proof of Proposition \ref{Lreppropn}}
We need to calculate
$\sa\sd[\SU_{\2\1}](\SL)$, 
where 
\begin{equation}
\SL:=2\cosh\pi b (\spp_b+\sq_a-\spp_a)+
e^{\pi b(2\sq_b-\spp_b-(\sq_a+\spp_a))}\,,
\end{equation}
and $\SU_{\2\1}$ is the operator representing the
move $U_{\2\1}$ which is diagrammatically represented as:
\[
\lower1cm\hbox{\epsfig{figure=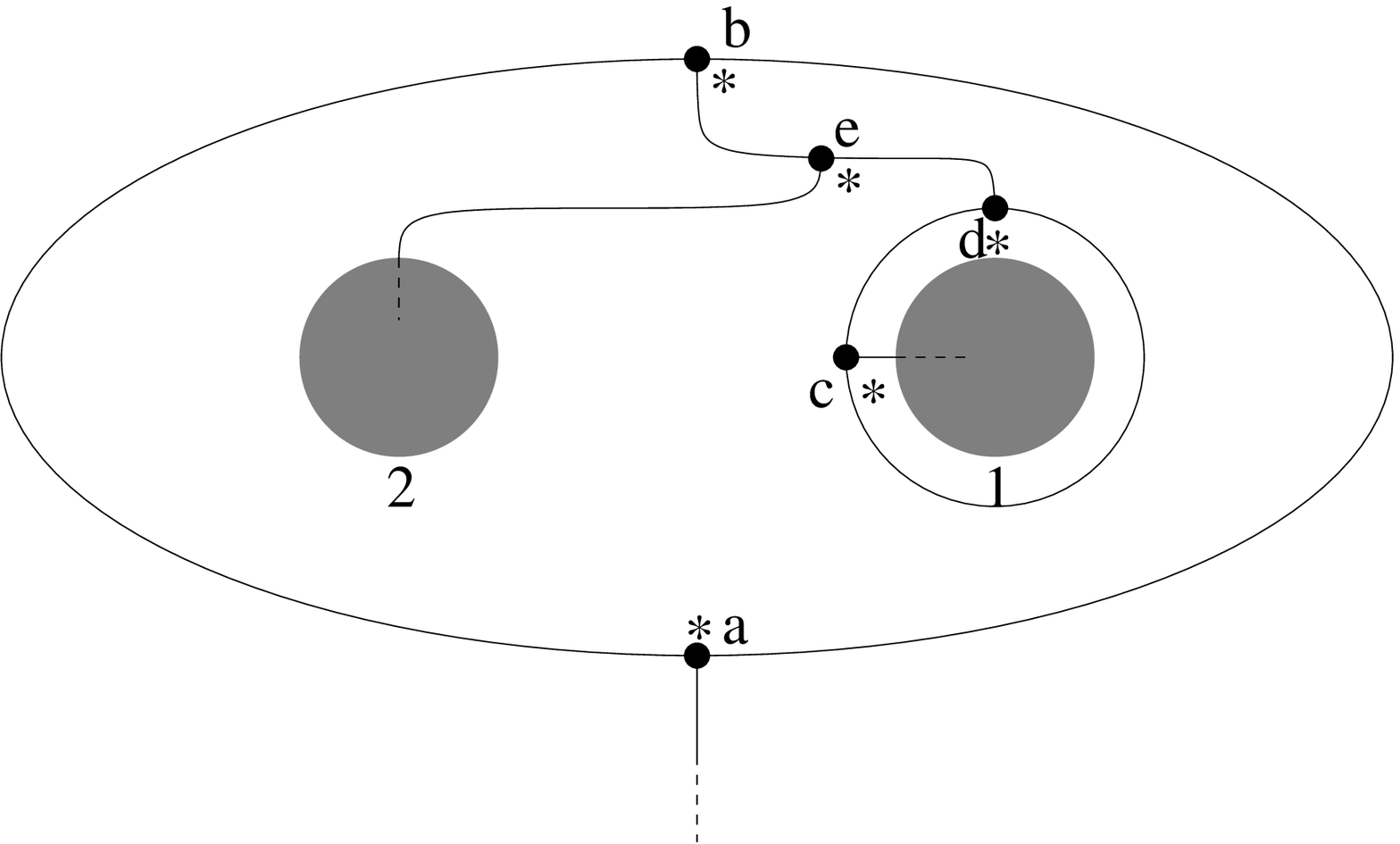,height=3cm}}
\quad \overset{U_{\2\1}}{\longrightarrow}\quad
\lower1cm\hbox{\epsfig{figure=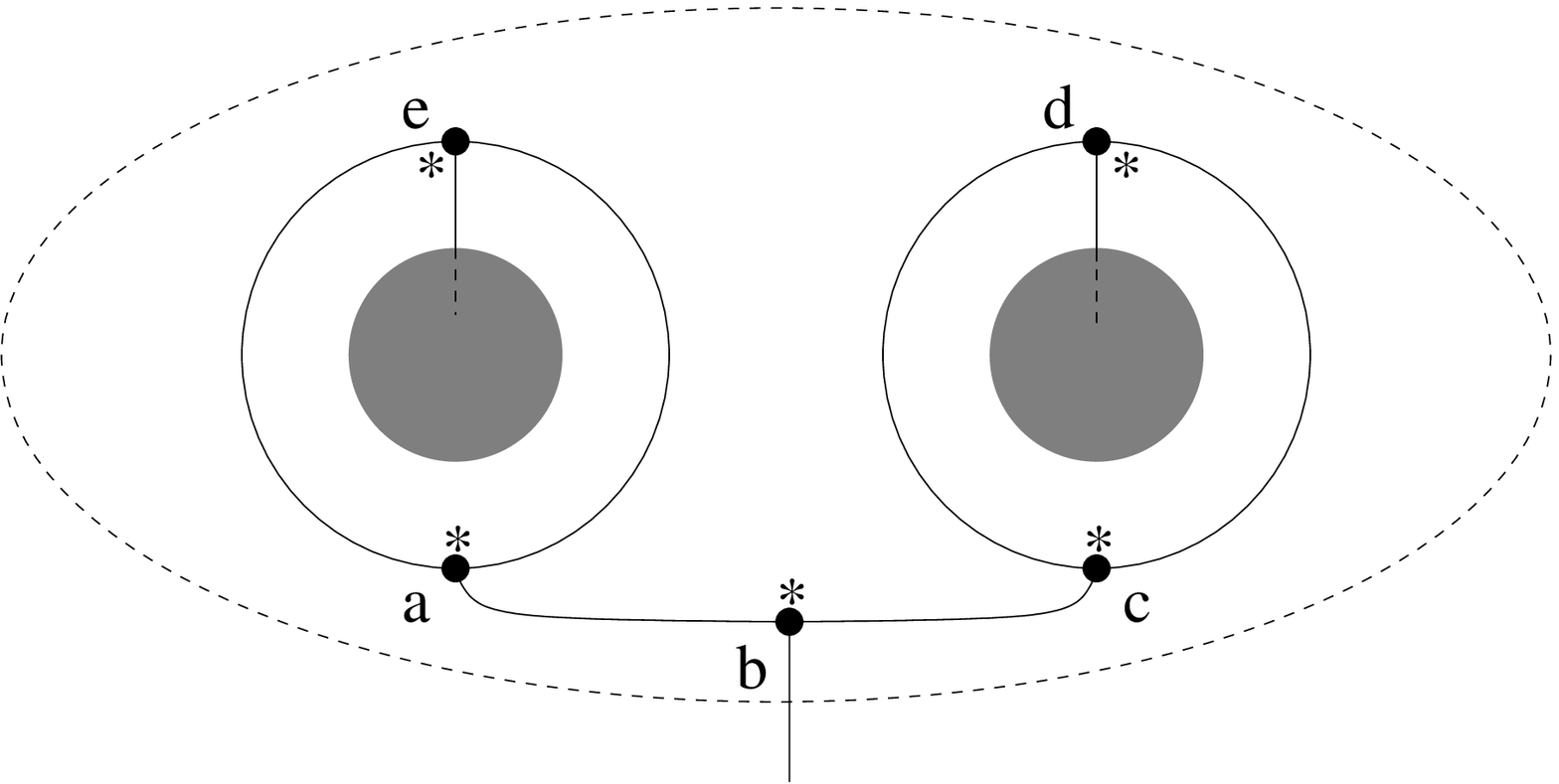,height=3cm}}
\] 
The calculation may be performed in three steps.

{\bf First step:} The move $\rho_{e}^{-1}\circ\omega_{eb}$, diagrammatically 
represented as follows
\[
\lower1cm\hbox{\epsfig{figure=leng4b.eps,height=3cm}}
\quad \overset{\rho_{e}^{-1}\circ\omega_{eb}}{\longrightarrow}\quad
\lower1cm\hbox{\epsfig{figure=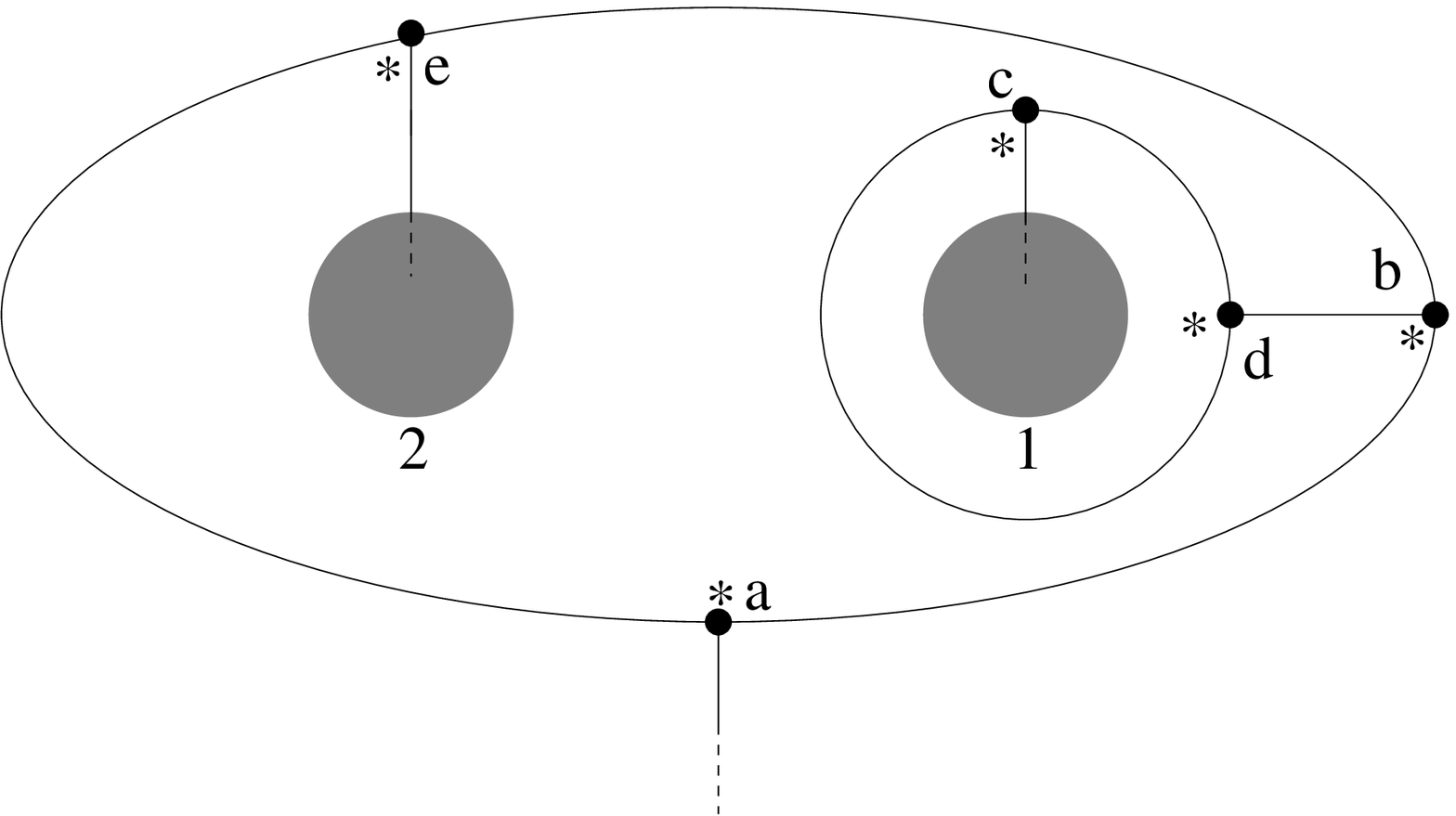,height=3cm}}
\] 
Calculation of $\sa\sd[\SA_{e}^{-1}\ST_{eb}](\SL)$:
\begin{align} 
\SL':=\sa\sd[\SA_{e}^{-1}\ST_{eb}](\SL)\,=\,&
2\cosh\pi b (\spp_b-\sq_e+\sq_a-\spp_a)+e^{\pi b(\spp_b+(2\spp_e-\sq_e)-(\sq_a+\spp_a))}\notag\\
&+
e^{\pi b(2\sq_b+\sq_e-\spp_b-(\sq_a+\spp_a))} \,.
\end{align}
{\bf Second step:}  The move $W_1$, diagrammatically 
represented as follows
\[
\lower1cm\hbox{\epsfig{figure=leng3.eps,height=3cm}}
\quad \overset{W_1}{\longrightarrow}\quad
\lower1cm\hbox{\epsfig{figure=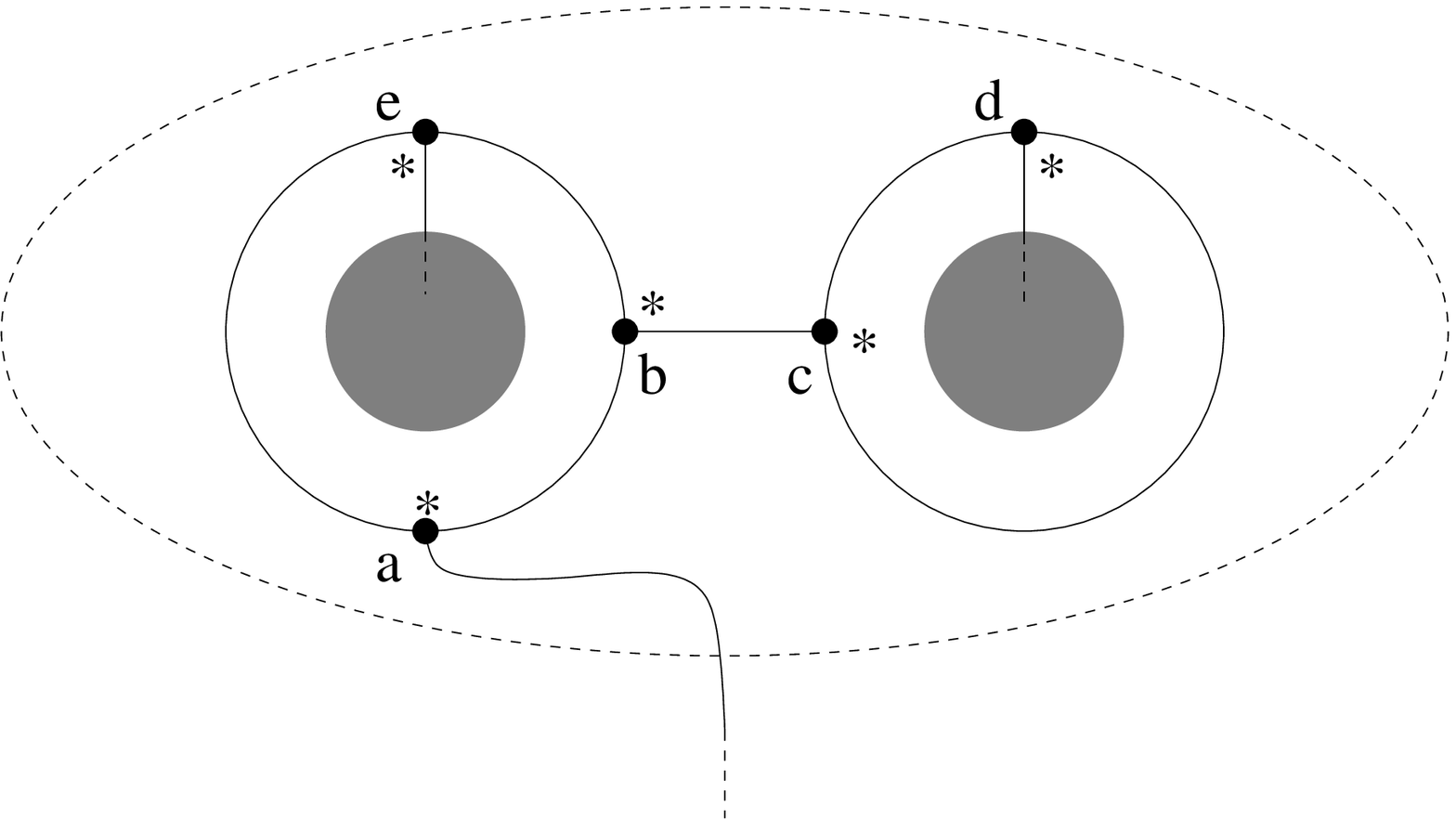,height=3cm}}
\] 
Calculation of $\sa\sd[\SW_\1](\SL')$:
\begin{align}
\SL'':=& \,\sa\sd[\SW_\1](\SL')\\
= & \,2\cosh\pi b (-\spp_b+2\sz_\1-\sq_e+\sq_a-\spp_a)+e^{\pi b(-\spp_b+\sz_\1
+(2\spp_e-\sq_e)-(\sq_a+\spp_a))}\notag\\
&+:e^{-{\pi b}(2\sq_b+\sq_e-\spp_b+(\sq_a+\spp_a))}\big(2\cosh2\pi b(\spp_b-\sz_\1)+\SL_\1\big):,
\notag\end{align}
where $\sz_\1:=\frac{1}{2}(\spp_c+\spp_d-\sq_c)$.

{\bf Third step:}  The move $\omega_{ba}$, diagrammatically 
represented as follows
\[
\lower1cm\hbox{\epsfig{figure=leng2.eps,height=3cm}}
\quad \overset{\omega_{ba}}{\longrightarrow}\quad
\lower1cm\hbox{\epsfig{figure=leng1.eps,height=3cm}}
\] 
Calculation of $\sa\sd[\ST_{ba}^{}](\SL'')$:
We factorize $\ST_{ba}^{}=\ST_{ba}'e^{-2\pi i \spp_b\sq_a}$
and collect the terms with equal weight with respect to 
the adjoint action of the 
argument $\sq_b-\sq_a+\spp_a$ of $\ST'_{ba}=e_b(\sq_b-\sq_a+\spp_a)$:
\begin{align*}
\SL''':=\sa\sd[\ST_{ba}']\Big( 
& :e^{2\pi b(\spp_b-\sz_\1)}\big(
e^{-\pi b(2\sq_b+\sq_e-\sq_a+\spp_a)}+e^{-\pi b(\sq_a-\spp_a)}\big): \\
& +e^{\pi b(-\spp_b+\sz_\1+(2\spp_e-\sq_e)-(\sq_a+\spp_a))}+ 
e^{-\pi b(2\sq_b+\sq_e+\spp_a-\sq_a)}\SL_\1\\
& \;+:e^{-2\pi b(\spp_b-\sz_\1)}\big(e^{-\pi b(2\sq_b+\sq_e+\spp_a-\sq_a)}+
e^{\pi b(\sq_a-\spp_a)}\big): \Big)\\
= & :e^{2\pi b(\spp_b-\sz_\1)}\big(
e^{-\pi b(2\sq_b+\sq_e-\sq_a+\spp_a)}+e^{-\pi b(\sq_a-\spp_a)}\big)
\big(1+e^{2\pi b(\sq_b-\sq_a+\spp_a)}\big)^{-1}: \\
& + e^{\pi b(-\spp_b+\sz_\1
+(2\spp_e-\sq_e)-(\sq_a+\spp_a))}+ 
e^{-\pi b(2\sq_b+\sq_e+\spp_a-\sq_a)}\SL_\1 \\
& \;+ :e^{-2\pi b(\spp_b-\sz_\1)}\big(e^{-\pi b(2\sq_b+\sq_e+\spp_a-\sq_a)}+
e^{\pi b(\sq_a-\spp_a)}\big)\big(1+e^{2\pi b(\sq_b-\sq_a+\spp_a)}\big):\,,
\end{align*}
where $\sz_\2:=\frac{1}{2}(\spp_a-\sq_a+\sq_e)$.
Collecting the terms yields
\begin{align}
\SL'''= & e^{-2\pi b(\spp_b-\sz_\1+\sq_b+\sz_\2)}
+e^{-2\pi b(\spp_b-\sz_\1)}\SL_\2+
e^{-2\pi b(\sq_b+\sz_\2)}\SL_\1\notag\\
& +2\cosh2\pi b(\spp_b-\sq_b-\sz_\1-\sz_\2)\,.
\end{align}

\end{document}